\documentclass[ejs,preprint]{imsart}

\RequirePackage[OT1]{fontenc}
\RequirePackage{amsthm,amsmath,amssymb,amsfonts}
\RequirePackage[numbers]{natbib}
\usepackage{pictexwd}
\usepackage{graphicx,url}
\usepackage[normalem]{ulem}
\usepackage[colorlinks,citecolor=blue,urlcolor=blue]{hyperref}
\usepackage{latexsym}
\usepackage{caption}
\usepackage{subcaption}

\arxiv{arXiv:1701.07359}

\startlocaldefs

\def\argmin{\mbox{argmin}}

\def\G{{\mathbb G}}
\def\IK{{\mathbb K}}
\def\P{{\mathbb P}}
\def\R{\mathbb R}
\def\Z{{\mathbb Z}}

\def\a{\alpha}
\def\b{\beta}
\def\dd{\Delta}
\def\d{\delta}
\def\e{\varepsilon}
\def\ee{\epsilon}
\def\f{\phi}
\def\t{\tau}
\def\s{\sigma}

\newtheorem{lemma}{Lemma}[section]

\newtheorem{remark}{Remark}[section]

\numberwithin{equation}{section}
\theoremstyle{plain}
\setattribute{journal}{name}{}
\endlocaldefs

\begin{document}
\begin{frontmatter}
\title{The nonparametric bootstrap for the current status model}
\runtitle{Current status bootstrap}

\begin{aug}
\author{\fnms{Piet} \snm{Groeneboom}\corref{}\ead[label=e1]{P.Groeneboom@tudelft.nl}
\ead[label=u1,url]{http://dutiosc.twi.tudelft.nl/\textasciitilde pietg/}}

\address{Delft University of Technology, Mekelweg 4, 2628 CD Delft, The Netherlands.\\ \printead{e1}}

\author{\fnms{Kim} \snm{Hendrickx}\ead[label=e2]{kim.hendrickx@uhasselt.be}
\ead[label=u2,url]{http://www.uhasselt.be/fiche_en?voornaam=Kim&naam=HENDRICKX}}
\runauthor{P.\ Groeneboom and K.\ Hendrickx}
\address{Hasselt University, I-BioStat, Agoralaan, B­3590 Diepenbeek, Belgium. \\ 		\printead{e2}}
\end{aug}



\begin{abstract}
It has been proved that direct bootstrapping of the nonparametric maximum likelihood estimator (MLE) of the distribution function in the current status model leads to inconsistent confidence intervals. We show that bootstrapping of functionals of the MLE can however be used to produce valid intervals. To this end, we prove that the bootstrapped MLE converges at the right rate in the $L_p$-distance. We also discuss applications of this result to the current status regression model.
\end{abstract}

\end{frontmatter}

\section{Introduction}
\label{sec:intro}
In the current status model, the variable of interest is a survival variable $X$ with distribution function $F_0$. However, instead of observing the exact survival time $X$, a censoring variable $T\sim G$ is observed together with the indicator $\dd = 1_{X\le T}$. Such data arise naturally in clinical trials when  a patient can only be checked at one measurement due to destructive testing. A lot of research has been published on the behavior of the maximum likelihood estimator (MLE) $F_n$ of the distribution function $F_0$. The limiting distribution of $n^{1/3}(F_n(t)-F_0(t)$) is after scaling by the constant $\kappa = \{4F_0(t)(1-F_0(t))f_0(t)/g(t)\}^{1/3}$ given by
\begin{align*}
\mathbb C = \arg\max_t\left\{ W(t)-t^2 \right\},
\end{align*}
where $W$ is a two-sided  Brownian motion with $W(0)= 0$ (see \cite{GrWe:92}). Other estimators with similar asymptotic properties are Chernoff's estimator of the mode (\cite{chernoff:64}), the Grenander estimator (\cite{Grenander:56}) of a nonincreasing density, Manski's  maximum score estimator (\cite{manski1975}) and Rouseeuw's least median of squares estimator (\cite{rousseeuw1984}). A general framework for cube-root $n$ asymptotics is given in \cite{kimpol:90}.

In this paper we investigate the behavior of Efron's nonparametric bootstrap method (\cite{efron1979}) for constructing confidence intervals for smooth functionals of the MLE. It is known that the nonparametric bootstrap is inconsistent for generating the limit distribution of the MLE. The authors of \cite{abrevaya_huang2005} prove that (conditional on the data),
\begin{align*}
&n^{1/3}\{4F_0(t)(1-F_0(t))f_0(t)/g(t)\}^{-1/3}\{\hat F_n(t)-F_n(t) \} \\
&\qquad\qquad \stackrel{\cal D}{\to} \quad  \arg\max_t(W(t) + \hat W(t)-t^2) - \arg\max_t(W(t)-t^2),
\end{align*}
where $\hat F_n$ is the bootstrap MLE and $W$ and $\hat W$ are two independent two-sided Brownian motions originating at zero. A similar result is obtained in \cite{kosorok:08} and in \cite{sen_mouli_woodroofe:10} for the Grenander estimator. The maximum score estimator of \cite{manski1975} is another example of a cube-root $n$ statistic with asymptotic distribution derived in \cite{kimpol:90}, where inconsistency of the nonparametric bootstrap for this estimator is shown in \cite{abrevaya_huang2005}.  

Constructing asymptotic confidence intervals for the distribution function in the current status model based on Chernoff's distribution and the normalizing constant $\kappa$ is complicated by the need to compute the critical values of $\mathbb C$ and to estimate the density $f_0$ consistently. Since this turns out to be a rather difficult task several alternative bootstrap methods have been proposed based on resampling from a smooth estimate. \cite{SenXu2015} consider a smooth kernel estimate $\tilde F$ of $F_0$ and resample the $\dd_i$ from a Bernoulli distribution with probability $\tilde F(T_i)$, while keeping the censoring variables $T_i$ fixed and center the values of the bootstrap samples by subtracting the smooth estimate of the distribution function.
\cite{kosorok:08} and \cite{sen_mouli_woodroofe:10} propose similar smooth respampling schemes for the Grenander estimator and a model-based smoothed bootstrap procedure for making inference on the maximum score estimator is developed in  \cite{patra:2011}. All methods result in consistent estimation of the (suitably standardized) distribution $\mathbb C$ conditional on the original data. 

A drawback of this approach is that smoothness conditions of $F_0$ are used which allow faster than cube-root $n$ estimation of $F_0$. This raises the question if one should really use confidence intervals based on the MLE instead of on a faster converging estimate. 

This latter procedure is followed in \cite{kim_piet:17:SJS}, where the authors consider constructing confidence intervals around the smoothed maximum likelihood estimator (SMLE) of $F_0$ in the current status model. The SMLE is a kernel estimate based on the MLE with an asymptotic normal distribution, instead of Chernoff's limiting distribution (\cite{piet_geurt_birgit:10}). The bootstrap method proposed in \cite{kim_piet:17:SJS} is however still based on the smooth bootstrap procedure described in \cite{SenXu2015} and not on Efron's nonparametric bootstrap. We show in this paper that the construction of confidence intervals around the SMLE based on the nonparametric bootstrap can also be proved to be valid, where one does not resample from a smooth estimate of $F_0$, but just resamples with replacement from the pairs $(T_i,\dd_i)$ in the original sample. This method already has been used without proof in \cite{piet_geurt:14} and also in \cite{piet_geurt:15} and the present manuscript intends to fill the gap of the missing proofs here.
An important difference with the smooth bootstrap in \cite{kim_piet:17:SJS} is that for the centering of the estimates in the nonparametric bootstrap samples the SMLE of the original sample is used, whereas this will not work for the resampling as proposed in \cite{kim_piet:17:SJS}; in the latter case one needs to center the estimates in the bootstrap samples by a kernel convolution of the SMLE in the original sample. It is not clear which method is better, and the most striking fact is the similarity of the results of the two methods in our simulations. An advantage of the purely nonparametric bootstrap, discussed in the present paper, might be its conceptual simplicity and the absence of the need to center with a convolution of the SMLE in the centering of the bootstrap samples instead of the SMLE itself. An advantage of the smooth bootstrap, discussed in \cite{kim_piet:17:SJS} might be the fact that only the indicators $\dd_i$ are being resampled, and that in this sense one stays closest to the sample distribution of the observation times $T_i$, which stay fixed in this procedure.

Although it is argued in \cite{durot2010} that the naive bootstrap will not work for their goodness-of-fit test for monotone functions, based on the Grenander estimator, no theoretical justification for this conjecture is given. Other examples where a smooth bootstrap procedure is used, are the likelihood ratio type two-sample test for current status data proposed by \cite{piet12e} and the test for equality of functions under monotonicity constraints proposed by \cite{cecile_rik:12}. Both tests establish asymptotic normality for the test statistic considered.

The paper is organized as follows: In Section \ref{sec:model} we introduce the current status model and review some interesting properties of the MLE. The validity of the nonparametric bootstrap is discussed in Section \ref{sec:bootstrap}. In Section  \ref{sec:applications} we provide two examples to illustrate the applicability of our result. In the first example we construct pointwise confidence intervals based on the smoothed MLE in the current status model. The second example deals with doing inferences for a finite dimensional regression parameter in the current status linear regression model. For both examples, the theoretical and finite sample behavior of the nonparametric bootstrap is discussed. Section \ref{sec:discussion} presents some concluding remarks. The proofs of our results are given in Section \ref{sec:appendix}. 

\section{The current status model and the MLE}
\label{sec:model}
Let $Z_1 = (T_1,\dd_1),\ldots,Z_n= (T_n,\dd_n) $ be an i.i.d. sample from the probability space $([0,R]\times\{0,1\}, {\cal A}, P)$, where $\dd_i = 1_{X_i \le T_i}$ and $R>0$.  The $X_i$ are interpreted as (nonnegative) survival times with distribution function $F_0$. Instead of observing $X$, a censoring variable $T \sim G$ is observed (with density $g$) independent of $X$. One could say that in the current status model, each observation $Z_i$ represents the current status of the item $i$ at time $T_i$. The density of $Z_i$ with respect to the product of Lebesgue measure and counting measure on $[0,R] \times \{0,1\}$ is given by
\begin{align*}
p_{F_0}(t,\d) = \left[\d F_0(t)  + (1-\d)  \{1-F_0(t)\}  \right]g(t).
\end{align*}
The maximum likelihood estimator $F_n$ is defined as the maximizer of the log likelihood given by (up to a constant not depending on $F$),
\begin{align}
\label{log_likelihood}
\ell_n(F) = n^{-1}\sum_{i=1}^{n}\left[\dd_i \log F(T_i)  + (1-\dd_i)  \log\{1-F(T_i)\}  \right],
\end{align}
over all distribution functions $F: [0,\infty] \mapsto [0,1]$. \cite{GrWe:92} show that the MLE can be characterized as the left-continuous slope of the greatest convex minorant of a cumulative sum diagram consisting of the points (0,0) and
\begin{align*}
\left(i, \sum_{j\le i} \dd_{(j)} \right),
\end{align*}
where we let $T_{(j)}$ denote the $j$th order statistic of the $T_i$ and $\dd_{(j)}$ be the $\dd_i$ corresponding to it (assuming no ties are present in the data). An important property of the MLE is the so-called \textit{switch relation}, see \cite{piet_geurt:14} p.\ 69. Let $\G_n$ be the empirical distribution function of $T_1,\ldots,T_n$ and define the process
$V_n$ by
\begin{align}
\label{V_n}
V_n(t)=n^{-1}\sum_{i=1}^n\dd_i1_{\{T_i\le t\}},
\end{align}
and the process (in $a$) $U_n$ by
\begin{equation}
\label{argmin_process}
U_n(a)=\argmin\{t\in\R:V_n(t)-a\G_n(t)\}.
\end{equation}
Then, taking $a=F_0(t)$, we get the {\it switch relation}:
\begin{align*}
F_n(t) \ge a \iff U_n(a)\le t,
\end{align*}
see also Figure \ref{fig:switch}.
\begin{figure}[!ht]
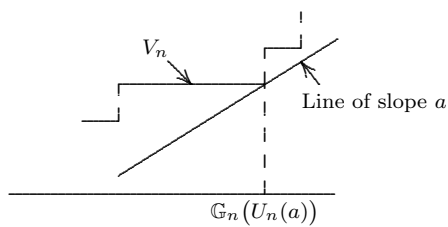

	\begin{center}
		\centerline{\input switch.pic}
	\end{center}
	\caption{\small The switch relation.}
	\label{fig:switch}
\end{figure}

\section{Bootstrapping the MLE}
\label{sec:bootstrap}
In this section we establish properties of the bootstrap MLE $\hat F_n$ based on the nonparametric bootstrap proposed by \cite{efron1979}. Our main concern is to show that conditional on the data $Z_1,\ldots,Z_n$, we have
\begin{align}
\label{bootstrap_result_main}
E\left\{ \left\| n^{1/3} \left\{ \hat F_n - F_0\right\}\right\|_p \Bigm| Z_1,\ldots,Z_n \right\}= O_p(1),
\end{align}
and
\begin{align}
\label{bootstrap_result_main1}
\sup_{t\in[0,R]}E\left\{\left. n^{1/3} \bigl|\hat F_n(t) - F_0(t)\bigr| \right|Z_1,\dots,Z_n\right\}= O_p(1).
\end{align}

Denote the empirical probability measure of $Z_1,\ldots,Z_n$ by $\P_n$. The bootstrap empirical measure is
\begin{align*}
\hat \P_n = n^{-1}\sum_{i=1}^{n}M_{ni}1_{Z_i},
\end{align*}
where $1_{Z_i}$ denotes the points mass at $Z_i=(T_i,\dd_i)$ and 
$$M_n=(M_{n1},\ldots,M_{nn}) \sim \text{multinomial}(n,n^{-1},\ldots,n^{-1}),$$
 is a vector of multinomial weights, independent of  $Z_1,\ldots,Z_n$. The bootstrap MLE $\hat F_n$ is computed using the weighted cumulative sum diagram formed by the point $(0,0)$ and
\begin{align*}
\left(\sum_{j=1}^i M_{n(j)},\sum_{j=1}^i M_{n(j)}\dd_{(j)}\right),
\end{align*}
where $M_{n(j)}$ corresponds to the multinomial weight corresponding to $T_{(j)}$.
The bootstrap MLE $\hat F_n$ is then calculated from the left-continuous slope of the convex minorant of this cusum diagram. 

To complete notation, we suppose that the vectors $((Z_1,\dots,Z_n),M_n), n=1,2,\ldots$ are defined on the product space
$([0,R]\times\{0,1\})^{\infty}\times \Z_+^{\infty}, {\cal B}, P_{ZM})$, where $\Z_+$ is the set of nonnegative integers and ${\cal B}$
is the collection of Borel sets, generated by the finite dimensional projections. 
We say that a real-valued function $\Gamma_n$ defined on the joint probability space is of order $o_{P_M}(1)$ in probability if for all $\epsilon,\eta>0$:
\begin{align*}
P^*\left( P_{M|Z}\left\{|\Gamma_n| > \epsilon \right\} > \eta \right) \to 0 \text{ as } n\to \infty,
\end{align*}
where $P^*$ denotes outer probability and $P_{M|Z}$ is the conditional probability measure w.r.t. the weights, given the sample $Z_1,\ldots,Z_n$.\\

To establish (\ref{bootstrap_result_main}), we need the following result, which is a bootstrap version of Lemma 11.5 in \cite{piet_geurt:14}.
\begin{lemma}
\label{lemma_Th11.3}
Suppose $F_0$ has a continuous density $f_0$ with support [0,R] that satisfies,
\begin{align*}
0 < \inf_{t\in [0,R]} f_0(t)< \sup_{t\in [0,R]} f_0(t) < \infty.
\end{align*}
Also suppose that the observation distribution $G$ has a continuous derivative $g$ that stays away form zero and infinity on $[0,R]$. 
Let	
\begin{equation*}
U(a)=F_0^{-1}(a) \qquad 0<a<1,
\end{equation*}
and define the process
\begin{equation*}
\hat U_n(a)=\text{\rm{argmin}}\{t\in [0,R]:\hat V_n(t)-a\hat \G_n(t)\} \qquad 0<a<1,
\end{equation*}   
with processes $\hat V_n$ and  $\hat \G_n$ defined by
\begin{align}
\label{hatV_n_G_n}
\hat V_n(t)=\int_{u\in[0,t]}\d\,d\hat{\P}_n(u,\d) \quad \text{ and } \quad  \hat \G_n(t)=\int_{u\in[0,t]}\,d\hat{\P}_n(u,\d) \quad\,t\in[0,R]. 
\end{align}
Then there are positive constants $K_1$ and $K_2$, such that, for all $a\in(0,1)$ and for all large $n$:
\begin{align*}
\left\{\exists x\in[0,R]:P_{M|Z}\left\{n^{1/3}\left|\hat U_n(a)-U(a)\right|\ge x\right\}> K_1e^{-K_2 x^{3/2}}\right\}=o_p(1),
\end{align*}
where $\{A\}$ denotes the indicator $1_A$ of the event $A$. 
\end{lemma}

Lemma \ref{lemma_Th11.3} implies that the probability that for all $x\in[0,R]$, and $a=F_0(t)$,
	$$  P_{M|Z}\left\{n^{1/3}\left|\hat U_n(a)-U(a)\right|\ge x\right\}\le K_1e^{-K_2 x^{3/2}}$$
	tends to 1 as $n\to\infty$.
The proof of Lemma \ref{lemma_Th11.3} is given in Section \ref{sec:appendix}. The proof uses empirical process theory and results on tail probabilities for $\|\sqrt{n}(\hat{\P}_n-\P_n)\|_{{\cal F}}$ for classes $\cal F$ with finite entropy integrals. Similar results are proved using martingale theory in Section 11.2 of \cite{piet_geurt:14} for the original sample and in \cite{kim_piet:17:SJS} for a smooth bootstrap empirical process. Since
\begin{align*}
E_{M|Z}\left[n^{1/3}\{\hat F_n(t)-F_0(t)\}_+  \right]^{p} =  \int_{0}^{\infty}\hspace{-0.3cm} P_{M|Z}\left\{n^{1/3}\{\hat F_n(t)-F_0(t)\}\ge x\right\}px^{p-1}dx,
\end{align*}
where $\{\hat F_n(t)-F_0(t)\}_+$ denotes the positive part of $\{\hat F_n(t)-F_0(t)\}$ and since,
\begin{align*}
&P_{M|Z}\left\{\hat U_n\left(a+n^{-1/3}x\right)\le t\right\} \\
&\quad =P_{M|Z}\Biggl[n^{1/3}\left\{ \hat U_n\left(a+n^{-1/3}x\right) -  U\left(a+n^{-1/3}x\right) \right\} \\
&\qquad\qquad\qquad\qquad\qquad\qquad\qquad\le n^{1/3}\left\{ t-  U\left(a+n^{-1/3}x\right)\right\}\Biggr],
\end{align*}
it follows from Lemma \ref{lemma_Th11.3} and the \textit{bootstrapped switch relation} given by
\begin{align*}
P_{M|Z}\left\{n^{1/3}\{\hat F_n(t)-F_0(t)\}\ge x\right\}=P_{M|Z}\left\{\hat U_n\left(a+n^{-1/3}x\right)\le t\right\},
\end{align*} 
that there exists a positive constant $K>0$ such that,
\begin{align*}
\left\{\exists t\in[0,R]: E_{M|Z}\left| \hat F_n(t) - F_0(t)\right|^p > Kn^{-p/3}\right\}=o_p(1).
\end{align*}
In particular, there exists a $K_1>0$ such that:
\begin{align*}
P\left\{\sup_{t\in[0,R]}E_{M|Z}\left|\hat F_n(t)-F_0(t)\right|>K_1n^{-1/3}\right\}\longrightarrow0,\qquad n\to\infty,
\end{align*}
and likewise there exists a $K_2>0$ such that:
\begin{align*}
P\left\{E_{M|Z}\bigl\|\hat F_n-F_0\bigr\|_2>K_2n^{-1/3}\right\}\longrightarrow0,\qquad n\to\infty.
\end{align*}
In the next section we show how (\ref{bootstrap_result_main}) can be used to justify the bootstrap validity for drawing inferences in models which can be estimated using smooth functionals of the MLE. The proofs for deriving the asymptotic behavior of these functionals are in general based on applications of the Cauchy-Schwarz inequality and on showing asymptotic equicontinuity. Both steps involve calculating the $L_2$-distance which can often be reduced to the $L_2$-distance between the MLE and the true underlying distribution function. Our main result given in (\ref{bootstrap_result_main}) is therefore important to show that the asymptotic properties of the estimates obtained in the original sample are still valid in the bootstrap sample conditionally on the data. The asymptotic behavior of the functionals does not depend on the distribution function of the MLE, which is, as shown in Theorem 5 of \cite{abrevaya_huang2005}, not the same in the original sample and the bootstrap sample (conditionally on the data). We note that the variances of the corresponding asymptotic distributions however still have the same order $n^{-2/3}$, just like our squared $L_p$-distances in (\ref{bootstrap_result_main}). 

\section{Applications}
\label{sec:applications}
In this Section we illustrate the applicability of our bootstrap results. In our first example we consider the current status model described in Section \ref{sec:model}  and estimate $F_0$ by the SMLE. In the second example we consider estimating a finite dimensional regression parameter for the current status model, where in addition to observing the vector $(T,\dd)$, also a covariate vector $X$ is observed.
\subsection{The Smoothed Maximum Likelihood Estimator (SMLE)}
\label{sec:SMLE}
We estimate $F_0$ by the SMLE $\tilde F_{nh}$ obtained by first estimating the MLE $F_n$ and then smoothing this using a smoothing kernel, i.e.,
\begin{align}
\label{SMLE}
\tilde F_{nh}(t)=\int\IK\left((t-x)/h\right)\,d F_n(x),
\end{align}
where $\IK$ is an integrated kernel,
\begin{align*}
\IK(u)=\int_{-\infty}^u K(x)\,dx,
\end{align*}
and where $h$ is a chosen bandwidth. Here  $dF_n$ represents the jumps of the discrete distribution function $F_n$ and $K$ is one of the usual symmetric twice differentiable kernels with compact support, used in density estimation. In our computer experiments, we used the triweight kernel
\begin{align*}
K(u)=\frac{35}{32}\left(1-u^2\right)^31_{[-1,1]}(u).
\end{align*}

For a constant $c>0$ and  $h= cn^{-1/5}$, the SMLE has been proved to converge at rate  $n^{-2/5}$ with asymptotic limit distribution,
\begin{align*}
n^{2/5}\left\{\tilde F_{nh}(t)-F_0(t)\right\}\stackrel{{\cal D}}\longrightarrow N(\b(t),\s^2(t)),
\end{align*}
where
\begin{align}
\label{mu-sigma}
\b(t)=\frac{c^2f_0'(t)}{2}\int u^2 K(u)\,du \quad \text{and} \quad \s^2(t)=\frac{F_0(t)\{1-F_0(t)\}}{cg(t)}\int K(u)^2\,du.
\end{align}
(see \cite{piet_geurt_birgit:10}). The SMLE is often used in the smooth bootstrap procedures described in Section \ref{sec:intro} (see also the numerical example below). Let $\tilde F_{nh}^*(t)$ be the bootstrapped SMLE based on replacing $F_n$ in (\ref{SMLE}) by the bootstrapped MLE $\hat F_n$, then we have the following result,
\begin{align}
\label{SMLE_AN_bootstrap}
n^{2/5}\left\{\tilde F_{nh}^*(t)-\tilde F_{nh}(t)\right\}\stackrel{{\cal D}}\longrightarrow N(0,\s^2(t)),
\end{align}
given the data $(T_1,\dd_1),\dots,(T_n,\dd_n)$, in probability. 
Note that, in contrast to the smooth bootstrap method described in \cite{kim_piet:17:SJS}, we do not need to estimate the convolution SMLE (see (\ref{convoluted_SMLE}) below). 

To prove the asymptotic normality result for the nonparametric bootstrap, given in (\ref{SMLE_AN_bootstrap}), we prove (in Section \ref{sec:appendix}) the following Lemma:
\begin{lemma}
	\label{lemma: toy-representation-bootstrapSMLE}
	Assume that the conditions of Lemma \ref{lemma_Th11.3} are satisfied and that $g$ has a bounded derivative $g'$ on $[0,R]$. Let $t$ be an interior point of $[0,R]$ such that $f_0$ has a continuous derivative $f_0'$ at $t$. If $h\sim cn^{-1/5}$ then,
	\begin{align*}
	\tilde F_{nh}^*(t) &=  \tilde F_{nh}^{(toy)*}(t) +  o_{P_M}(n^{-2/5}),
	\end{align*}
	in probability, where
	\begin{align}
	\label{bootstrap_toy_estimator_definition}
	\tilde F_{nh}^{(toy)*}(t) &= \int\IK((t-u)/h)\,dF_0(u)+\int \frac{K((t-u)/h)\,\left\{\d-F_0(u)\right\}}{hg(u)}\,d\hat\P_n(u,\d).
	\end{align}
\end{lemma}
Since 
\begin{align}
\label{toy-original}
\tilde F_{nh}(t) = \tilde F_{nh}^{(toy)}(t)  + o_p(n^{-2/5}),
\end{align}
where $\tilde F_{nh}^{(toy)}(t)$ is defined by (\ref{bootstrap_toy_estimator_definition}) with $\hat \P_n$ replaced by $\P_n$, we have by Lemma \ref{lemma: toy-representation-bootstrapSMLE} that,
\begin{align*}
n^{2/5} \left\{ \tilde F_{nh}^*(t) -  \tilde F_{nh}(t)\right\} &=
n^{2/5}  \int \frac{K((t-u)/h)\left\{\d-F_0(u)\right\}}{hg(u)}\,d(\hat\P_n- \P_n)(u,\d) \\
&\qquad+ o_{P_M}(1),
\end{align*} 
in probability, which converges, conditional on the data $(T_1,\dd_1),\ldots (T_n,\dd_n)$ to the same asymptotic limit as
\begin{align*}
n^{2/5}  \int \frac{K((t-u)/h)\left\{\d-F_0(u)\right\}}{hg(u)}\,d(\P_n- P)(u,\d),
\end{align*} 
in probability (see e.g.  \cite{hall:92book} for more details about the use of the bootstrap for kernel estimators).
Finally, applying the central limit theorem on the expression above proves the asymptotic normality result for the bootstrapped SMLE given in (\ref{SMLE_AN_bootstrap}).
The proof of Lemma \ref{lemma: toy-representation-bootstrapSMLE} is a generalization of the proof for the representation of the SMLE $\tilde F_{nh}(t)$ as the ``toy-estimator'' defined in (\ref{toy-original}). The proof is outlined in Section 11.3 of \cite{piet_geurt:14} and uses the result of Theorem 11.3 given in Section 11.2 which is the analogue of our Lemma \ref{lemma_Th11.3} in the original sample.

\begin{remark}{\rm
		In practice, one should use a boundary correction to ensure consistent estimation of $F_0$ near the boundaries of the support $[0,R]$. In our experiments we used the method of \cite{schuster:85}, see also p.\ 328 in \cite{piet_geurt:14}. It is straightforward to show that the nonparametric bootstrap method remains valid under this boundary correction. Moreover, one should also take into account the bias defined in (\ref{mu-sigma}) when constructing confidence intervals around the SMLE. 
		The bias issue is discussed in more details via a simulation study in  Section  \ref{subsec:simulation_SMLE}.}
\end{remark}

\medskip
In the remainder part of this Section, we show the applicability of this bootstrap result (\ref{SMLE_AN_bootstrap}) by constructing pointwise confidence intervals (CIs) around the SMLE. We consider two different simulation models and a real data example to illustrate the performance of these CIs. 

In the first simulation study we compare our nonparametric bootstrap CIs with (a) the smooth bootstrap CIs proposed in \cite{kim_piet:17:SJS}, (b) the likelihood ratio intervals around the MLE $F_n$ proposed in \cite{banerjee_wellner:2005}, (c) the smooth bootstrap MLE-based intervals proposed in \cite{SenXu2015} and (d) Wald-type CIs, derived from the asymptotic normality of the SMLE.

In a second simulation study, we discuss the difficulties with the construction of pointwise CIs around the SMLE that are not necessarily specific to the bootstrap procedure but that have to be taken into account in order to obtain good CIs around the SMLE under current status data. We first describe a bandwidth selection procedure for choosing the bandwidth of the SMLE and we next discuss the effect of the bias on the performance of the CIs. The algorithms to produce the proposed CIs around the SMLE can be found in the R package \texttt{curstatCI}.

\subsubsection{Simulation Study 1: comparing CIs for the distribution function under current status data}
\label{subsec:simulation_SMLE}
To illustrate the performance of the nonparametric bootstrap procedure for constructing pointwise CIs of the distribution function, we consider a first simulation study based on $N=5,000$ simulation runs from a model where both $X$ and $T$ have a Uniform(0,2) distribution. In this model the bias $\b(t)$ defined in (\ref{mu-sigma}) is zero for all $t \in [0,2]$. The $1-\a$ bootstrap interval is given by
\begin{equation}
\label{CI_type2}
\left[\tilde F_{nh}(t)-Q_{1-\a/2}^*(t)\sqrt{S_{nh}(t)},
\tilde F_{nh}(t)-Q_{\a/2}^*(t)\sqrt{S_{nh}(t)}\right],
\end{equation}
where $Q_{\a}^*(t)$ is the $\a$th quantile of $B$ values of $W_{nh}^*(t)$ defined by
\begin{align*}
W_{nh}^*(t)=\left\{\tilde F_{nh}^*(t)-\tilde F_{nh}(t)\right\}/\sqrt{S_{nh}^*(t)},
\end{align*}
where $S_{nh}(t)$ resp. $S_{nh}^*(t)$ are estimates of the variance $\s^2(t)$ defined in (\ref{mu-sigma}) (apart from the factor $cg(t)$ which drops out in the Studentized bootstrap procedure) given by
\begin{align*}
S_{nh}(t)&=\frac1{n^{2}}\sum_{i=1}^n K_h(t-T_i)^2\left(\dd_i-F_n(T_i)\right)^2 , \\ S_{nh}^*(t)&=\frac1{n^{2}}\sum_{i=1}^n  M_{ni} K_h(t-T_i)^2\left(\dd_i-\hat F_n(T_i)\right)^2.
\end{align*}
In Figure \ref{fig:SMLE_simulation}(a) we compare the proportion of times that $F_0(t)$ is not in the 95\% bootstrap CIs for $t = 0.02,0.04,\ldots, 2$ with the corresponding proportions obtained with (a) the smooth bootstrap procedure proposed in \cite{kim_piet:17:SJS}, (b) the likelihood ratio intervals around the MLE $F_n$ proposed in \cite{banerjee_wellner:2005} and (c) the smooth bootstrap MLE-based intervals proposed in \cite{SenXu2015}. For samples of size $n=1,000$, $B=1,000$ bootstrap samples were generated for both methods and the triweight kernel is used for calculation of the SMLE with $h=2n^{-1/5}$, where the constant $c=2$ corresponds to the length of the support of the observation variable $T$.
For the smooth bootstrap procedures (a) and (c), first a bootstrap sample $(T_1,\dd_1^*),\dots,(T_n,\dd_n^*)$ is obtained by keeping the $T_i$ in the original sample fixed and by resampling the $\dd_i^*$ from a Bernoulli distribution with probability $\tilde F_{nh}(T_i)$, then the bootstrap MLE $\hat F_n$ and SMLE $\tilde F_{nh}^*$ are estimated based on the $(T_i,\dd_i^*),i=1,\ldots,n$. T
he smooth bootstrap $1-\a$ intervals around the SMLE proposed in \cite{kim_piet:17:SJS} are then constructed via ({\ref{CI_type2}}), except that the SMLE $\tilde F_{nh}(t)$ in the definition of $W_{nh}^*(t)$ is replaced by the convolution SMLE given by
\begin{align}
\label{convoluted_SMLE}
\int \IK_h(t-u)\,d\tilde F_{nh}(u),
\end{align}
and that the variance estimate in the bootstrap sample is given by
\begin{align*}
\frac1{n^2}\sum_{i=1}^n K_h(t-T_i)^2\left(\dd_i^*-\hat F_n(T_i)\right)^2.
\end{align*}
The convolution SMLE corresponds to the extra level of smoothing introduced by the smooth bootstrap procedure and is hence not required for the nonparametric bootstrap. The smooth bootstrap CIs of \cite{SenXu2015} around the MLE are given by
\begin{equation*}
\left[F_{n}(t)-Z_{1-\a/2}^*(t),
 F_{n}(t)-Z_{\a/2}^*(t)\right],
\end{equation*}	
where $Z_{\a}^*(t)$ is the $\a$th quantile of $B$ values of 
$\hat F_n(t) - \tilde F_{nh}(t),$
where again the extra level of smoothing is introduced (since one subtracts $\tilde F_{nh}$ and not $F_n$) to justify the smooth bootstrap procedure.

The performance of the SMLE-based CIs is comparable. The bootstrap intervals based on the classical bootstrap procedure avoid however calculation of the convolution SMLE defined in (\ref{convoluted_SMLE}). The CIs in (b) and (c) have similar coverage proportions in the middle of the interval $[0,2]$ but have a worse behavior near the boundaries of the interval compared to the SMLE-based intervals.
	
Figure \ref{fig:SMLE_length}(a) shows the average length of both bootstrap intervals around the SMLE in comparison with the average length of the likelihood ratio CIs of \cite{banerjee_wellner:2005} and the smooth MLE-based CIs of \cite{SenXu2015}. The latter intervals are constructed around the MLE $F_n$ instead of the SMLE $\tilde F_{nh}$. The length of the MLE-based intervals is larger than the length of the SMLE-based intervals due to the fact that the MLE converges at the slower rate $n^{1/3}$.  \begin{figure}[!ht]
	\centering
	\begin{subfigure}[b]{0.3\textwidth}
		\includegraphics[width=\textwidth]{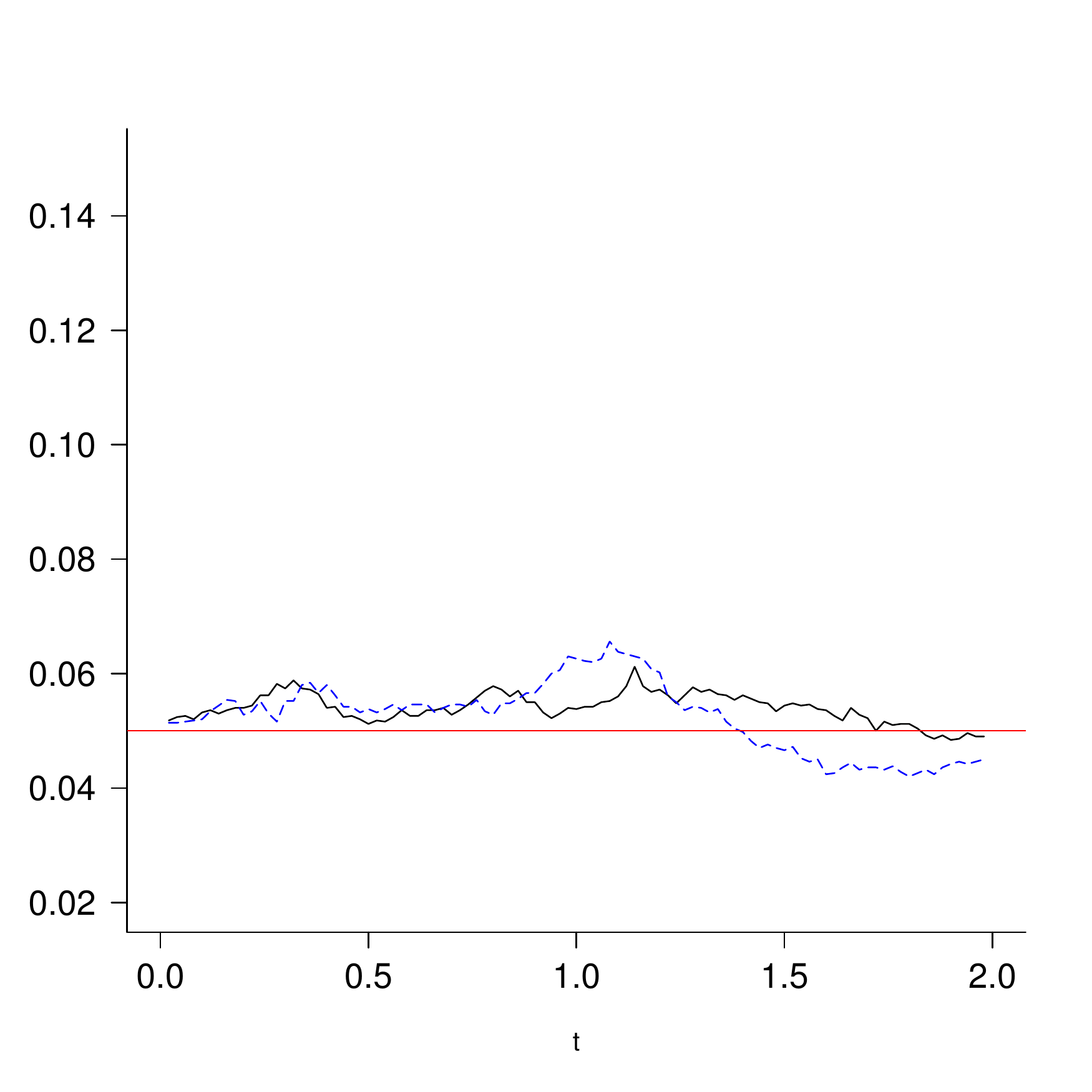}
		\caption{}
	\end{subfigure}
	\begin{subfigure}[b]{0.3\textwidth}
		\includegraphics[width=\textwidth]{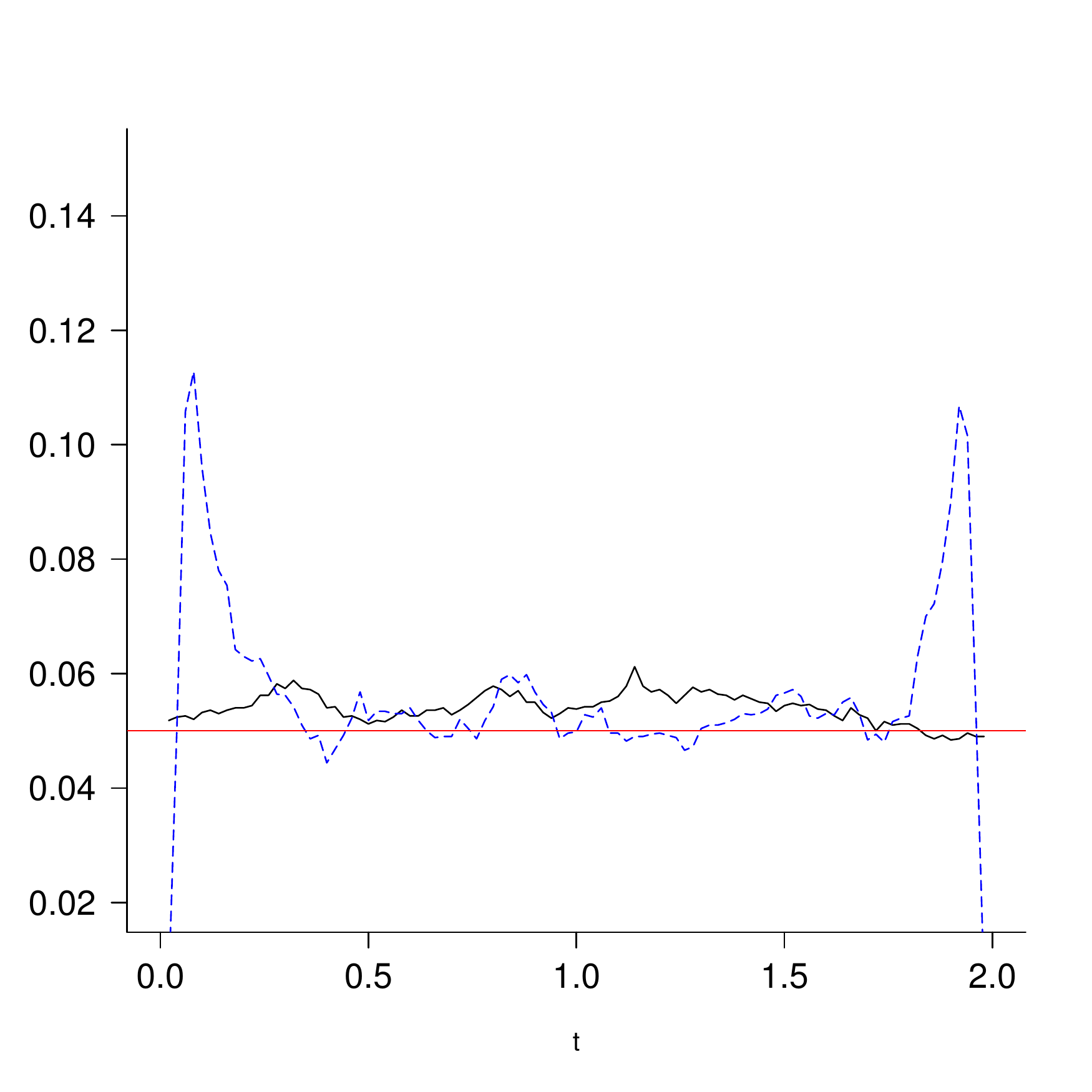}
		\caption{}
	\end{subfigure}	
	\begin{subfigure}[b]{0.3\textwidth}
		\includegraphics[width=\textwidth]{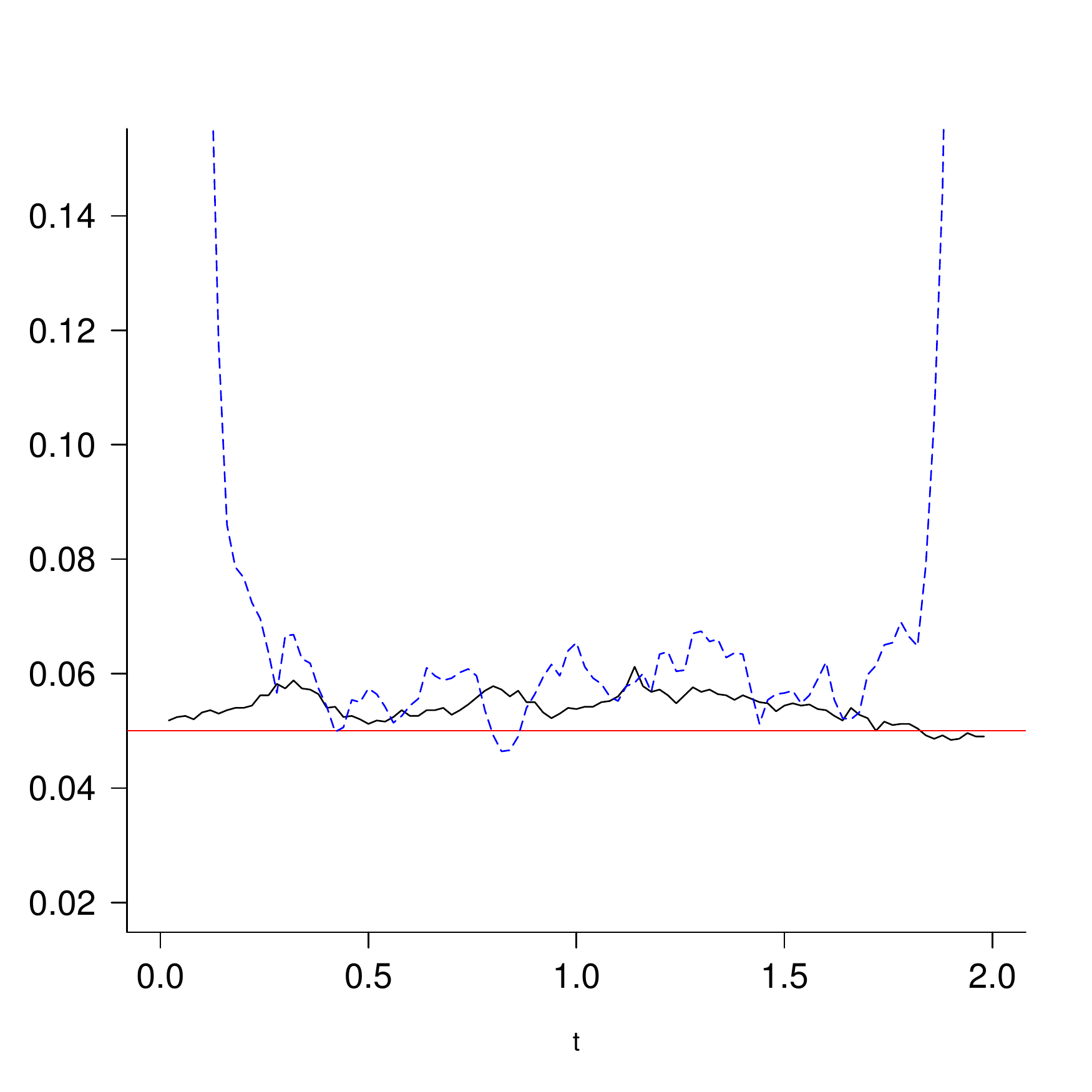}
		\caption{}
	\end{subfigure}	
	\caption{\small Uniform samples: Proportion of times that $F_0(t),\, t=0.02,0.04,\dots$ is not in the $95\%$ CIs for the classical bootstrap CIs defined in (\ref{CI_type2}) (black, solid) and (a) the smooth bootstrap (blue, dashed) procedure in constructing CIs around the SMLE of \cite{kim_piet:17:SJS}, (b) the likelihood ratio CIs of \cite{banerjee_wellner:2005} (blue, dashed) and (c) the smooth MLE-based CIs of \cite{SenXu2015} (blue, dashed). $n =1,000$, $N=5,000$,  $B=1,000$ and $h = 2n^{-1/5}$.}
	\label{fig:SMLE_simulation} 
\end{figure}
\begin{figure}[!ht]
	\centering
	\begin{subfigure}[b]{0.3\textwidth}
		\includegraphics[width=\textwidth]{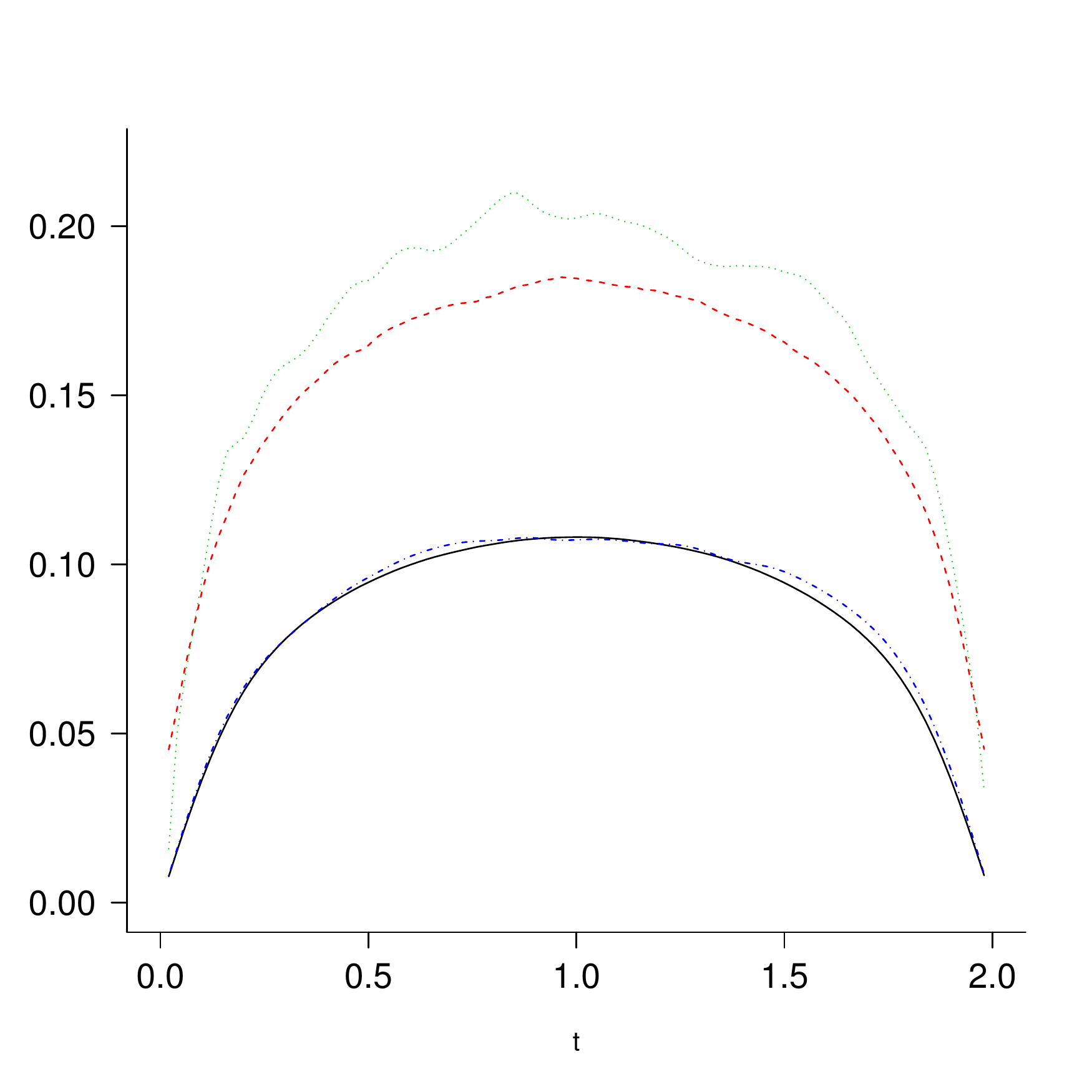}
		\caption{}
	\end{subfigure}
	\begin{subfigure}[b]{0.3\textwidth}
		\includegraphics[width=\textwidth]{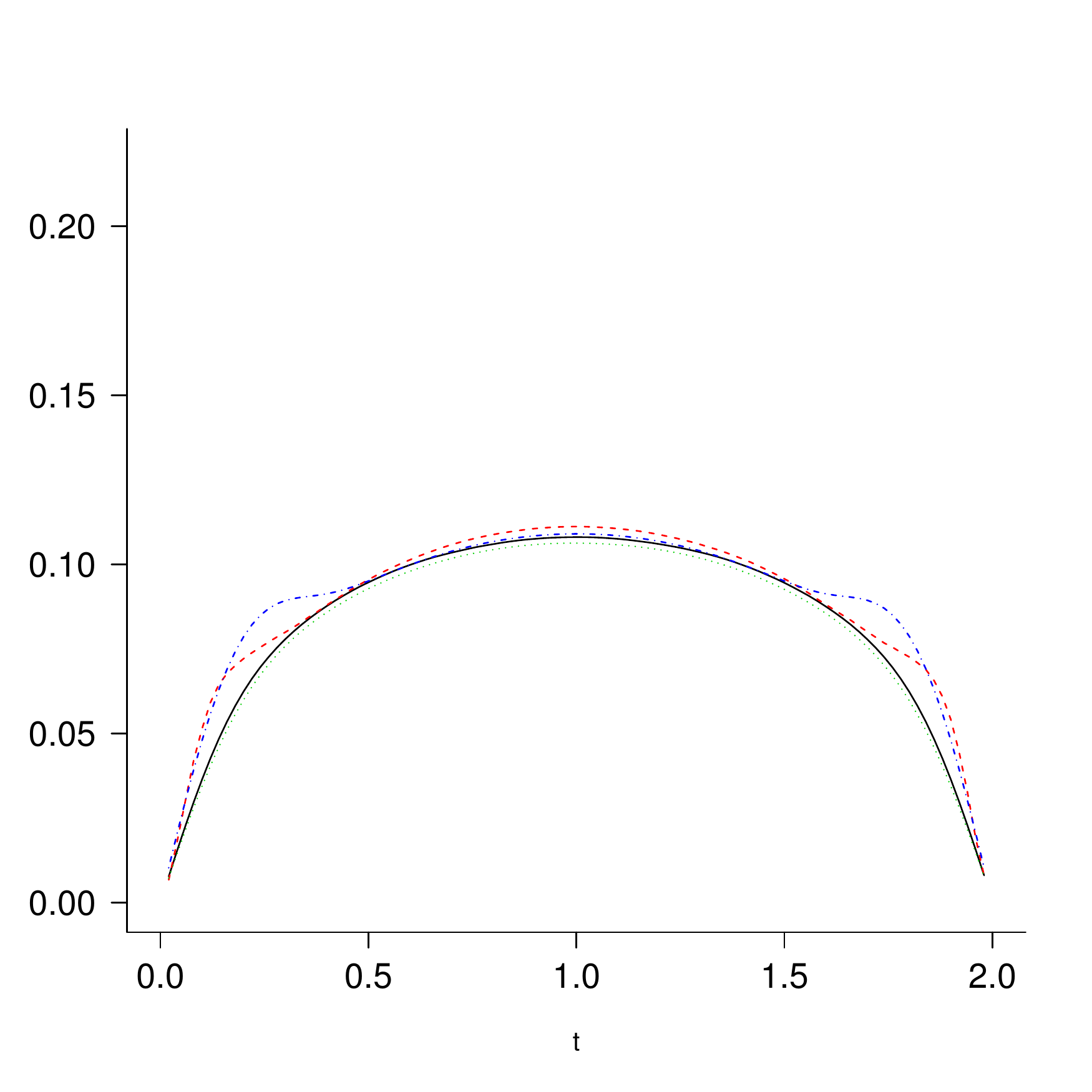}
		\caption{}
	\end{subfigure}	
	\caption{\small Uniform samples: average length of the bootstrap CIs defined in (\ref{CI_type2}) (black, solid) and (a) the smooth bootstrap CIs (blue, dashed-dotted) of \cite{kim_piet:17:SJS}, the likelihood ratio CIs of \cite{banerjee_wellner:2005} (red, dashed) and the smooth MLE-based CIs of \cite{SenXu2015} (green,dotted); and (b) Wald-type CIs using the first estimate $\hat \s_{1,nh}$ (red,dashed), the second estimate $\hat \s_{2,nh}$ (blue,dashed-dotted) and the third estimate $\hat \s_{3,nh}$ (green,dotted). $n =1,000$, $N=5,000$,  $B=1,000$ and $h = 2n^{-1/5}$.}
	\label{fig:SMLE_length} 
\end{figure}

Instead of constructing the Studentized bootstrap intervals where the quantiles of the limiting distribution of the SMLE are derived from the bootstrap distribution, one can alternatively consider Wald-type confidence intervals using the quantiles of the normal distribution and an estimate of the asymptotic variance.  We compare three different estimates $\hat \s_{nh}(t)$ for $\s(t)$ defined in (\ref{mu-sigma}) and construct CIs given by
\begin{align}
\label{CI_wald}
[\tilde F_{nh}(t) &- z_{1-\a/2}(n^{-2/5}\hat \s_{nh}(t)) -  \b(t) n^{-2/5}; \\
&\qquad\qquad\qquad\qquad\qquad \,\,\tilde F_{nh}(t) -  z_{\a/2}(n^{-2/5} \hat \s_{nh}(t)) - \b(t) n^{-2/5} ]\nonumber,
\end{align}
where $z_{\a}$ is the $\a$th quantile of the standard normal distribution. In this simulation study $\b(t)$ defined in (\ref{mu-sigma}) is zero. The effect of $\b(t)$ on the behavior of the intervals will be discussed in the second simulation study below. A first estimate for $\hat \s_{nh}(t)$ is given by
\begin{align}
\label{sigma_estimate1}
\hat \s^2_{1,nh}(t) = \frac{F_n(t)\{1-F_n(t)\}}{cg_{nh}(t)}\int K(u)^2\,du,
\end{align}
where $g_{nh}$ is a classical kernel estimate for the density $g$ of the observation time $T \sim U(0,2)$, using again the Epanechnikov kernel with bandwidth $h=2n^{-1/5}$.
A second estimate for $\s(t)$ is inspired by the fact that the SMLE is asymptotically equivalent to the toy-estimator defined in (\ref{toy-original}), which has a sample variance
\begin{align}
\label{sigma_estimate2}
s^2_{nh}(t) = \frac1{n^{2}}\sum_{i=1}^n \frac{K_h(t-T_i)^2\left(\dd_i-F_0(T_i)\right)^2}{g(T_i)^2}.
\end{align}
This suggests taking the second estimate $n^{-2/5}\hat \s_{2,nh}(t)$ equal to the root of (\ref{sigma_estimate2}) where $F_0$ is replaced by the MLE $F_n$ and $g$ is replaced by the kernel density estimate $g_{nh}$. 
\\
Contrary to the bootstrap procedure for constructing CIs defined in (\ref{CI_type2}), both estimates $\hat \s_{1,nh}(t)$ and $\hat \s_{2,nh}(t)$ require estimating the density $g$. A bootstrap based estimate for the variance, avoiding estimating $g$, is finally given by
\begin{align}
\label{sigma_estimate3}
\hat \s^2_{3,nh}(t) = \frac1{B}\sum_{i=1}^B \left(\tilde F_{nh}^*(t)-\tilde F_{nh}(t)\right)^2.
\end{align}
Figure \ref{fig:SMLE_wald} compares the coverage proportions between the bootstrap CIs in (\ref{CI_type2}) with the Wald-type CIs in (\ref{CI_wald}) using the three different variance estimates described above. Pointwise confidence bands for the variance estimates are illustrated in Figure  \ref{fig:var_wald}. The curves show the average variance estimate and the 5\% and 95\% empirical quantiles of the variance estimates at points $t=0.02,0.04,\dots,2$. The best results for the Wald-type CIs are obtained with the second variance estimate $\hat \s_{2,nh}^2(t)$ but the coverage proportions and average lengths (shown in Figure \ref{fig:SMLE_length}(b)) are inferior to the results obtained with the bootstrap CIs in (\ref{CI_type2}). Estimating the density $g$ in $\hat \s_{1,nh}(t)$ and $\hat \s_{2,nh}(t)$ requires an additional bandwidth selection, whereas the estimate $\hat \s_{3,nh}(t)$ is straightforward to obtain and does not depend on an estimate of $g$. The variance of the first estimate $\hat \s_{1,nh}^2(t)$ is larger than the variance of the second and third variance estimates $\hat \s_{2,nh}^2(t)$ and $\hat \s_{3,nh}^2(t)$ , especially near the boundaries of the support.
\begin{figure}[!ht]
	\centering
	\begin{subfigure}[b]{0.3\textwidth}
		\includegraphics[width=\textwidth]{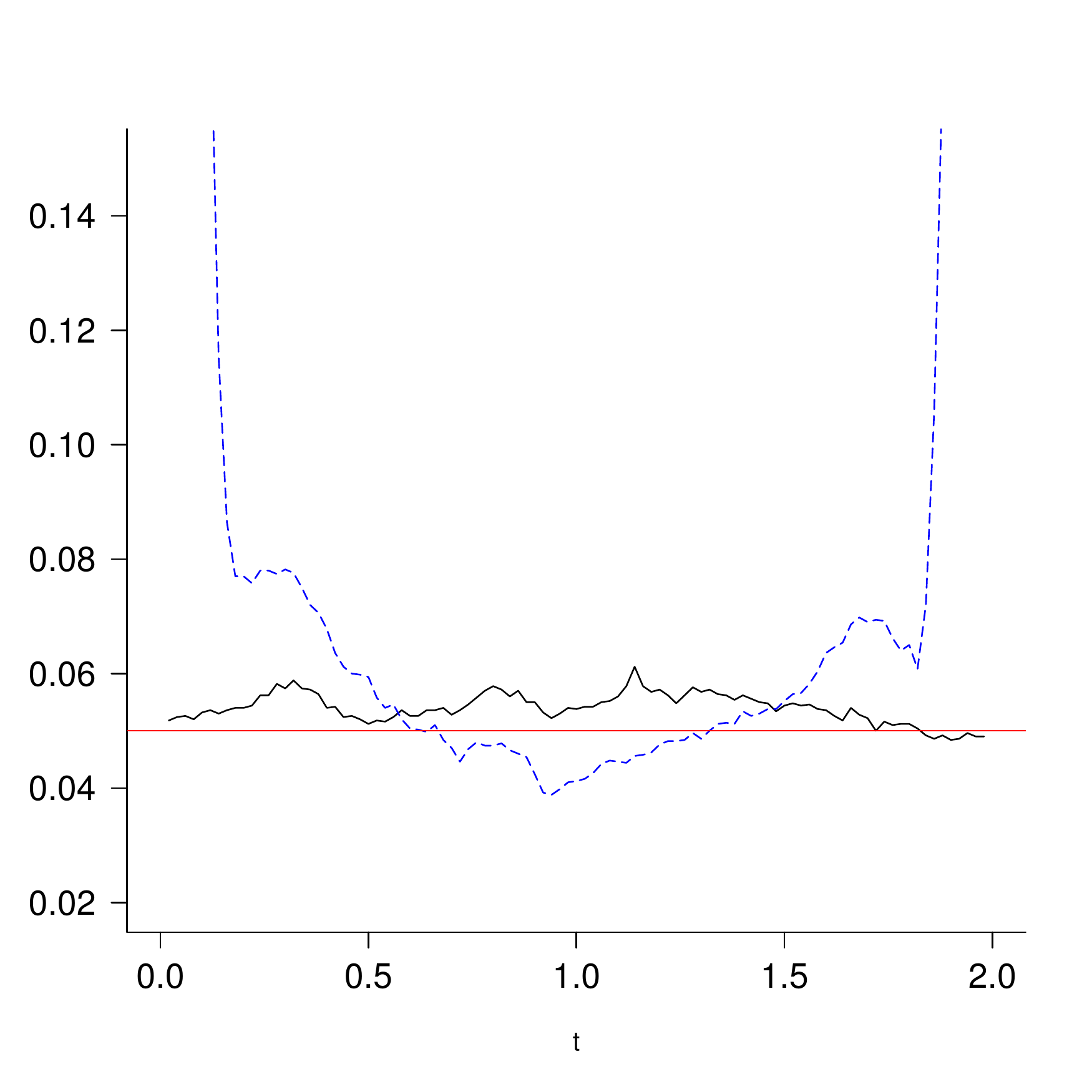}
		\caption{}
	\end{subfigure}
	\begin{subfigure}[b]{0.3\textwidth}
		\includegraphics[width=\textwidth]{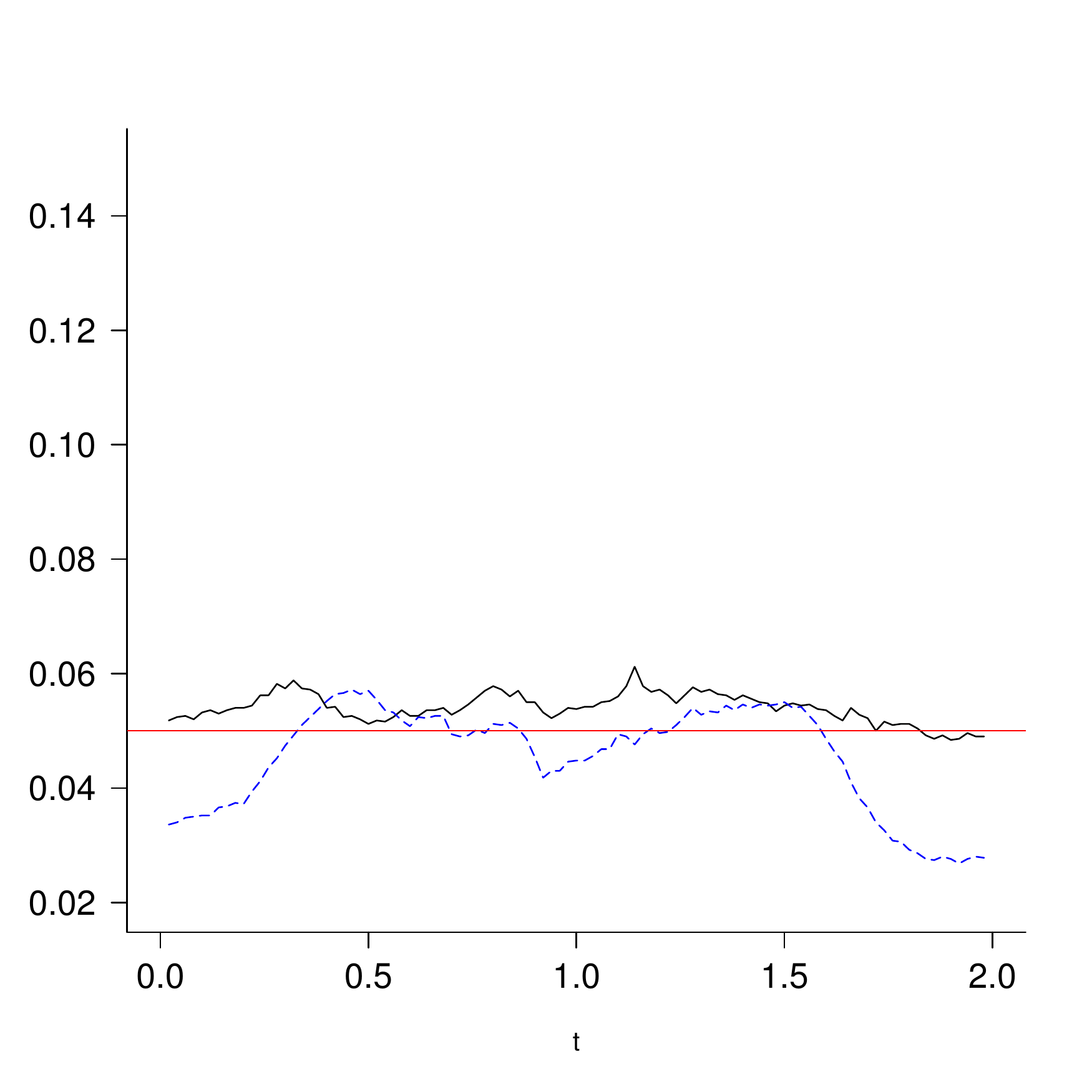}
		\caption{}
	\end{subfigure}	
	\begin{subfigure}[b]{0.3\textwidth}
		\includegraphics[width=\textwidth]{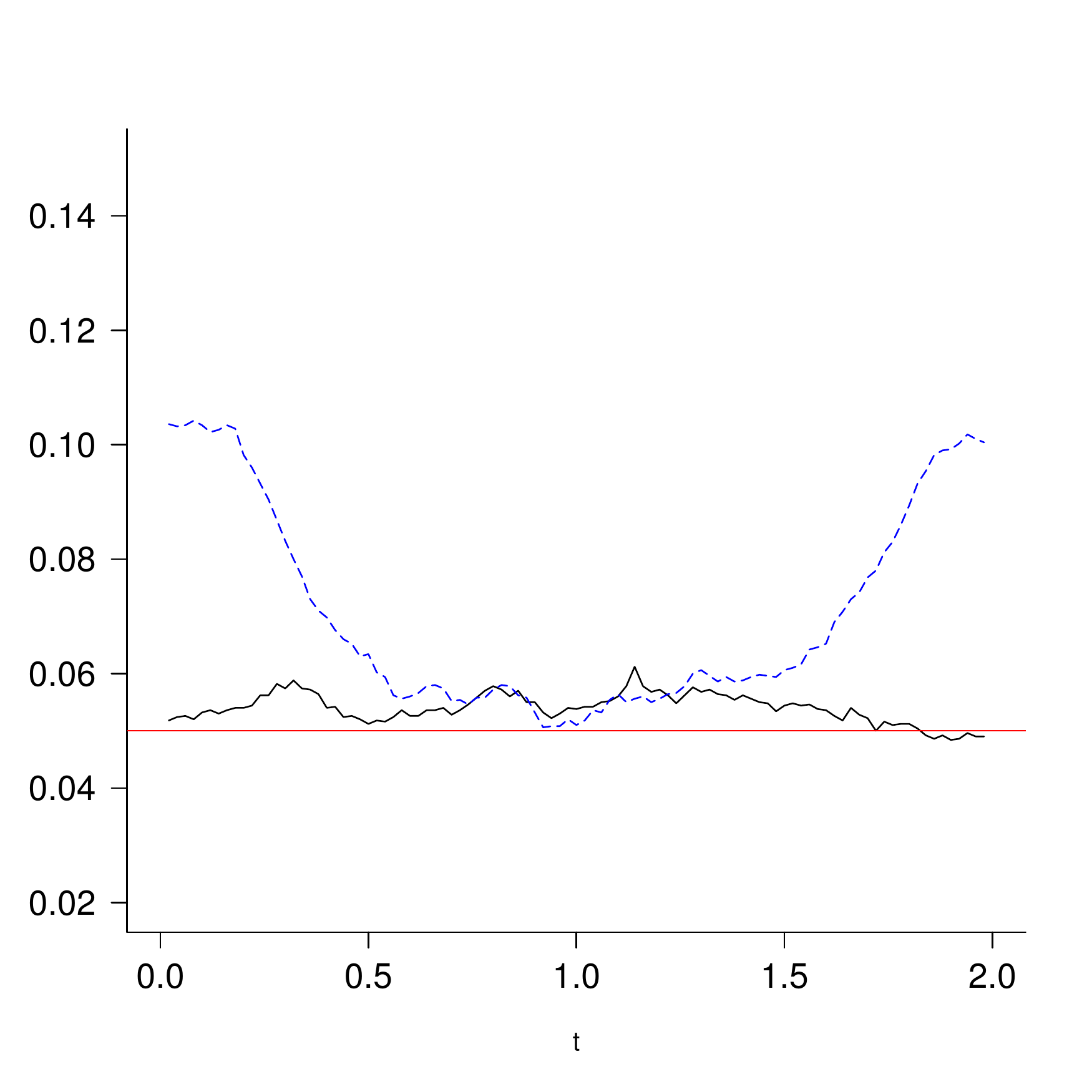}
		\caption{}
	\end{subfigure}	
	\caption{\small Uniform samples: Proportion of times that $F_0(t),\, t=0.02,0.04,\dots$ is not in the $95\%$ CIs for the bootstrap CIs defined in (\ref{CI_type2}) (black, solid) and Wald-type confidence intervals defined in (\ref{CI_wald}) using (a) the first estimate $\hat \s_{1,nh}^2$ (blue,dashed), (b) the second estimate $\hat \s_{2,nh}^2$ (blue,dashed) and (c) the third estimate $\hat \s_{3,nh}^2$ (blue,dotted). $n =1,000$, $N=5,000$,  $B=1,000$ and $h = 2n^{-1/5}$.}
	\label{fig:SMLE_wald} 
\end{figure}

\begin{figure}[!ht]
	\centering
	\begin{subfigure}[b]{0.3\textwidth}
		\includegraphics[width=\textwidth]{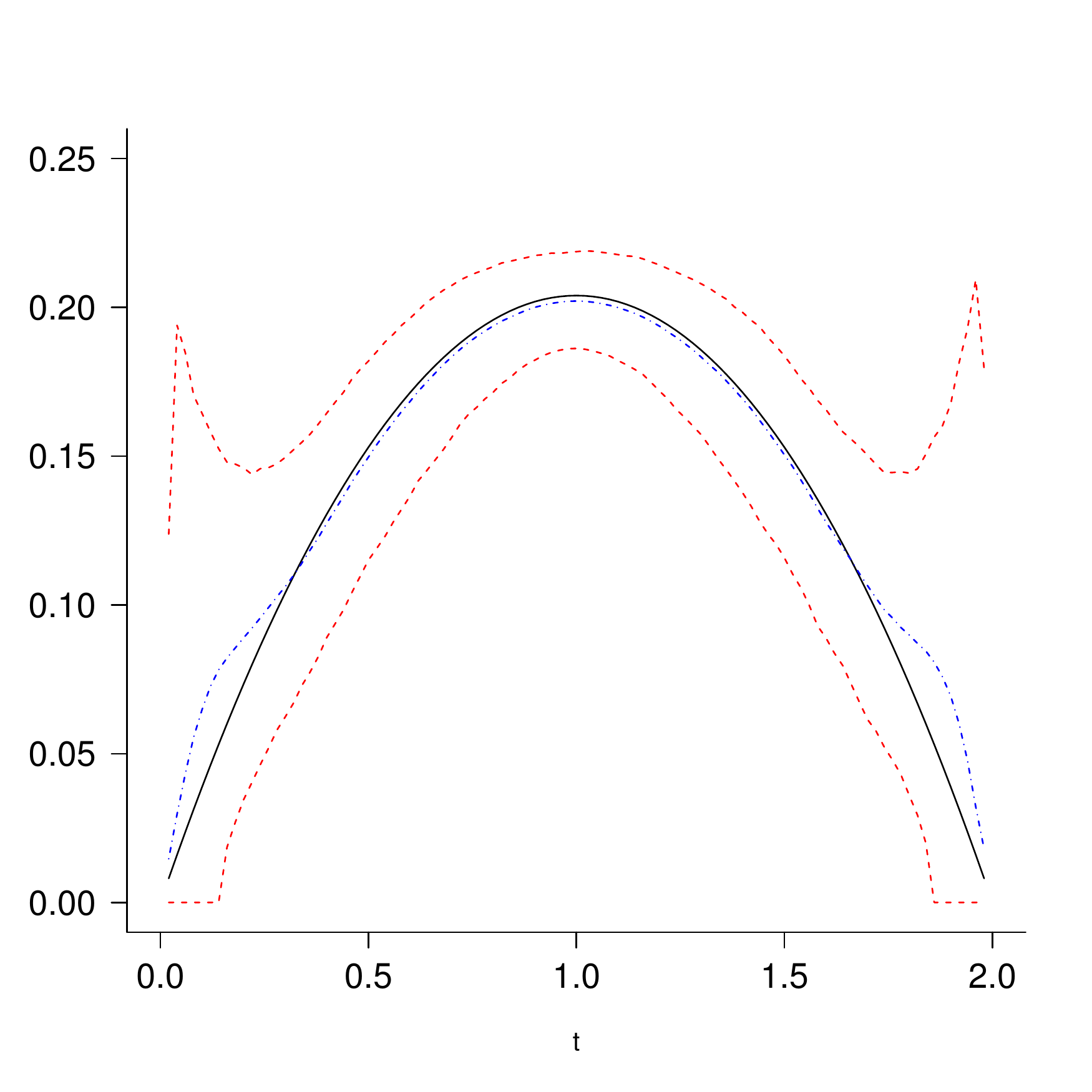}
		\caption{}
	\end{subfigure}
	\begin{subfigure}[b]{0.3\textwidth}
		\includegraphics[width=\textwidth]{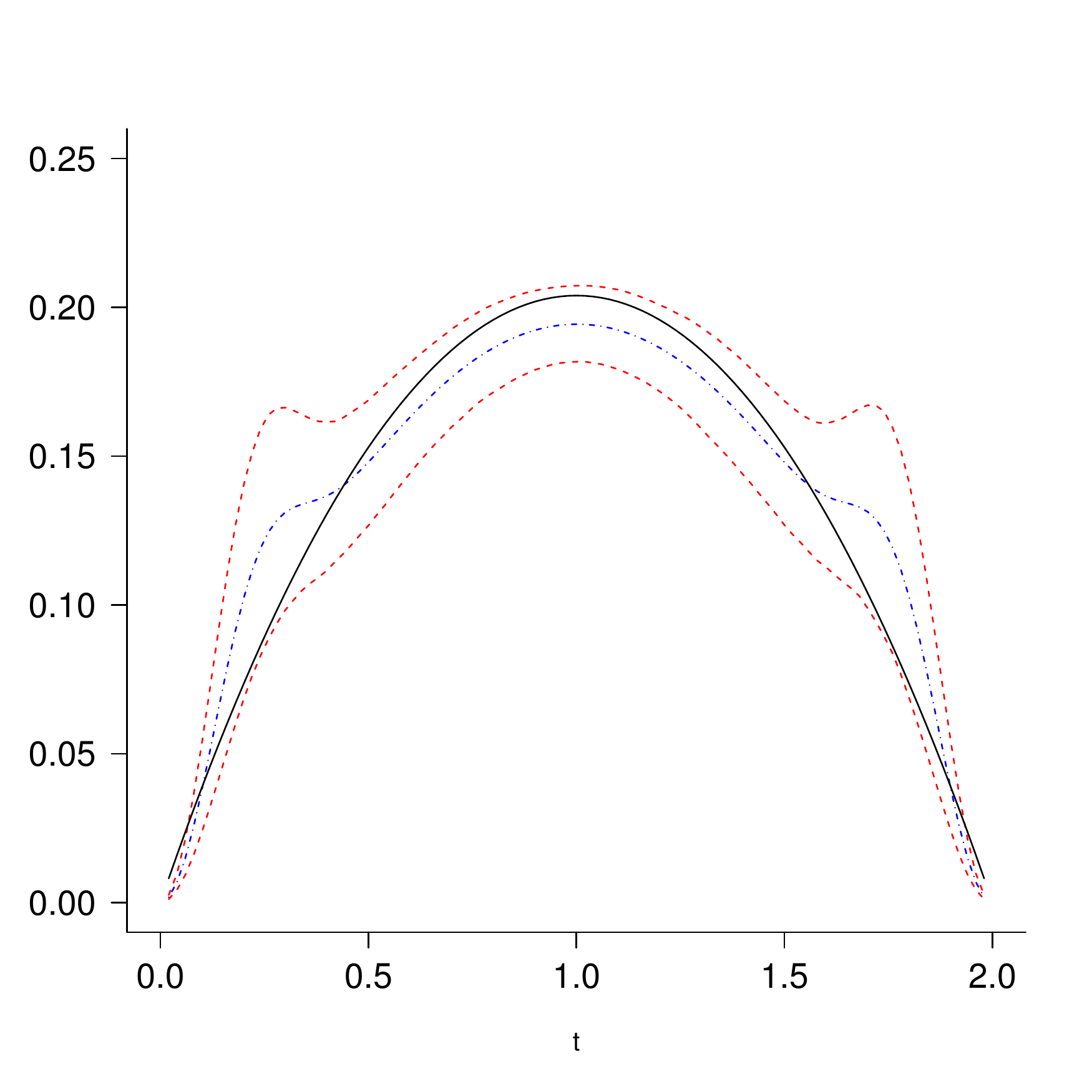}
		\caption{}
	\end{subfigure}	
	\begin{subfigure}[b]{0.3\textwidth}
		\includegraphics[width=\textwidth]{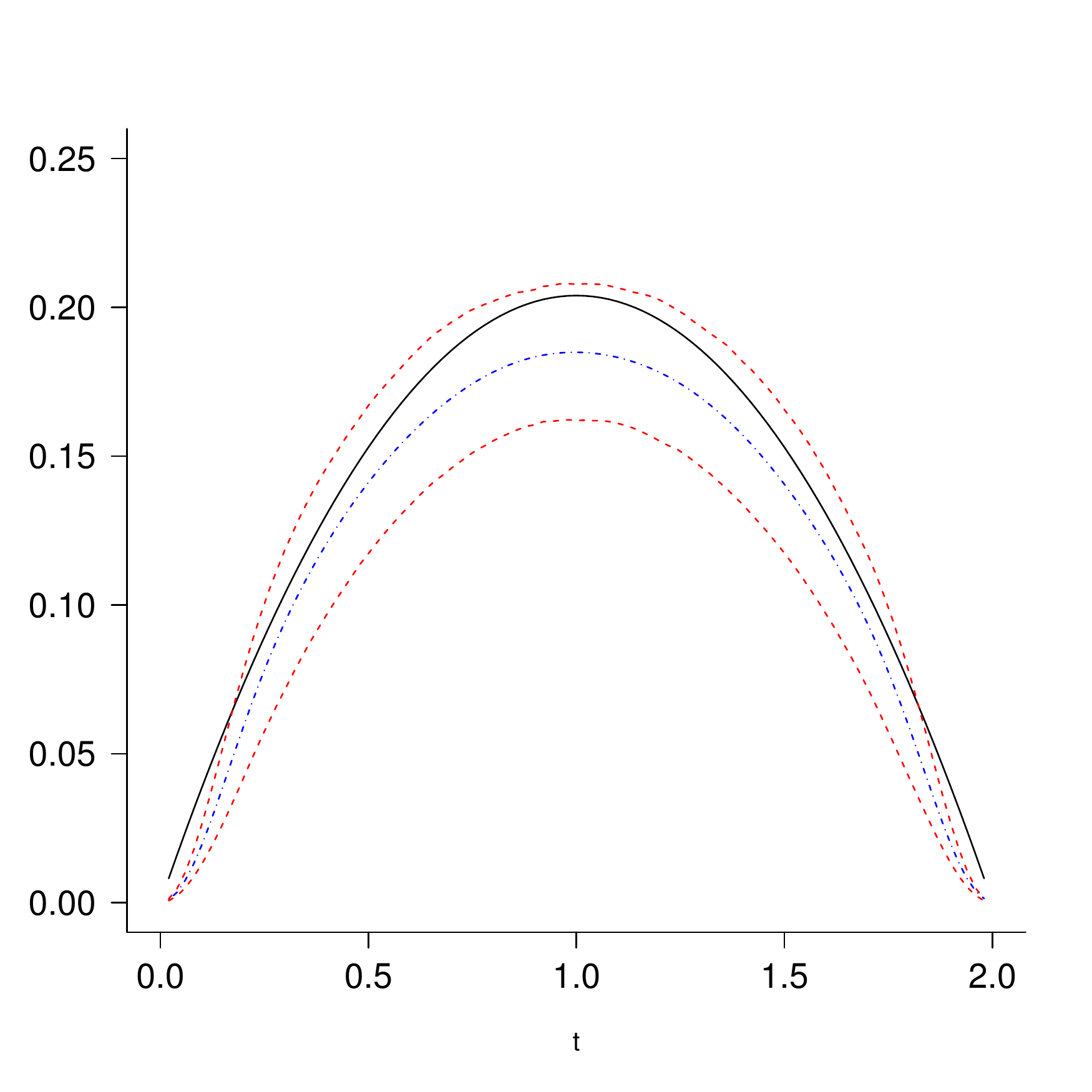}
		\caption{}
	\end{subfigure}	
	\caption{\small Uniform samples: True variance $\s^2$ (black, solid), mean estimate (blue, dashed-dotted) and the 5\% and 95\% empirical quantiles of the estimates (red, dashed) for (a) the first estimate $\hat \s_{1,nh}^2$, (b) the second estimate $\hat \s_{2,nh}^2$ and (c) the third estimate $\hat \s_{3,nh}^2$. $n =1,000$, $N=5,000$,  $B=1,000$ and $h = 2n^{-1/5}$.}
	\label{fig:var_wald} 
\end{figure}
\medskip
Although we have proven validity of the nonparametric bootstrap for constructing pointwise CIs around the SMLE, the performance of the CIs is often influenced by several other aspects that are not specifically due to the nonparametric bootstrap algorithm. In what follows we describe some of these issues further and analyze the problems that can arise in the construction of the CIs.
In a second simulation study we investigate the bias effect. Estimation of the bias defined in (\ref{mu-sigma}) is known to be a rather difficult task since it requires estimating the derivative $f_0'$ of the density $f_0$ under current status data. Sufficiently accurate estimates of the bias are hard to obtain by direct estimation of $f_0'$. Besides estimating the derivative directly we therefore also explore the effect of the bandwidth choice on the performance of the pointwise CIs. We first describe a procedure for selecting the bandwidth and next examine the quality of (a) a bootstrap based estimate of the bias, (b) a direct estimate of the bias using an estimate of $f_0'$ and (c) undersmoothing the bandwidth on the reduction of the bias effect present in the pointwise CIs.

\subsubsection{Bandwidth selection }
\label{subsec:bandwidth}
In the previous simulation study, we considered taking the bandwidth equal to $h=2n^{-1/5}$, where the factor $2$ is based on the size of the support $[0,2]$ of the density $f_0$. This choice gave satisfactory results on the performance of the CIs discussed above. A bad choice of the bandwidth can however seriously affect the performance of the SMLE. It is therefore advisable to use an approach that selects the bandwidth with respect to some optimization criteria.  We apply the method proposed in \cite{hall:90} to select the bandwidth which uses bootstrap subsamples of smaller size from the original sample to estimate the pointwise mean squared error (MSE) of the SMLE.  The method works as follows: to obtain an approximation to the optimal bandwidth minimizing the pointwise MSE, we generate $B$ bootstrap subsamples of size  $m = o(n)$ from the original sample using the subsampling principle and take $c_{t, opt}$ as the minimizer of
\begin{align}
\label{MSE}
\widehat{MSE}(c) = B^{-1}\sum_{b=1}^{B}\left\{\tilde F_{m,cm^{-1/5}}^b(t) - \tilde F_{n,c_0n^{-1/5}}(t) \right\}^2,
\end{align}
where $\tilde F_{n,c_0n^{-1/5}}$ is the SMLE in the original sample of size $n$ using an initial bandwidth $c_0n^{-1/5}$ for some constant $c_0$. The bandwidth used for estimating the SMLE is next given by $h =c_{t,opt}n^{-1/5}$ where $c_{t, opt}$ minimizes $\widehat{MSE}(c)$ as a function of $c$. In the simulation study below we show the results for $m=50$ when generating subsamples from a sample of size $n=1,000$. Other subsample sizes $m=30,100$ were considered as well which resulted in similar optimal bandwidth choices. We used subsamples $m=100$ reps. $m =250$ when we generated data sets of size $n =5,000$ resp. $n =10,000$ from the model.

\subsubsection{Simulation study 2: correcting the asymptotic bias}
\label{subsec:Simulation_SMLE2}	
To investigate the effect of the bias on the construction of the pointwise CIs in (\ref{CI_type2}), we consider a second simulation study where the event times are generated from a truncated exponential distribution on $[0,2]$ and the censoring times are uniformly distributed on $[0,2]$. The density of the event times is given by $f_0(t) = \exp(-t)/(1-\exp(-2))1_{[0,2]}(t)$ and therefore the bias $\b(t)$ defined in (\ref{mu-sigma}) will influence the performance of the CIs. 

Figure \ref{fig:exponential_h1} compares the proportion of times that $F_0(t)$ is not in the 95\% bootstrap CIs for $t = 0.02,0.04,\ldots, 2$ with the corresponding proportions in the bias corrected CIs given by
\begin{align}
\label{CI_bias}
[\tilde F_{nh}(t)-Q_{1-\a/2}^*(t)\sqrt{S_{nh}(t)} &- \b(t)n^{-2/5}, \nonumber\\
&\tilde F_{nh}(t)-Q_{\a/2}^*(t)\sqrt{S_{nh}(t)} - \b(t)n^{-2/5}],
\end{align}
where $Q_{1-\a/2}^*(t)$ and $S_{nh}(t)$ are defined above and where $\b(t)$ is the true bias of the SMLE at timepoint $t$ defined in (\ref{mu-sigma}). The bandwidth of the SMLE is selected by the procedure described in Section \ref{subsec:bandwidth}.
 The coverage proportions of the uncorrected CIs are clearly smaller than the nominal 95\%-level at the left endpoint of the interval $[0,2]$ in correspondence to the region where $\b(t)$ is largest and correcting for the bias effect is needed to obtain good CIs. Figure \ref{fig:exponential_h1} suggests that the coverage proportions of the intervals will be satisfying if the bias can be estimated sufficiently accurately.  
\begin{figure}[!ht]
	\centering
	\includegraphics[width=0.6\textwidth]{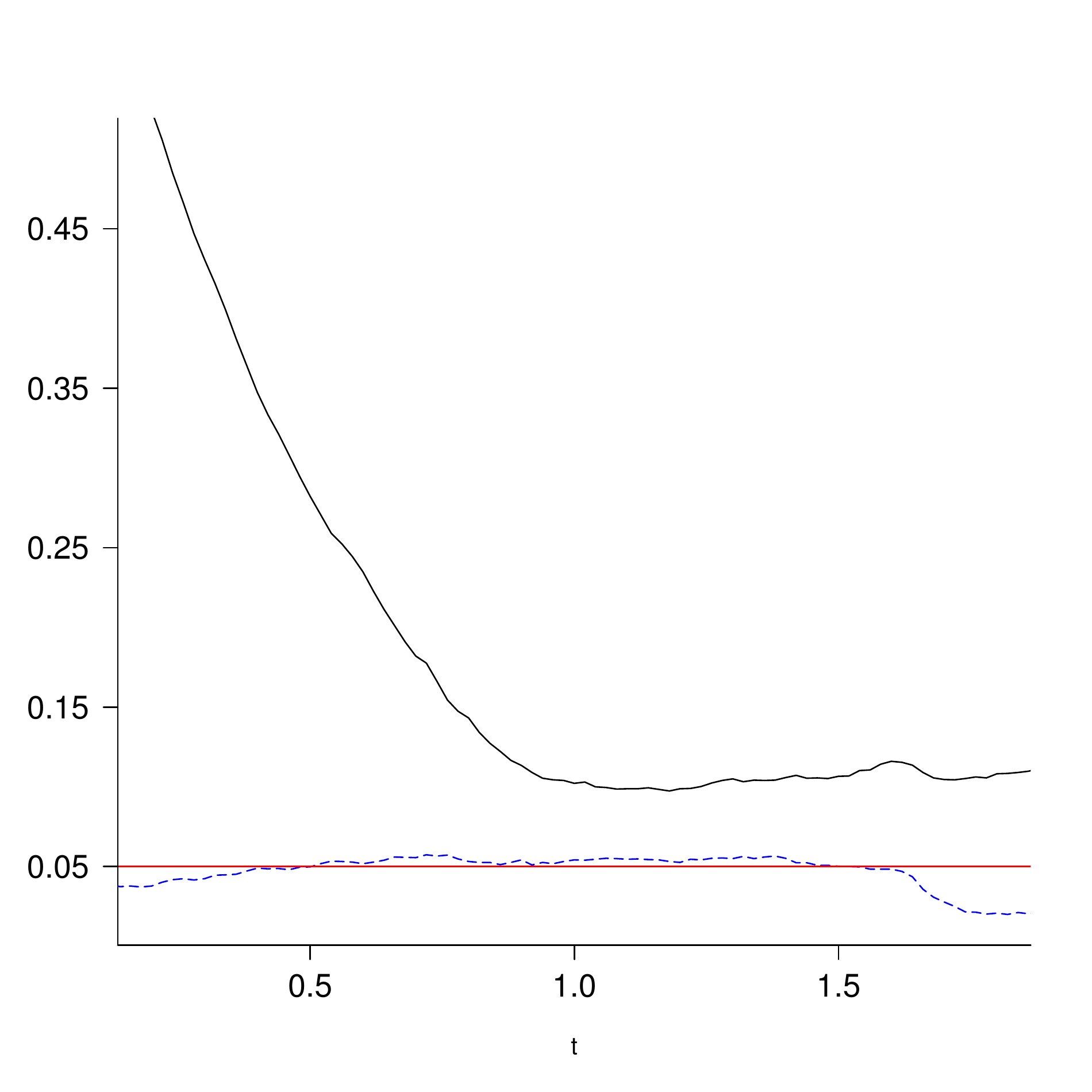}
	\caption{\small Exponential samples: Proportion of times that $F_0(t),\, t=0.02,0.04,\dots$ is not in the $95\%$ CIs for the  bootstrap CIs defined in (\ref{CI_type2}) (blue, dashed) and the bias corrected CIs defined in (\ref{CI_bias}) (black, solid). $n =1,000$, $N=5,000$,  $B=1,000, m=50$ and $h = c_{t,opt}n^{-1/5}$. }
	\label{fig:exponential_h1}
\end{figure}
Estimation of the bias requires estimating the density $f_0$, which is a rather difficult task with current status data. A kernel based estimate of $f_0'$ using the MLE $F_n$ is given by
\begin{align}
\label{fprime}
\tilde f_{n\bar h}'(t)=\bar h^{-2}\int K'\left((t-x)/\bar h\right)\,d F_n(x),
\end{align}
where the bandwidth $\bar h \sim n^{-1/9}$. In our experiments, we take the bandwidth of the estimate $\tilde f_{n\bar h}'(t)$ equal to $\bar h= \bar c_{t, opt}n^{-1/9}$ where $\bar c_{t, opt}$ is selected by the same bootstrap-MSE approach discussed in Section \ref{subsec:bandwidth}, but with the SMLE replaced by this derivative estimate. 
To obtain good estimates of $f_0'$ near the boundaries of the support, we consider the boundary correction method explained in Section 9.2 of \cite{piet_geurt:14}.  A direct estimator of the actual bias is then obtained by first replacing $f_0'(t)$ in (\ref{mu-sigma})  by the estimate $\tilde f_{nh}'(t)$ and next multiplying with $n^{-2/5}$, i.e. the order of the actual bias that has to be taken into account when constructing the CIs.

Similarly to the estimate of the pointwise MSE defined in (\ref{MSE}), we can also construct a bootstrap method for estimating the bias by using the subsampling principle described in \cite{hall:90}. Our estimate $\widehat{Bias}(t)$ of the actual bias $\b(t)n^{-2/5}$, is given by
 \begin{align*}
 \widehat{Bias}(t) = \left\{ B^{-1}\sum_{b=1}^{B}\left\{\tilde F_{m,c_{t,opt}m^{-1/5}}^b(t) - \tilde F_{nc_0n^{-1/5}}(t) \right\}\right\}\left(\frac{m}{n}\right)^{2/5}.
 \end{align*}
Figure \ref{fig:bias_comparisons} compares the average true bias effect $\b(t)n^{-2/5}$ and the average bias estimates obtained by either the direct estimation approach or the bootstrap based bias estimate for sample sizes $n=1000,5000$ and $n = 10,000$. 
Note that, since the bandwidth constant $c_{t,opt}$ used for estimating the SMLE is different in each simulation run, the true bias (depending on $c_{t,opt}$, see (\ref{mu-sigma})) in each run is also different and therefore the average true bias is shown in Figure \ref{fig:bias_comparisons}.  The actual size of the bias decreases with increasing sample size. 
\begin{figure}[!ht]
	\centering
	\begin{subfigure}[b]{0.31\textwidth}
		\includegraphics[width=\textwidth]{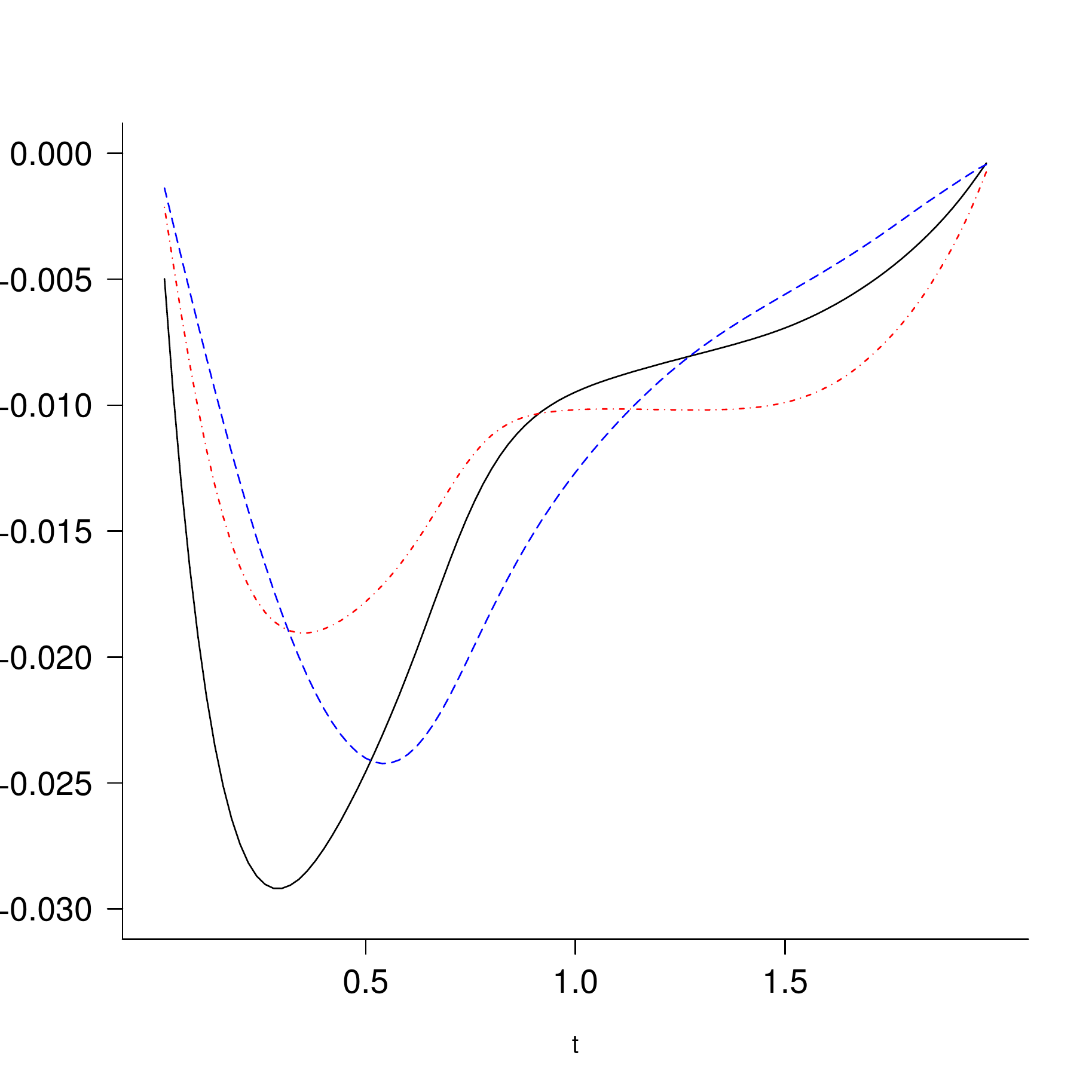}
		\caption{}
	\end{subfigure}
	\begin{subfigure}[b]{0.31\textwidth}
		\includegraphics[width=\textwidth]{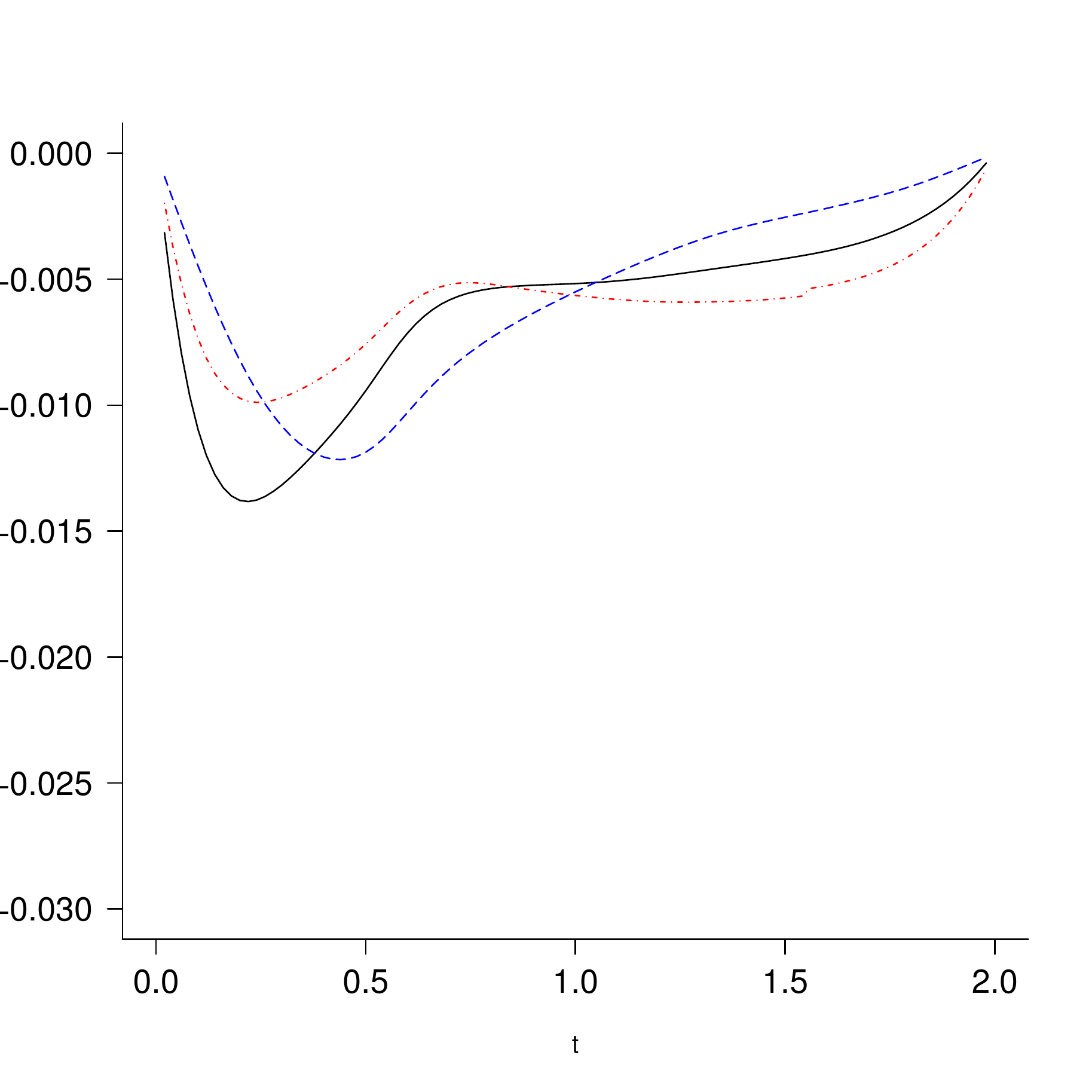}
		\caption{}
	\end{subfigure}
	\begin{subfigure}[b]{0.31\textwidth}
		\includegraphics[width=\textwidth]{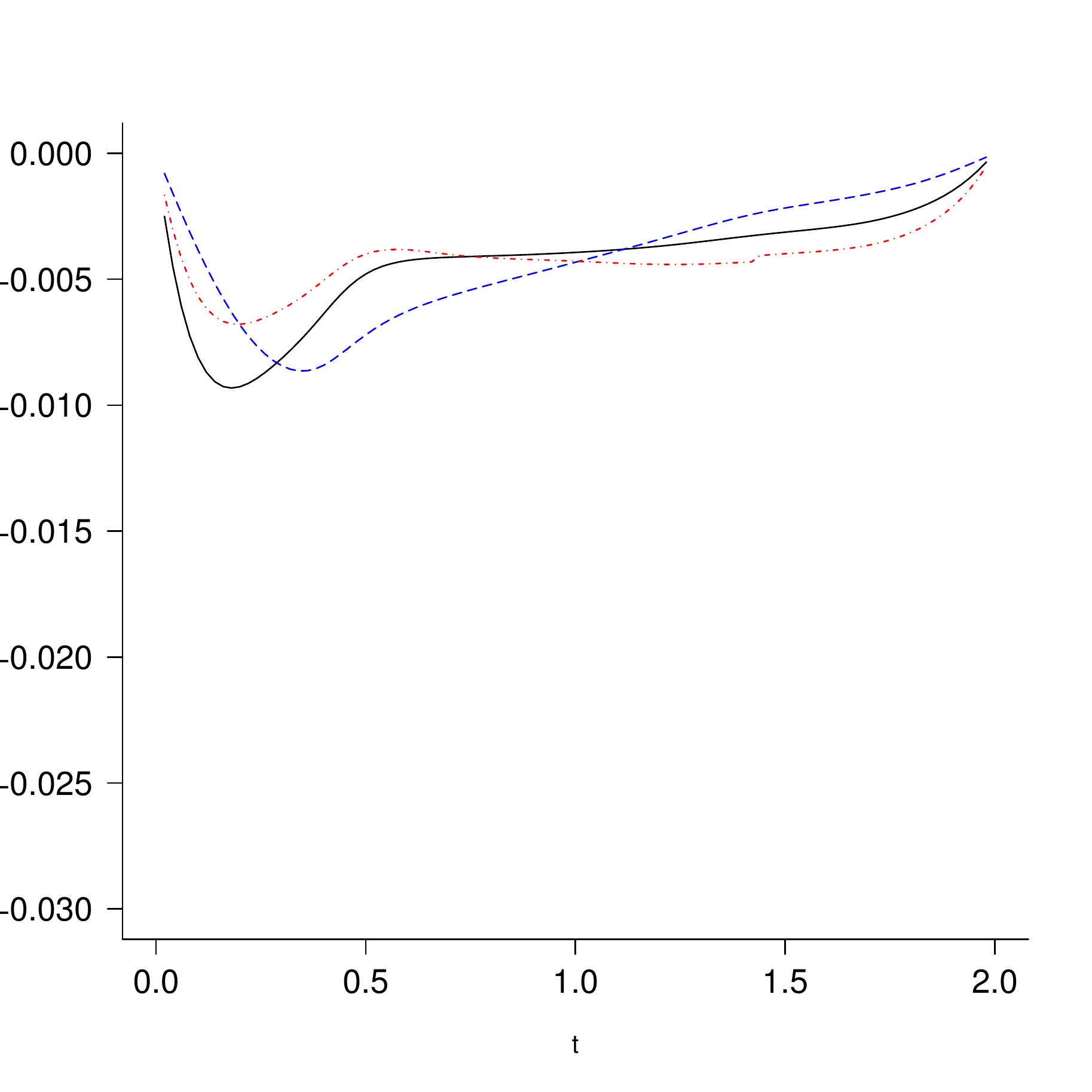}
		\caption{}
	\end{subfigure}
	\caption{\small Exponential samples: Average true bias (black solid) and average estimated bias for the  bootstrap based estimate (blue, dashed) and the direct estimate (red, dashed-dotted) for samples (and subsamples) of size (a) $n =1,000, m =50$, (b) $n =5,000, m=100$ and (c) $n =10,000, m=250$.  $N=5,000$,  $B=1,000$ and $h = c_{t,opt}n^{-1/5}$. }
	\label{fig:bias_comparisons}
\end{figure}

The proportion of times that $F_0(t)$ is not in the 95\% bootstrap CIs, shown in Figure \ref{fig:bias_percentages},  decreases if one corrects for the bias by one of the discussed bias estimates. The results for the direct bias estimate using the estimate $\tilde f_{n\bar h}'$ are slightly better than the results for the bootstrap estimate of $\b(t) n^{-2/5}$. 
The coverage proportions are however still anti-conservative for points at the left end of the support.
We also considered constructing the bias corrected CIs in the uniform model used in Section \ref{subsec:simulation_SMLE} where the actual bias is zero (results not shown). The results of the uncorrected CIs in (\ref{CI_type2}) were slightly better and estimating the bias in this model has a somewhat negative effect on the coverage proportions of the pointwise CIs around the SMLE.
\begin{figure}[!ht]
	\centering
	\begin{subfigure}[b]{0.31\textwidth}
		\includegraphics[width=\textwidth]{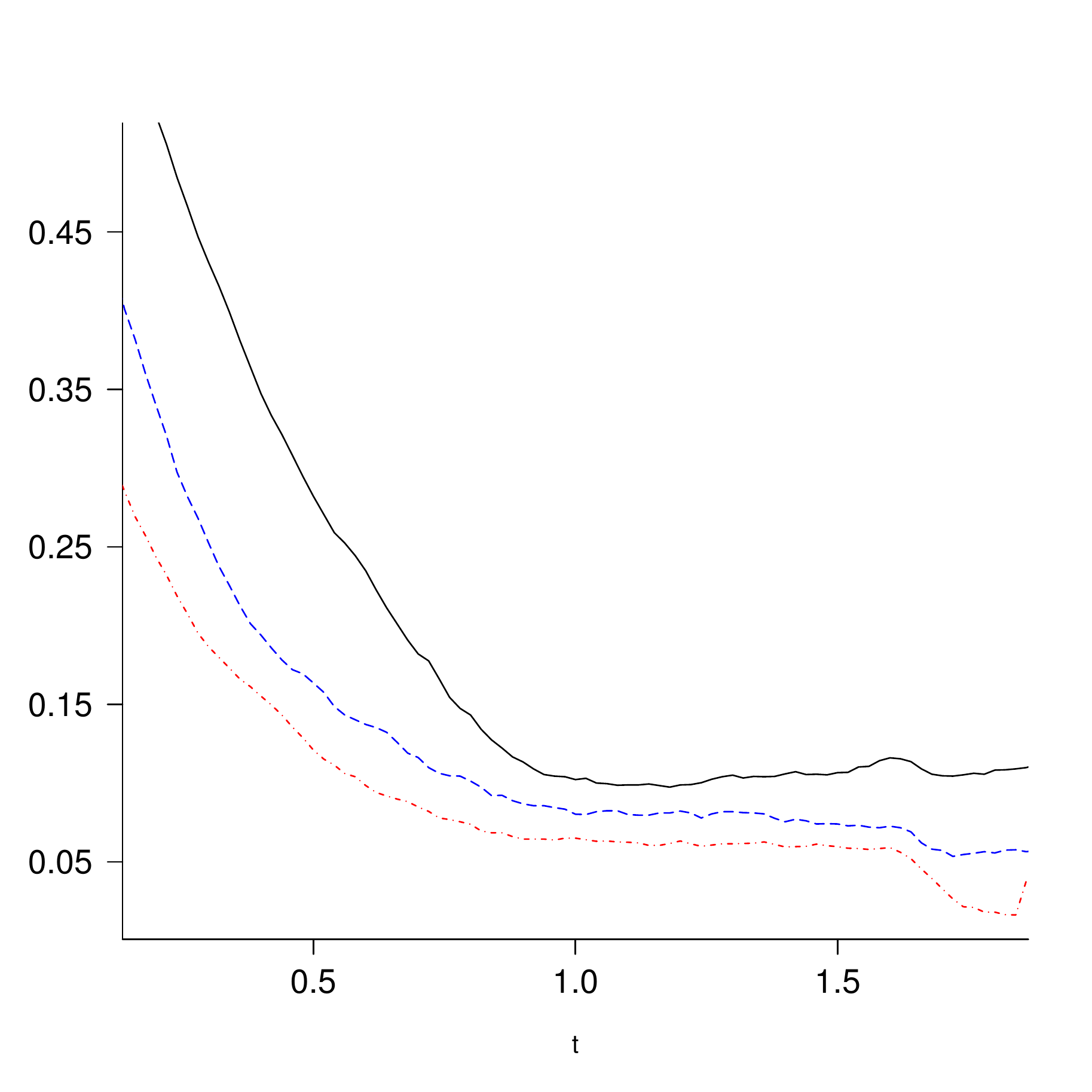}
		\caption{}
	\end{subfigure}
	\begin{subfigure}[b]{0.31\textwidth}
		\includegraphics[width=\textwidth]{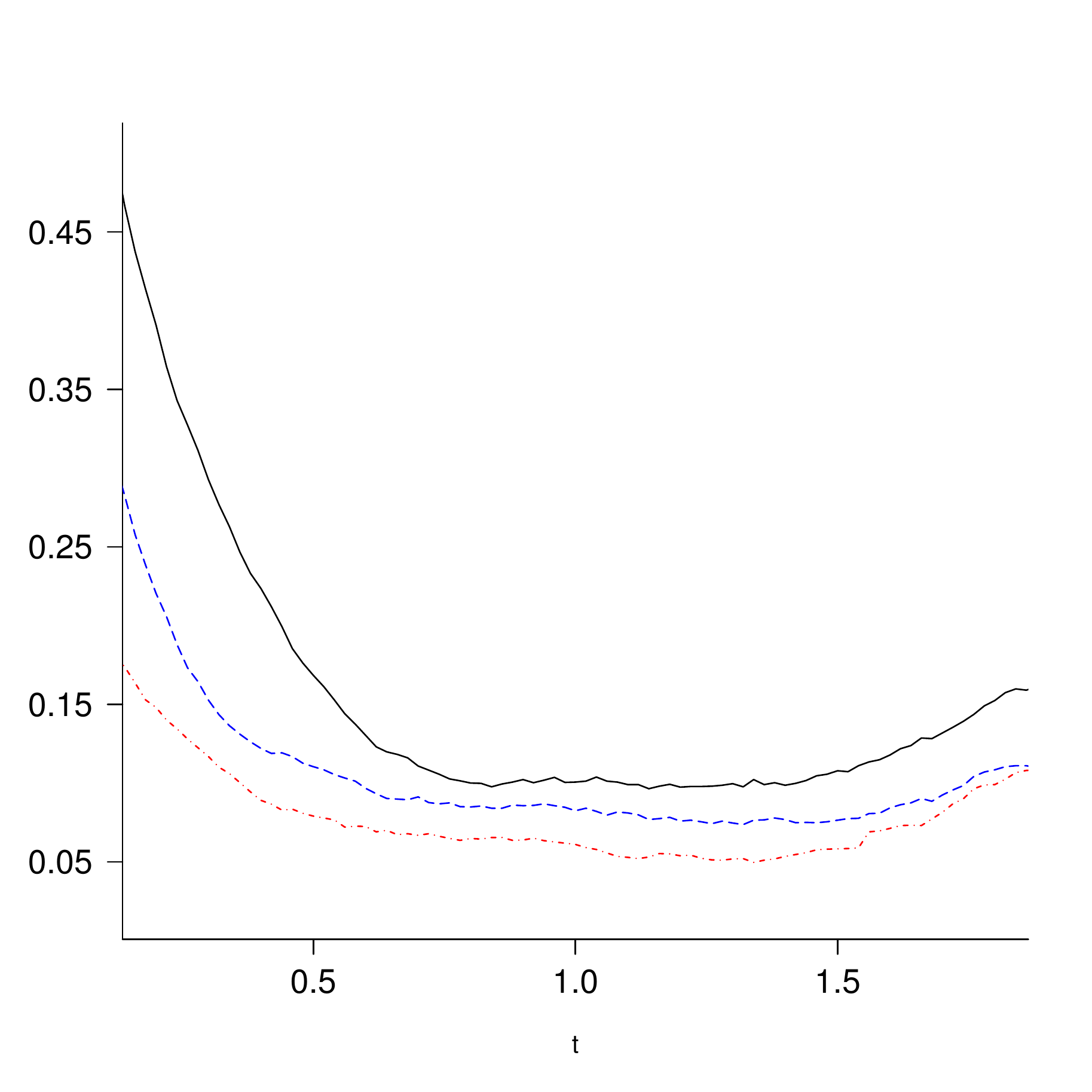}
		\caption{}
	\end{subfigure}
	\begin{subfigure}[b]{0.31\textwidth}
		\includegraphics[width=\textwidth]{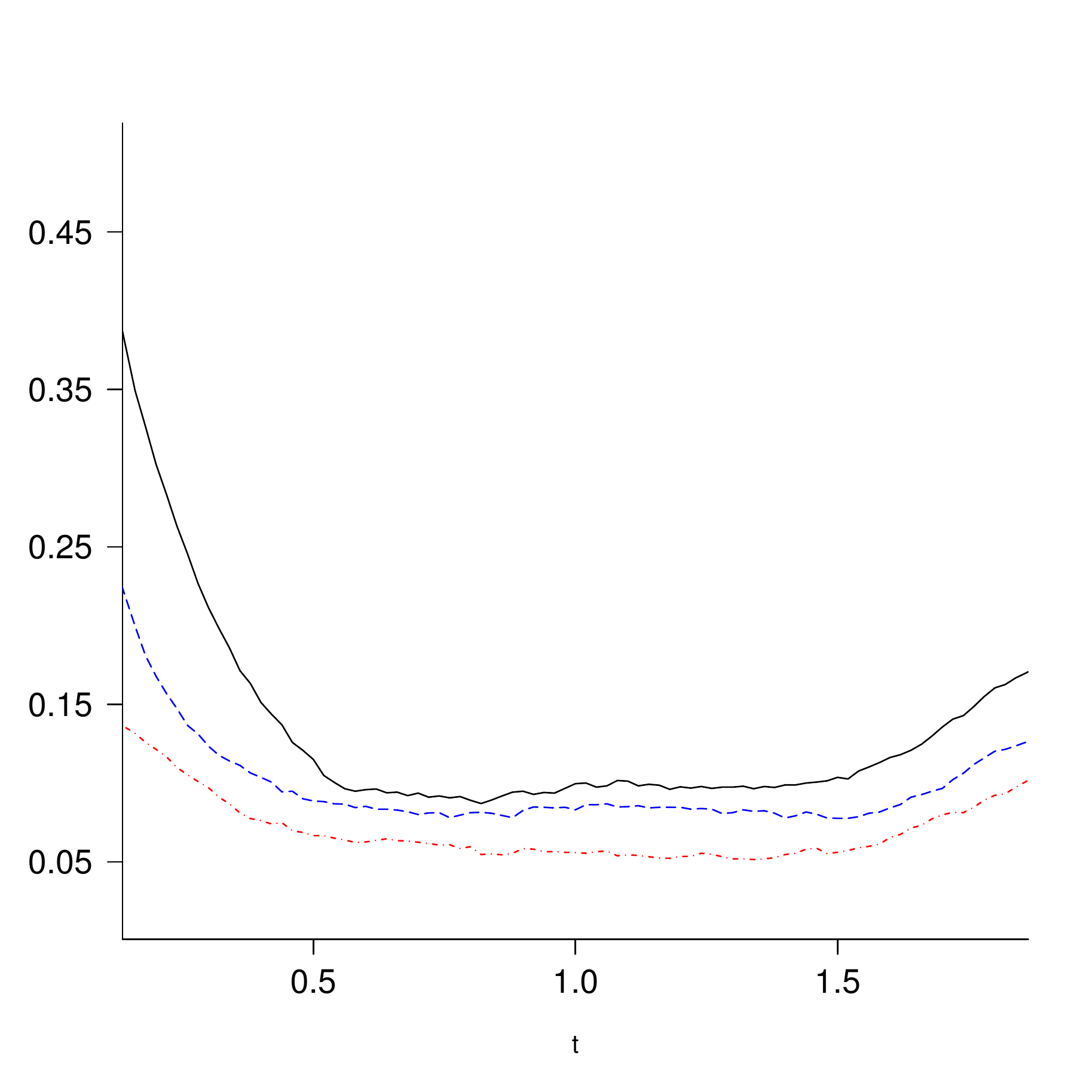}
		\caption{}
	\end{subfigure}
	\caption{\small Exponential samples without undersmoothing: Proportion of times that $F_0(t),\, t=0.02,0.04,\dots$ is not in the $95\%$ CIs defined in (\ref{CI_type2}) (black, solid) and the bias corrected confidence intervals defined in (\ref{CI_bias}) with bootstrap based bias estimate (blue,dashed) and direct bias estimate (red, dashed-dotted) for samples (and subsamples) of size (a) $n =1,000, m =50$, (b) $n =5,000, m=100$ and (c) $n =10,000, m=250$. $N=5,000$,  $B=1,000$ and $h = c_{t,opt}n^{-1/5}$. }
	\label{fig:bias_percentages}
\end{figure}

Similarly to the methods proposed in \cite{kim_piet:17:SJS} we next investigate how the choice of the bandwidth can affect the coverage proportions and average length of our CIs.  To this end, we consider the concept of undersmoothing proposed by \cite{hall:92} and take $c_{t,opt}n^{-1/4}$ as the bandwidth used in constructing the CIs defined in (\ref{CI_type2}). The coverage proportions of the CIs for the exponential model, shown in Figure \ref{fig:undersmoothing}, illustrate that the performance of the CIs around the SMLE improve by undersmoothing. 
We also observed that if we considered a smaller bandwidth choice  $h = (1/3) c_{t,opt}n^{-1/5}$ , the coverage proportions even improve further and give satisfactory results in the left end point of the support.
This illustrates that a smaller bandwidth choice can indeed correct for the bias in the CIs.

The results of the CIs in (\ref{CI_type2}) in the uniform model with a bandwidth $h=c_{t,opt}n^{-1/4}$ or $h=(1/3)c_{t,opt}n^{-1/5}$ are in line with the results obtained with a bandwidth $h=c_{t,opt}n^{-1/5}$ and similar to the results shown in Figure \ref{fig:SMLE_wald}. This shows that undersmoothing in a model without bias has no negative effect on the coverage proportions of our CIs.

By undersmoothing, the length of our SMLE-based CIs increases but 
the average length of the CIs remains remarkably smaller than the average length of the CIs around the MLE proposed by \cite{banerjee_wellner:2005} and \cite{SenXu2015} (see Table \ref{table:AL}).

\begin{figure}[!ht]
	\centering
	\begin{subfigure}[b]{0.31\textwidth}
		\includegraphics[width=\textwidth]{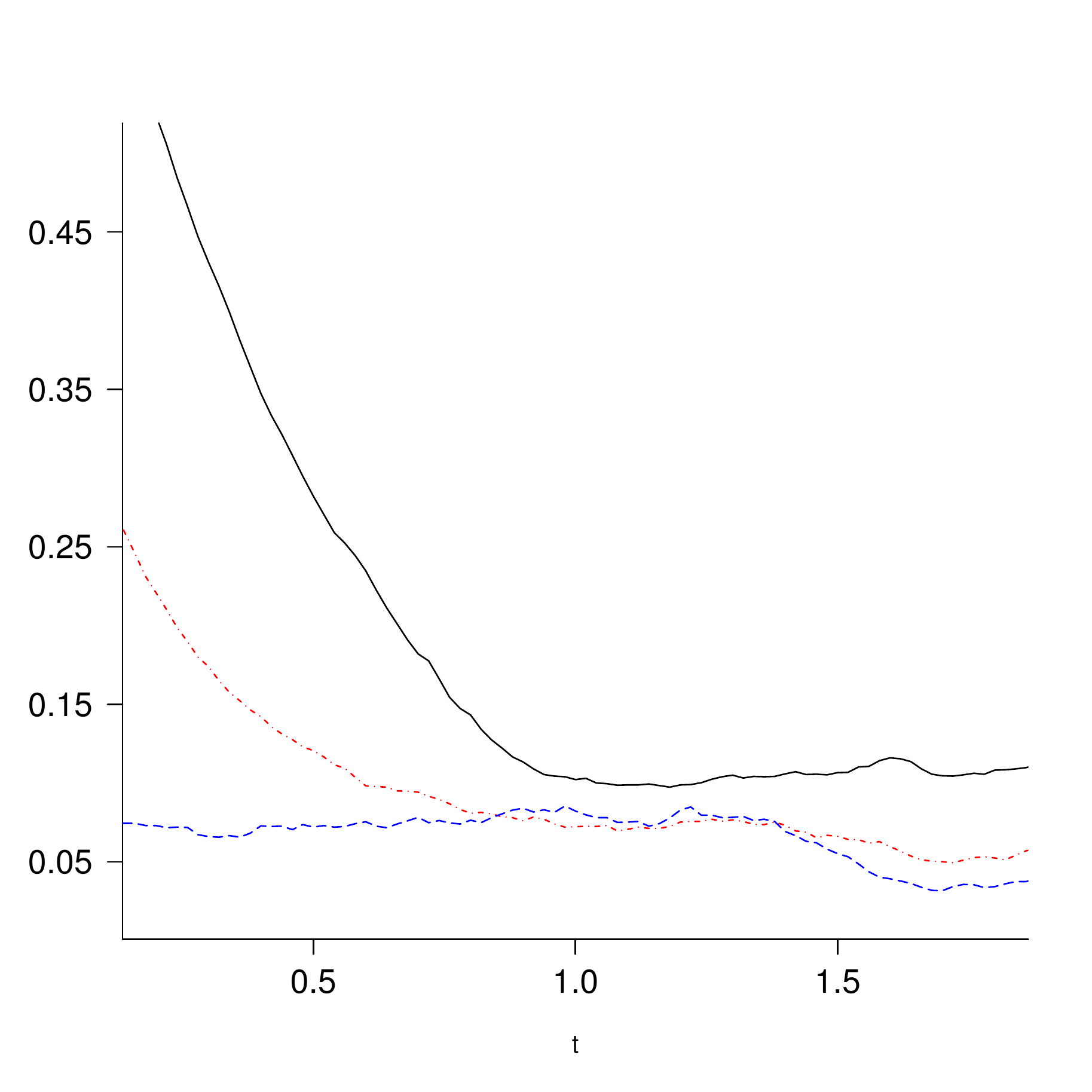}
		\caption{}
	\end{subfigure}
	\begin{subfigure}[b]{0.31\textwidth}
		\includegraphics[width=\textwidth]{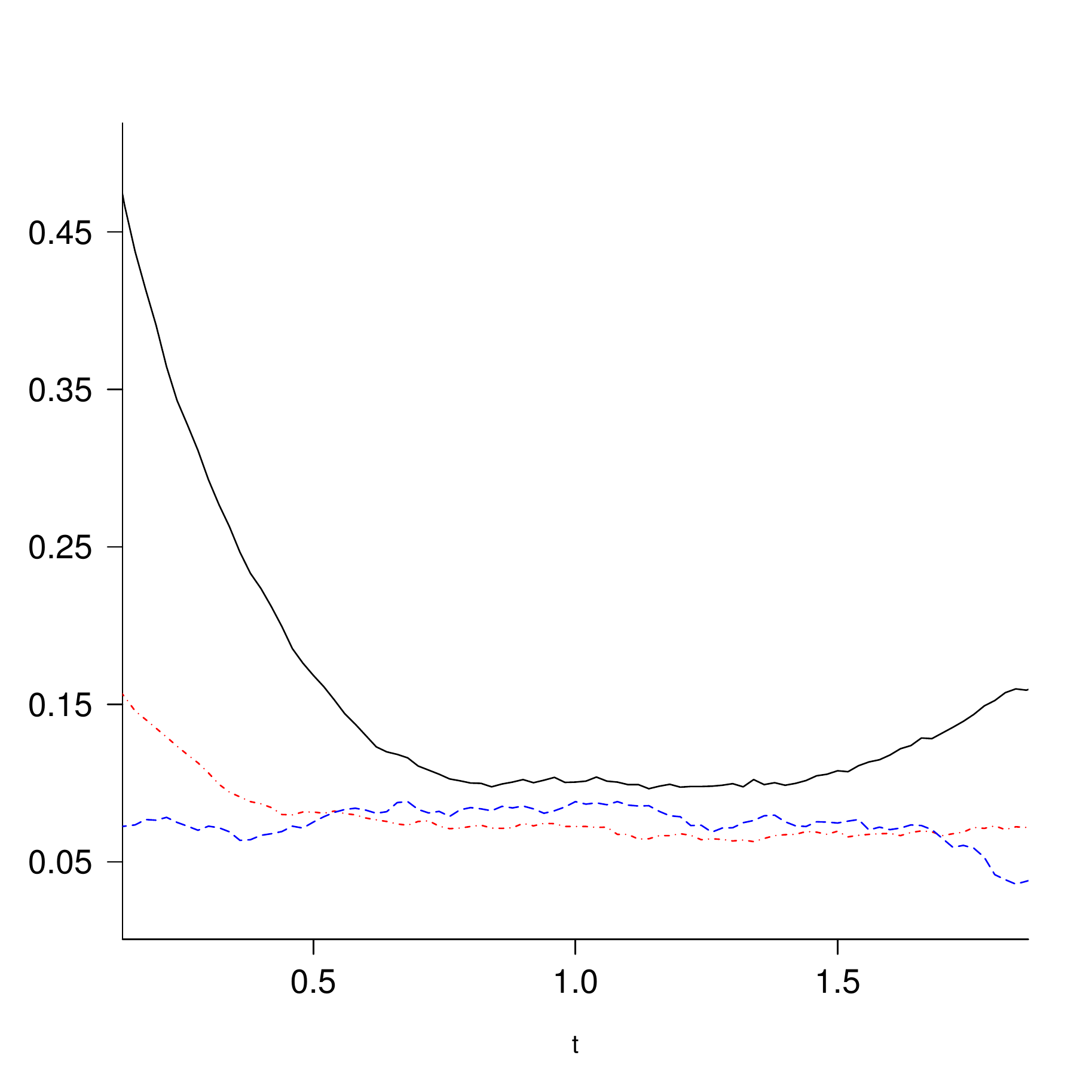}
		\caption{}
	\end{subfigure}
	\begin{subfigure}[b]{0.31\textwidth}
		\includegraphics[width=\textwidth]{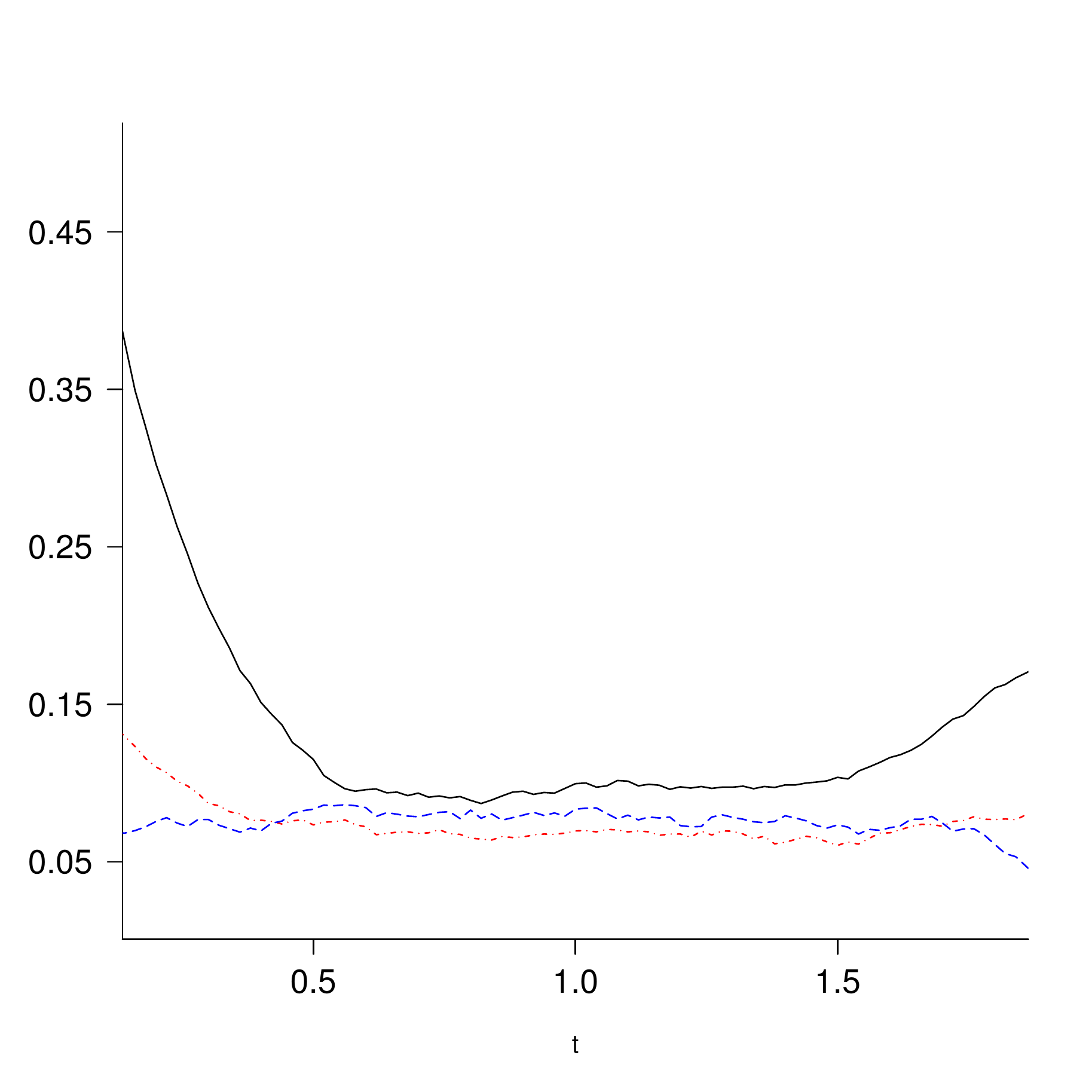}
		\caption{}
	\end{subfigure}
	\caption{\small Proportion of times that $F_0(t),\, t=0.02,0.04,\dots$ is not in the $95\%$ CIs defined in (\ref{CI_type2}) with $h = c_{t,opt}n^{-1/5}$ (black, solid), $h = c_{t,opt}n^{-1/4}$ (red, dashed-dotted)  and $h = (1/3) c_{t,opt}n^{-1/5}$  (blue, dashed) for samples (and subsamples) of size (a) $n =1,000, m =50$, (b) $n =5,000, m=100$ and (c) $n =10,000, m=250$. $N=5,000$ and  $B=1,000$. }
	\label{fig:undersmoothing}
\end{figure}

\begin{table}
\centering
\caption{Average length of the SMLE-based CIs defined in (\ref{CI_bias}) for different bandwidth choices ($h \sim n^{-1/5}$ and $h \sim n^{-1/4}$) and average length of the MLE-based CIs proposed by \cite{banerjee_wellner:2005} and \cite{SenXu2015} at timepoints $t = 0.5,1,1.5$.}
\label{table:AL}
\begin{tabular}{c|ccc|ccc}
\hline
 & \multicolumn{3}{c|}{Uniform} & \multicolumn{3}{c}{Exponential}\\
Method  & $t = 0.5$ & $t = 1$ & $t = 1.5$ & $t = 0.5$ & $t = 1$ & $t = 1.5$ \\
\hline
SMLE ($h \sim n^{-1/5}$) & 0.064819 & 0.077020 & 0.064976 & 0.085540 & 0.087565 & 0.057716 \\
SMLE ($h \sim n^{-1/4}$) & 0.079671 & 0.092096 &  0.079757 & 0.085540 & 0.087565 & 0.057716 \\
MLE (\cite{banerjee_wellner:2005}) & 0.164767 &0.184590& 0.165699 &0.204079 &0.161122& 0.104002\\
MLE (\cite{SenXu2015}) & 0.183982 & 0.202430 & 0.186452 & 0.225882 & 0.176159 & 0.118541\\
\hline
\end{tabular}
\end{table}

\subsubsection{Rubella data}
We also applied the bootstrap procedures to the Rubella data set described by \cite{Keiding:96}. The data set contains 230 observations on the prevalence of rubella in Austrian males. For the smooth bootstrap, confidence intervals were calculated in \cite{kim_piet:17:SJS} using the bandwidth $h = c_{t,opt}n^{-1/4}$.  Figure \ref{fig:Rubella}  shows the CIs obtained with the nonparametric bootstrap and illustrates the applicability of our method in a real data example. For comparison, we also show the confidence intervals obtained by the methods of \cite{banerjee_wellner:2005} and \cite{SenXu2015}. The latter confidence intervals were obtained by the Rcpp scripts in \cite{github:15}. The nonparametric bootstrap SMLE-based CIs, including the data-driven bandwidth procedure, can be generated with the R package \texttt{curstatCI}.

\begin{figure}[!ht]
	\centering
	\begin{subfigure}[b]{0.3\textwidth}
		\includegraphics[width=\textwidth]{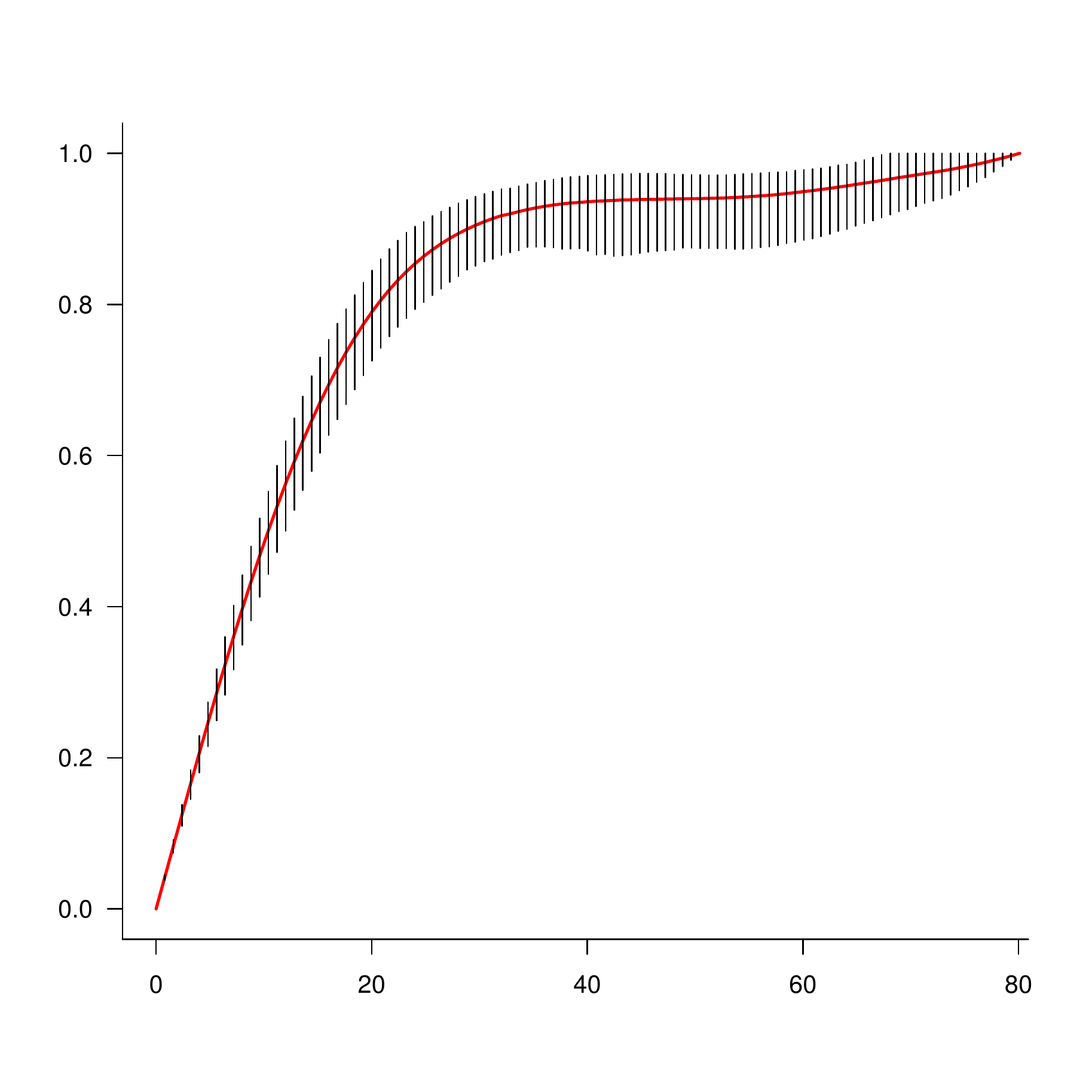}
		\caption{}
	\end{subfigure}
	\begin{subfigure}[b]{0.3\textwidth}
		\includegraphics[width=\textwidth]{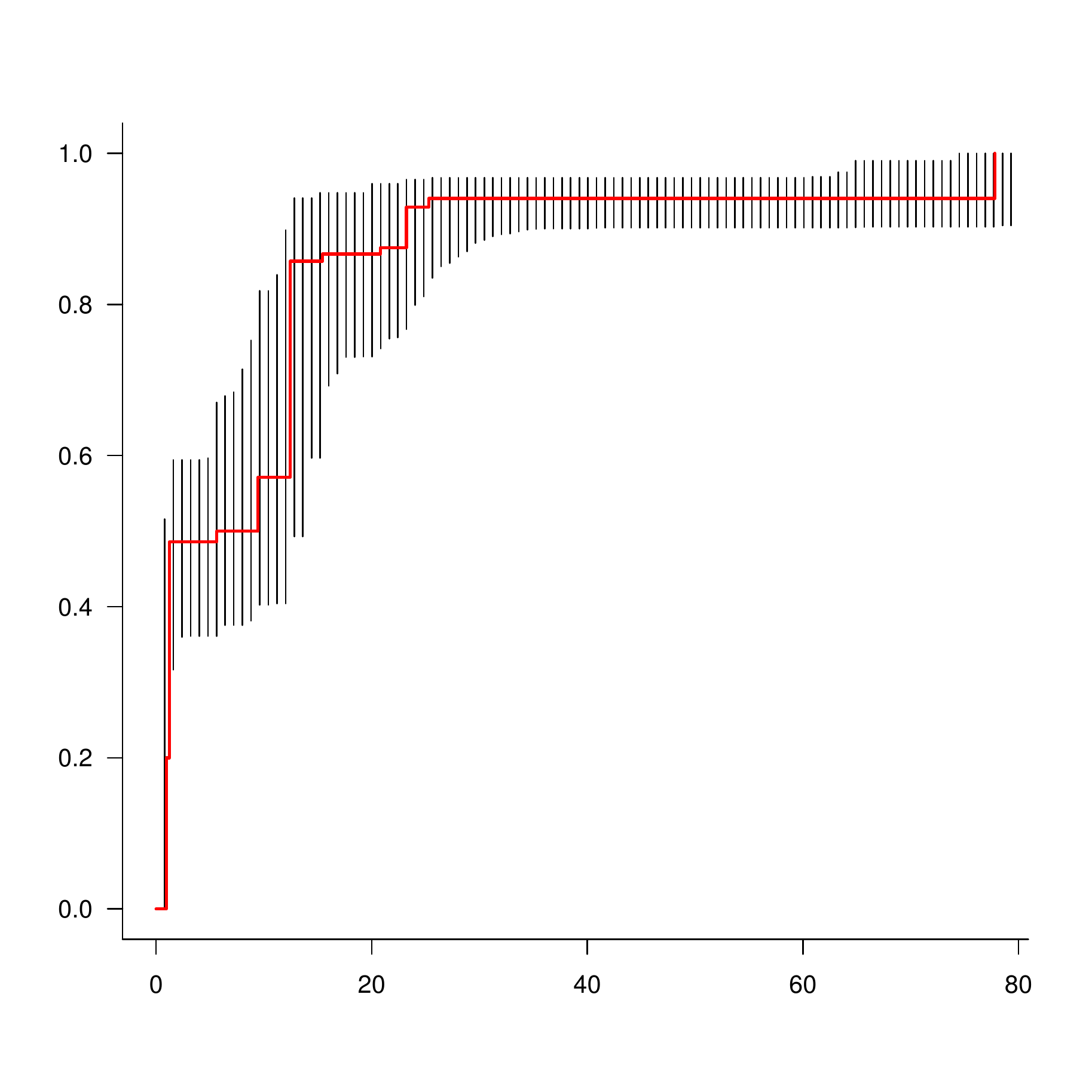}
		\caption{}
	\end{subfigure}	
	\begin{subfigure}[b]{0.3\textwidth}
		\includegraphics[width=\textwidth]{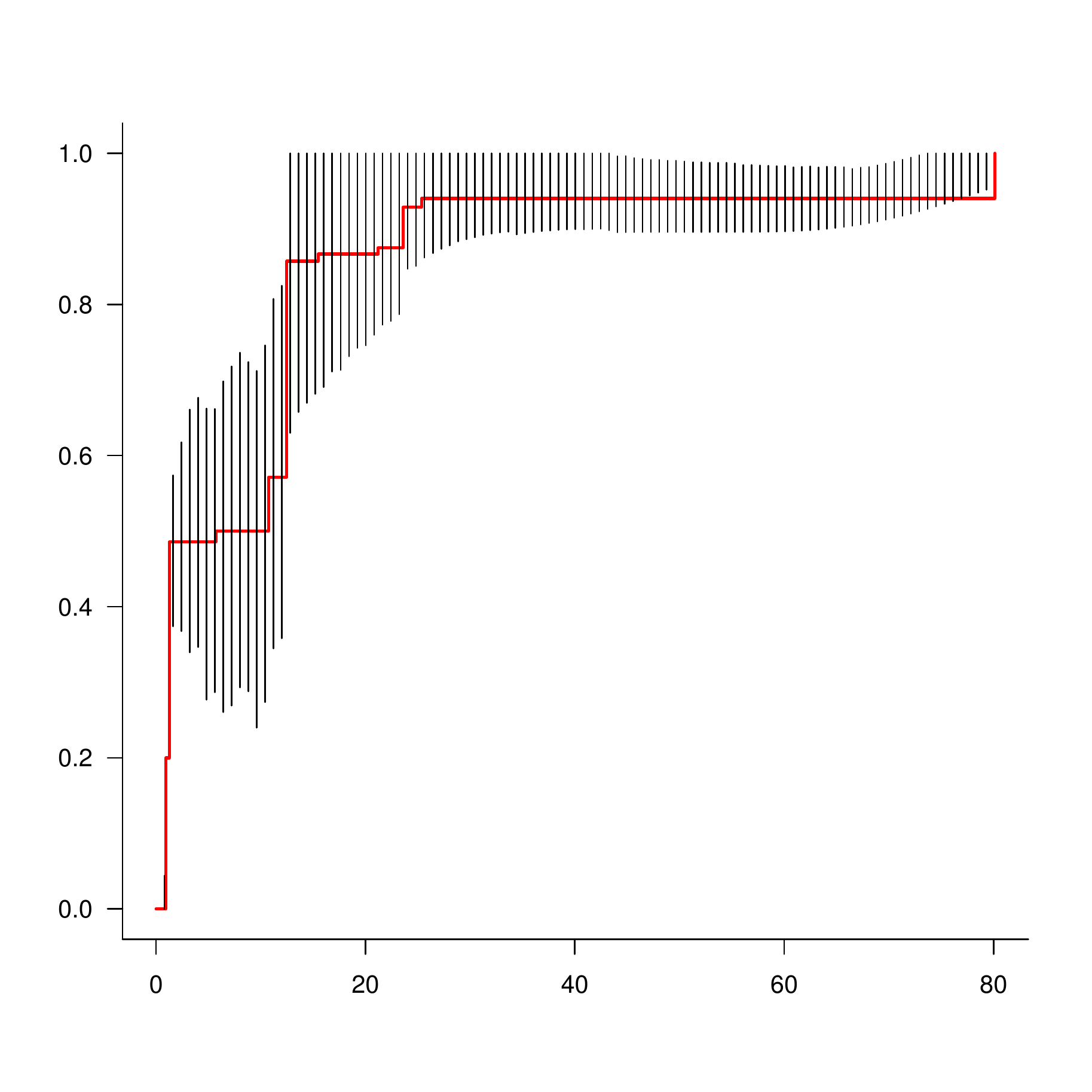}
		\caption{}
	\end{subfigure}	
	\caption{\small Rubella data: (a) SMLE (red, solid) and CI defined in (\ref{CI_type2}) based on  $n=230$ observations using $B=1,000$ bootstrap samples and bandwidth $h = c_{t,opt}n^{-1/4}$ ($c_{t,opt}$ obtained by subsampling with $B=1,000$ bootstrap samples of smaller size $m = 50$). (b) MLE (red, solid) and CI obtained by the method of Banerjee and Wellner \cite{banerjee_wellner:2005}, (c) MLE (red, solid) and CI obtained by the method of Sen and Xu \cite{SenXu2015} with $B=1,000$ `smooth' bootstrap samples from the SMLE with bandwidth $h=80n^{-1/5}$.}
	\label{fig:Rubella} 
\end{figure}

\subsection{The current status linear regression model}
\label{sec:applications_CSLR}
In the current status linear regression model we are interested in the estimation of the regression parameter $\b_0$ based on observations $(T_1,X_1 , \dd_1= 1_{\{Y_1\le T_1\}}),\ldots (T_n,X_n, \dd_n= 1_{\{Y_n\le T_n\}}) $ from $(T,X,\dd)$ where we assume that 
\begin{align*}
Y_i = \b _0'X_i + \varepsilon_i \qquad i=1,2,\ldots
\end{align*}
with i.i.d. random error terms $\e_i$, independent of $(T_i,X_i)$ with unknown distribution function $F_0$. 

In \cite{kim_piet:17:AOS} a simple score estimator $\b_n$ was introduced depending on the MLE $F_{n,\b}$ for fixed $\b$, defined as,
\begin{align} 
F_{n,\b} \stackrel{def}{=} \arg \max_{F\in {\cal F}} \sum_{i=1}^n\left[\dd_i\log F(T_i-\b' X_i)+(1-\dd_i)\log\{1-F(T_i- \b' X_i)\}\right],
\end{align}
where ${\cal F } = \{F: \Re \mapsto [0,1]: \text{ F is a distribution function} \}$. The estimator  $\b_n$ for $\b_0$ is next defined as a zero-crossing (see Definition 4.1 in  \cite{kim_piet:17:AOS}) of
\begin{align}
\label{score0}
\sum_{ F_{n,\b}(T_i-\b'X_i)\in[\ee,1-\ee]} X_i\bigl\{ \dd_i -F_{n,\b}(T_i-\b' X_i)\bigr\},
\end{align}
for some  fixed truncation parameter $\ee \in (0,1/2)$. It is proved in \cite{kim_piet:17:AOS} that  $\sqrt{n}\bigl\{\b_n-\b_0\bigr\}$ is asymptotically normal with mean zero and variance $V^{-1}WV^{-1}$ where
\begin{align*}
&V=E_{\ee}\Bigl[ f_0(T- \b_0'X)\,\left\{ X - E(X| T- \b_0'X) \right\}\left\{ X - E(X| T- \b_0'X) \right\}'\Bigr],\\
&W =E_{\ee}\Bigl[ F_0(T- \b_0'X)\{1- F_0(T- \b_0'X)\}\left\{ X - E(X| T- \b_0'X) \right\}\times
\\
&\qquad\qquad\qquad\qquad\qquad\qquad\qquad\qquad\qquad\qquad\qquad\left\{ X - E(X| T- \b_0'X) \right\}'\Bigr],\nonumber
\end{align*}
where $ E_\ee(w(T,X,\dd)) = \int_{F_0(t-\b_0'x) \in[\ee,1-\ee]} w(t,x,\d) \, dP(t,x,\d)$ is the truncated expectation of $w(T,X,\dd)$ for some deterministic function $w$ and where $P$ denotes the probability measure of $(T,X, \dd)$.

A bootstrap version $\hat \b_n$ based on a bootstrap sample from $\P_n$ is then defined as the zero-crossing of
\begin{align}
\sum_{\hat F_{n,\b}(T_i-\b' X_i)\in[\ee,1-\ee]}M_{ni} X_i\bigl\{\dd_i - \hat F_{n,\b}(T_i-\b' X_i)\bigr\}=0,
\end{align}
where $\hat F_{n,\b}$ is the MLE in the bootstrap sample. A straightforward extension of the results given in Section \ref{sec:bootstrap} shows that, as $n$ tends to infinity,
\begin{align*}
E_{M|Z}\left|n^{-1/3}\left\{ \hat F_{n,\b}(t-\b'x) - F_\b(t-\b'x)\right\}\right|^p , 
\end{align*}
stays bounded in probability for all $(t,x)\in \{(t,x): F_\b(t-\b'x)\in [\ee,1-\ee]\}$ and for all $\b$ in a neighborhood of $\b_0$  where $F_\b$ is defined by
\begin{align}
\label{def_F_beta}
F_{\b}(u)&=P\left\{\dd_i =1\bigm|T_i -\b' X_i=u\right\} = \int F_0(u+(\b-\b_0)'x)f_{X|T-\b'X}(x|u)\,dx.
\end{align} 
The validity of the bootstrap method follows from the fact that, in probability, we have conditionally on the data $(T_1,X_1,\dd_1),\dots,(T_n,X_n,\dd_n)$ that,
\begin{align}
\label{representation_bootstrap_beta_n}
&-\sqrt n V(\hat \b_n  - \b_n ) =  \sqrt n\int_{F_0(t-\b_0'x) \in[\ee,1-\ee]}\{x-E(X|T-\b_0' X= t-\b_0'x)\}\nonumber\\
&\quad\qquad\qquad\qquad\qquad\qquad\qquad\qquad\qquad\cdot\{\d -F_0(t-\b_0' x)\}\,d(\hat \P_n-\P_n)(t,x,\d) \nonumber\\
&\qquad\qquad\qquad\qquad+ o_{P_M}(1+ \sqrt n (\hat\b_n - \b_n)),
\end{align}
where the dominant term in the right-hand side of the display above is normally distributed with mean zero and variance $W$ conditional on $(T_1,X_1,\dd_1),\dots,$ $(T_n,X_n,\dd_n)$.

\begin{remark}
	{\rm The nonparametric bootstrap is also valid for the second estimator of $\b_0$ proposed in \cite{kim_piet:17:AOS} based on a different score function involving the MLE $\hat F_{n,\b}$ and the derivative of the SMLE $\tilde F_{nh,\b}$ (constructed by the procedure described in Section \ref{sec:SMLE}).}
\end{remark}

To provide more insight into the finite sample behavior of the classical bootstrap estimators we show in Tables \ref{table:beta_estimation1} and \ref{table:beta_estimation2} the results of two simulation studies for a one-dimensional regression model $Y = \b_0X + \e$. In the first simulation setting we take $\b_0=0.5$ and consider Uniform(0,2) distributions for the variables $T$ and $X$; for the distribution of the random error $\e$ we take $f_0(e) = 384(e-3/8)(5/8-e)1_{[3/8,5/8]}(e)$. A picture of the density and distribution function of the random error in model 1 is shown in Figure \ref{fig:model1}. The first model is also analyzed in \cite{kim_piet:17:AOS}. In the second simulation model $T,X$ and $\e$ are independently sampled from a standard normal distribution and $\b_0 =1$. A similar model was considered in \cite{Abrevaya:99}. 

\begin{figure}[!ht]
	\begin{subfigure}[b]{0.33\textwidth}
		\centering
		\includegraphics[width=\textwidth]{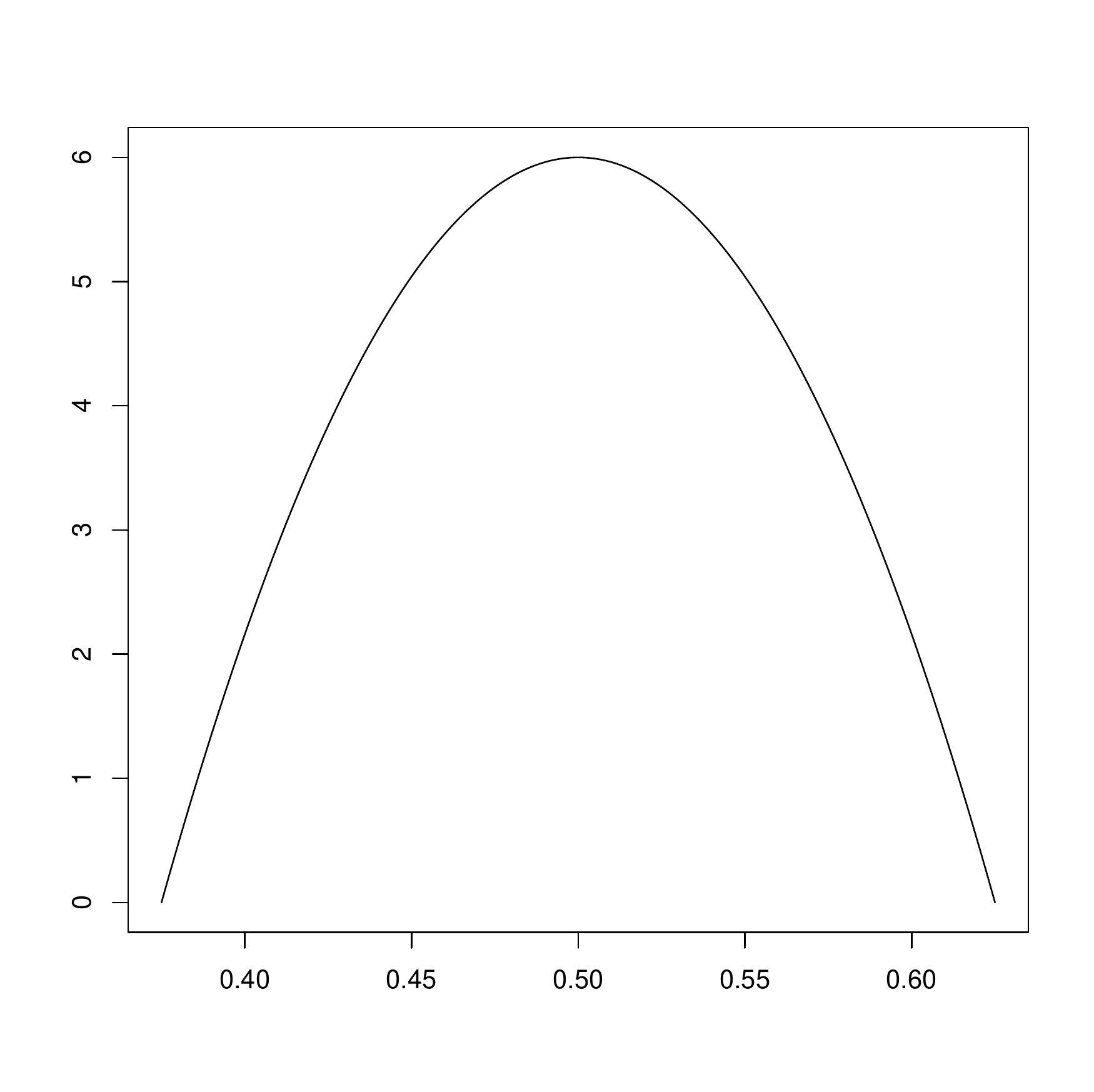}
	\end{subfigure}
	\hspace{0.05cm}
	\begin{subfigure}[b]{0.33\textwidth}
		\includegraphics[width=\textwidth]{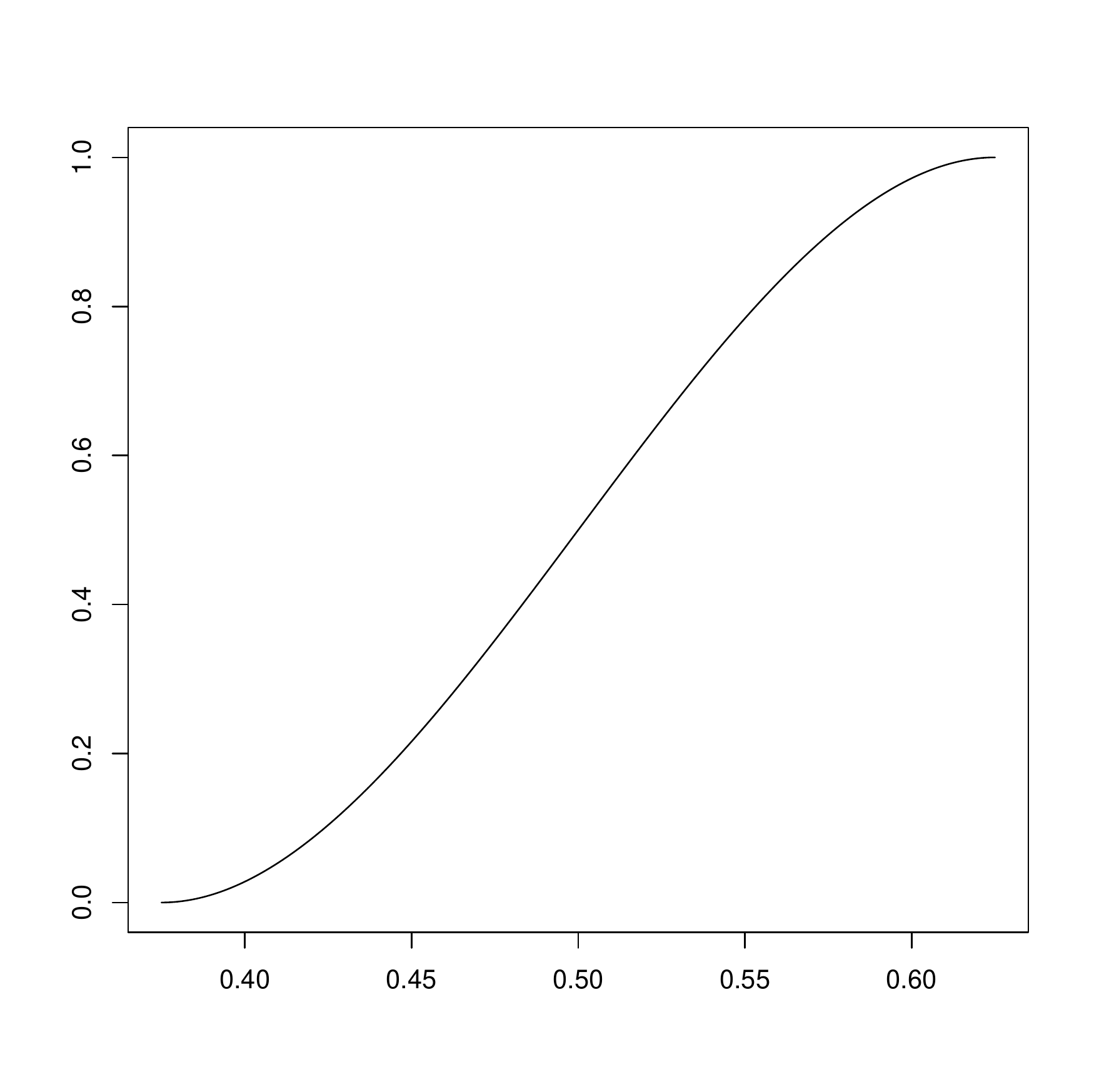}
	\end{subfigure}
	\caption{\small  The density $f_0$ (left panel) and distribution function $F_0$ (right panel) of the random error $\e$ in simulation model 1.  }
	\label{fig:model1}
\end{figure}

With these simulations we want to point out that it is not necessary to use smoothing techniques for doing inferences in the current status linear regression model. We compare the simple score estimator (SSE) described above with Han's maximum rank correlation estimator (\cite{han:87}, MRCE) and with the efficient score estimator (ESE) proposed in \cite{kim_piet:17:AOS}. The asymptotic behavior of the MRCE for the current status model, also obtained without any smoothing techniques, is established in \cite{Abrevaya:99} where the author also proposes consistent kernel-based estimates of the asymptotic variance of the MRCE. We use these variance estimates to construct estimates for $V, W$ and the almost (determined by the truncation parameter $\ee$) efficient variance of the SSE. For more details about the variance estimation we refer to \cite{Abrevaya:99}.

A summary of $N=1,000$ simulation runs from models 1 and 2 for different sample sizes $n$ is given in Tables \ref{table:beta_estimation1} and \ref{table:beta_estimation2}. For each estimator, the mean, $n$ times the variance and $n$ times the mean squared-error (MSE) is given in columns 3-5. The asymptotic variance of the estimators equals 0.193612 for the SSE,  0.158699 for the ESE and 0.192857 for the MRCE in model 1 using truncation parameter $\ee= 0.001$. The corresponding asymptotic variances in model 2 equal 5.046413, 4.994988  and 5.35448 respectively. The asymptotic variance of the SSE without truncation (i.e. $\ee = 0$) equals the asymptotic variance of the MRCE in model 1. The efficient variances are 0.151706 in model 1 and 4.994987 in model 2. Note that the differences between the limiting variances for the different estimation methods are tiny and that the effect of the truncation parameter $\ee$ on the asymptotic behavior of the score estimators is small. Tables \ref{table:beta_estimation1} and \ref{table:beta_estimation2} show that $n$ times the variance tends to converge to the asymptotic variance for all estimators. The ESE performs worse for small sample sizes and the results suggest to use the SSE for point estimation of the regression parameter $\b_0$. 

We constructed Wald-type CIs, similar to the intervals proposed in \cite{Abrevaya:99}, using the asymptotic normal limiting distribution of the estimators and compared the coverage proportion and average length of these intervals with bootstrap CIs based on the nonparametric bootstrap described in this paper using $B=1,000$ samples from the original data. For the MRCE, the validity of the classical bootstrap is proved in \cite{subbotin2007}. The Wald-type CIs remain anti-conservative for the ESE in model 2. 

We observed (result not shown) that, in both models, the bias in estimating the efficient variance of the ESE remains larger than the bias of the asymptotic variance estimates for the SSE and the MRCE. Tables \ref{table:beta_estimation1} and \ref{table:beta_estimation2} show that the coverage proportion of the classical bootstrap CIs converges to the nominal $95\%-$level and the average length of the CIs obtained by resampling from the original data is smaller than the corresponding length of the Wald-type CIs. We also investigated the behavior of Studentized bootstrap CIs (results not shown) based on the variance estimate used in the construction of the Wald-type CIs, but no improvement was observed for the behavior of the bootstrap intervals. 

Our results do not indicate better performances corresponding to smoothing techniques and therefore suggest that smoothing should not be the primary concern in inferences for the current status linear regression model. Note that the Wald-type CIs are constructed using smoothing kernel estimation for the variance estimate and that the only results obtained without any smoothing are the bootstrap CIs for the SSE and the MRCE. It is noteworthy that the SSE tends to perform better than the MRCE, which is not based on a nuisance parameter that is not estimable at $\sqrt n-$rate. Based on these results, we recommend the use of the SSE in combination with the nonparametric bootstrap procedure for doing inference in the current status linear regression model.

\begin{table}[!ht]
	\centering
	\caption{\rm Simulation model 1: mean, $n$ times the variance and $n$ times MSE. CP: coverage  proportion of 95\% confidence intervals (Wald-type intervals based on a kernel variance estimate and classical bootstrap intervals) that contain the true parameter value $\b_0 = 0.5$, AL: Average length of the CI, for different samples sizes $n$ based on $N=1,000$ simulation runs and $B=1,000$ bootstrap samples. $\ee = 0.001$. SSE = simple score estimator, MRCE = maximum rank correlation estimator and ESE = efficient score estimator. }
	\label{table:beta_estimation1}
	\vspace{0.5cm}
	\scalebox{0.93}{
		\begin{tabular}{|ll|ccc||cc|cc|}
			\hline
			Estimate & $n$ & mean & $n\times$var & $n\times$MSE & \multicolumn{2}{|c|}{Wald-type CI} & \multicolumn{2}{|c|}{Bootstrap CI} \\
			& && &&CP & AL & CP & AL \\
			\hline
			SSE& 100&  0.498943 & 0.310723 & 0.310968 &  0.978 &0.265883 & 0.824 & 0.204163 \\
			& 500&0.499717 &0.220885 & 0.220925 & 0.982 & 0.097457 &  0.897 & 0.080317 \\
			& 1000 & 0.500720 &  0.217415 & 0.217933 & 0.977 & 0.065837 & 0.924 & 0.055648\\
			&5000 & 0.499993 & 0.195111 & 0.195112 & 0.977 & 0.027159 & 0.945 & 0.024423\\
			&&&&&&&&\\
			MRCE & 100&  0.497996 & 0.308180& 0.308582  & 0.979 & 0.268731 & 0.821 & 0.205522\\
			& 500 & 0.499761 &0.251232 & 0.251260 & 0.978 & 0.098028 & 0.862 & 0.089143 \\
			&1000 & 0.500553 &  0.246388 & 0.246693	 & 0.973 & 0.063990&  0.911 & 0.053129\\
			& 5000 & 0.499876 & 0.208386 & 0.208462 & 0.965 & 0.027197 & 0.922 & 0.026987 \\			
			&&&&&&&&\\
			ESE & 100& 0.500145 & 0.337755 & 0.337757 & 0.964 & 0.252687 & 0.824 & 0.223849\\
			& 500& 0.499671 & 0.217428 &0.217482 & 0.978 & 0.094390 &0.896 & 0.080003 \\
			& 1000 & 0.500742 & 0.207401 &0.207953& 0.973& 0.063990 & 0.911& 0.053129 \\
			& 5000 & 0.500228 & 0.185614 &0.185874& 0.972 & 0.026396& 0.904& 0.022285 \\			
			\hline
		\end{tabular}}
	\end{table}

	\begin{table}[!ht]
		\centering
		\caption{\rm Simulation model 2: mean, $n$ times the variance and $n$ times MSE. CP: coverage  proportion of 95\% confidence intervals (Wald-type intervals based on a kernel variance estimate and classical bootstrap intervals) that contain the true parameter value $\b_0 = 1$, AL: Average length of the CI, for different samples sizes $n$ based on $N=1,000$ simulation runs and $B=1,000$ bootstrap samples. $\ee = 0.001$. SSE = simple score estimator, MRCE = maximum rank correlation estimator and ESE = efficient score estimator.}
		\label{table:beta_estimation2}
		\vspace{0.5cm}
		\scalebox{0.93}{
			\begin{tabular}{|ll|ccc||cc|cc|}
				\hline
				Estimate & $n$ & mean & $n\times$var & $n\times$MSE & \multicolumn{2}{|c|}{Wald-type CI} & \multicolumn{2}{|c|}{Bootstrap CI} \\
				& && &&CP & AL & CP & AL \\
				\hline
				SSE & 100& 0.935732  & 4.525330 & 4.938096  &0.922&1.000283 &0.855 & 0.79952\\
				& 500& 0.966217  & 4.676249 & 5.246881  &0.926 & 0.399728 & 0.902&0.364210 \\
				& 1000&  0.977799 & 5.032432  &  5.525339 &0.933 & 0.279928 &0.914& 0.262449\\
				& 5000 & 0.989466 & 4.580756 & 5.135616 & 0.945&0.124375&0.948&  0.121388 \\
				&&&&&&&&\\
				MRCE & 100& 1.038510  & 8.500588 & 8.648890   &0.925 &1.125225 &0.889& 1.364034\\
				& 500&  1.006050 &  6.443404& 6.461690   &0.932 & 0.429007 &0.912&  0.473787 \\
				& 1000&  1.002680 &6.294143  & 6.301326  & 0.939& 0.296537  &0.903&  0.320908\\
				& 5000 & 0.998502 &5.160694 & 5.171915 &0.962 &0.129512 &0.954&  0.136487\\
				&&&&&&&&\\
				ESE & 100& 0.974199 & 5.722576 &5.789144 &0.768 &0.604649 &0.827& 0.910229\\
				&500& 0.998806&5.984291 &5.985003 &0.823& 0.290297&0.902& 0.430819\\
				&1000&1.005545 &6.032743 & 6.063495 &0.841 &0.214280  &0.928& 0.302124\\
				& 5000 & 1.002462 &  5.244373  &  5.274692 & 0.892&0.104281&0.951& 0.131427\\
				\hline  		
			\end{tabular}}
		\end{table}

\section{Discussion}
\label{sec:discussion}
In this paper we studied the behavior of the nonparametric bootstrap in current status models. Asymptotic results show that, given the data, the $L_2-$distance between the bootstrap MLE $\hat F_n$ and the underlying distribution function $F_0$ is of order $n^{-1/3}$. This result is noteworthy given the fact that the nonparametric bootstrap is inconsistent for generating the distribution of the MLE. Despite this negative result, we show that it is still possible to use the MLE while doing inferences for certain functionals in the current status model. We illustrated the effectiveness of this result by constructing pointwise confidence intervals around the SMLE and proved the validity of interval estimation in the current status linear regression model. 

The result is applicable to several other nonparametric estimators depending on a cube-root $n$ convergence class. Because of its connection with the MLE, applications of the nonparametric bootstrap involving the Grenander estimator, such as the smoothed Grenander estimator used in \cite{cecile_rik:12} or the goodness-of-fit tests described in \cite{durot2010}, are worthy of study in further research.

Extensions to semiparametric models, where one considers bootstrapping a finite dimensional parameter, are also possible  such as the score estimator for the semiparametric monotone single index model proposed by \cite{fadoua_kim_piet:2017}, which is similar to the current status linear regression estimator discussed in Section \ref{sec:applications_CSLR}. A general bootstrap consistency result for semiparametric M-estimators is derived in \cite{cheng2010}. However, if computations are in first instance based on nonparametric maximum likelihood estimators or least squares estimators of the infinite dimensional parameter, fixing temporarily the finite-dimensional parameter, the use of local smooth functional theory is needed, where the remainder terms involving the cube-root-$n$ M-estimator of the nuisance parameter are shown to be negligible by an application of a result of the type (\ref{bootstrap_result_main}). The treatment of the remainder terms in this local smooth functional theory is a highly non-trivial matter. On the other hand, in \cite{cheng2010}, this negligibility is assumed to hold by their condition  SB3.
	
Furthermore, the results in \cite{cheng2010} hold for a class of exchangeable bootstrap weights of which the multinomial weights considered in this paper are a special case. Although we did not investigate this in the present paper, extensions of our nonparametric bootstrap results to the more general bootstrap resampling schemes seem possible as well.

Another interesting extension of this research is the construction of confidence bands for the distribution instead of the currently proposed pointwise confidence intervals. Note that our main result (\ref{bootstrap_result_main1}) does not imply:
\begin{align}
\label{bootstrap_result_main2}
E\left\{\sup_{t\in[0,R]}\left. n^{1/3} \bigl|\hat F_n(t) - F_0(t)\bigr| \right|Z_1,\dots,Z_n\right\}= O_p(1).
\end{align}
A bound on $\sup_{t\in[0,R]}\left. n^{1/3} \bigl|\hat F_n(t) - F_0(t)\bigr| \right|$ which no doubt would contain logarithmic factors, would be needed for confidence bands instead of our pointwise confidence intervals. The idea is that the process $t\mapsto n^{1/3} \bigl\{\hat F_n(t) - F_0(t)\bigr\}$ will fall apart into asymptotically independent pieces, and that we therefore expect Gumbel-type distributions to enter, via the maximum of independent random variables. The theory for this still has to be developed, however. What struck us in the present simulation studies is how comparatively well the global behavior of our pointwise confidence intervals still was, indicating that the extra logarithmic factors do not have such a very large impact.

Probably results similar to those presented in the current paper will follow for the more challenging interval censoring, type II models where the development of the local limit theory for the MLE has not yet been settled. It is reasonable to believe that the nonparametric bootstrap also allows for inferences with the maximum smoothed likelihood estimator studied in \cite{piet:14}.

\section{Appendix}
\label{sec:appendix}
\subsection{Proof of Lemma \ref{lemma_Th11.3}}
Before proving Lemma \ref{lemma_Th11.3} we provide two technical lemmas.
\begin{lemma}
\label{lemma:hat_U_n-deviation_bootstrap}
Let $\a>0$. There exist constants $K_1,K_2>0$ such that, for each $j\ge1,\,j\in\mathbb N$,
\begin{align}
\label{unconditional_upper_bound1}
&P_{M|Z}\Biggl\{\exists y\in\bigl[(j-1)n^{-1/3},jn^{-1/3}\bigr):\nonumber\\
&\qquad\qquad\left|\int_{u\in(U(a),U(a)+y]}\{\d-F_0(u)\}\,d\bigl(\hat{\P}_n-\P_n\bigr)(u,\d)\right| \ge\a(j-1)^2n^{-2/3}\Biggr\}\nonumber\\
&\le K_1\exp\left\{-K_2(j-1)^{3/2}\right\},
\end{align}
in probability.
	
Likewise, there exist constants $K_1,K_2>0$ such that, for each $j\ge1,\,j\in\mathbb N$,
\begin{align}
\label{unconditional_upper_bound2}
&P_{M|Z}\Biggl\{\exists y\in\bigl[-jn^{-1/3},-(j-1) n^{-1/3}\bigr):\nonumber\\
&\qquad\qquad\left|\int_{u\in(U(a)+y,U(a)]}\{\d-F_0(u)\}\,d\bigl(\hat{\P}_n-\P_n\bigr)(u,\d)\right|
\ge\a(j-1)^2n^{-2/3}\Biggr\}\nonumber\\
&\le K_1\exp\left\{-K_2(j-1)^{3/2}\right\},
\end{align}
in probability.
\end{lemma}

\begin{proof}
We only prove (\ref{unconditional_upper_bound1}), since the proof of (\ref{unconditional_upper_bound2}) is similar.
Let  ${\cal F}_t$ be the (Vapnik-Cervonenkis) class of functions
\begin{align*}
{\cal F}_t=\left\{(\d-F_0(v))1_{(U(a),U(a)+u]}(v):u\in[0,t],\,\d\in\{0,1\}\right\},
\end{align*}
with envelope
\begin{align*}
F_t(v,\d)=1_{(U(a),U(a)+t]}(v),\qquad v\in[0,t].
\end{align*}
To prove (\ref{unconditional_upper_bound1}), we use that an exponential tail bound can be derived from a bounded Orlicz norm $\|\cdot\|_{P,\psi}$, i.e., 
when taking $\psi_1(x)=\exp(x)-1$, for $x\ge0$, we get, for $x>0$ the inequality
\begin{align}
\label{Orlicz-ineq}
P(|X|>x) \le 2\exp\left\{-x/\|X\|_{P,\psi_1}\right\},
\end{align}
where
\begin{align*}
\|X\|_{P,\psi_1}=\inf\left\{C>0:E\psi_1\left(\frac{|X|}{C}\right)\le1\right\}.
\end{align*}
Using the second statement of Theorem 2.14.5 in \cite{vdvwe:96}, with $p=1$, we get, the following inequality:
\begin{align}
\label{Orlicz_ineq1}
&\left\|\left\|\sqrt{n}\left(\hat{\P}_n-\P_n\right)\right\|_{{\cal F}_t}^*\right\|_{\P_n,\psi_1}\nonumber\\
&\qquad\lesssim \left\|\left\|\sqrt{n}\left(\hat{\P}_n-\P_n\right)\right\|_{{\cal F}_t}^*\right\|_{\P_n,1}
+n^{-1/2}\{1+\log n\}\|F_t\|_{\P_n,\psi_1}\,,
\end{align}
where $\|\cdot\|^*_{{\cal F}_t}$ denotes the so-called measurable majorant of $\|\cdot\|_{{\cal F}_t}$ (see \cite{vdvwe:96}). (Note that we use temporarily the "*"  notation which is used for bootstrap variables in the rest of the paper.)

Furthermore, we have by the rightmost inequality of Theorem 2.14.1 of \cite{vdvwe:96} that
\begin{align*}
	\left\|\left\|\sqrt{n}\left(\hat{\P}_n-\P_n\right)\right\|_{{\cal F}_t}^*\right\|_{\P_n,1}\lesssim J\left(1,{\cal F}_t\right)\left\|F_t\right\|_{\P_n,2}\,,
\end{align*}
where $J(\d,{\cal F}_t)$ is defined by
\begin{align*}
	J(\d,{\cal F}_t)=\sup_Q\int_0^{\d}\sqrt{1+\log N\left(\e\|F\|_{Q,2},{\cal F}_t,L_2(Q)\right)}\,d\e,
\end{align*}
and where the supremum is over all discrete probability measure $Q$ with $\|F_t\|_{Q,2}>0$. Since ${\cal F}_t\subset {\cal F}_{R-U(a)}$ for all $t\in[0,R-U(a)]$, and since ${\cal F}_{R-U(a)}$ is a Vapnik-Cervonenkis class,
$J(\d,{\cal F}_t)$ is bounded by a fixed constant
for all $t\in[0,R-U(a)]$, and we get:
\begin{align*}
\left\|\left\|\sqrt{n}\left(\hat{\P}_n-\P_n\right)\right\|_{{\cal F}_t}^*\right\|_{\P_n,1}\lesssim \left\|F_t\right\|_{\P_n,2}\,,
\end{align*}
uniformly for all $t\in[0,R-U(a)]$. Note that
\begin{align}
\label{envelope_upbound}
	\left\|F_t\right\|_{\P_n,2}^2=\int_{u\in U(a),U(a)+t]}\,d\P_n(u,\d)=\int_{u\in U(a),U(a)+t]}\,d\G_n(u),
\end{align}
$t\in[U(a),R-U(a)].$
We next evaluate the second term on the right-hand side of (\ref{Orlicz_ineq1}). We have:
	\begin{align*}
	\int\psi_1\left(\frac{F_t(u,\d)}c\right)\,d\P_n(u,\d)=\left\{e^{1/c}-1\right\}\int 1_{(U(a),U(a)+t]}(u)\,d\G_n(u),
	\end{align*}
and
	\begin{align*}
	&\left\{e^{1/c}-1\right\}\int 1_{(U(a),U(a)+t]}(u)\,d\G_n(u)\le1\\
	&\qquad\iff c\ge \frac1{\log\left\{1+1/\int_{u\in U(a),U(a)+t]}\,d\G_n(u)\right\}}\,.
	\end{align*}
Thus (\ref{Orlicz_ineq1}) becomes, using (\ref{envelope_upbound}),
\begin{align}
	\label{Orlicz_ineq2}
	&\left\|\left\|\sqrt{n}\left(\hat{\P}_n-\P_n\right)\right\|_{{\cal F}_t}^*\right\|_{\P_n,\psi_1}\nonumber\\
	&\le c_1\left\{\int_{u\in U(a),U(a)+t]}\,d\G_n(u)\right\}^{1/2}+\frac{1+\log n}{n^{1/2}\log\left\{1+1/\int_{u\in U(a),U(a)+t]}\,d\G_n(u)\right\}}\,,
\end{align}
	for a constant $c_1>0$. If $t\ge Kn^{-1/3}$ we get for the second term in probability,
	\begin{align*}
	\frac{1+\log n}{n^{1/2}\log\left\{1+1/\int_{u\in U(a),U(a)+t]}\,d\G_n(u)\right\}}\ll c_1\left\{\int_{u\in U(a),U(a)+t]}\,d\G_n(u)\right\}^{1/2}.
	\end{align*}
We have:
\begin{align*}
&\int_{u\in [U(a),U(a)+t]}\,d\G_n(u)\\
&=\int_{u\in [U(a),U(a)+t]}\,dG(u)+\int_{u\in [U(a),U(a)+t]}\,d\bigl(\G_n-G\bigr)(u)\\
& =\int_{u\in [U(a),U(a)+t]}\,dG(u)+O_p\left(n^{-1/2}\right)= O(t) +O_p\left(n^{-1/2}\right) \\
&=  O(t) +O_{P_M}\left(n^{-1/2}\right),
\end{align*}
in probability (since a term defined only on the probability space $({\cal X}, {\cal A}, P)$ of order $O_p(1)$ is also of order $O_{P_M}(1)$ in probability).
So we obtain, for $j\ge K$ in probability, conditioning on $(T_1,\dd_1),(T_2,\dd_2),\ldots$ using the inequality on Orlicz norms on p.\ 96 or 239 of \cite{vdvwe:96}:
\begin{align*}
&P_{M|Z}\Biggl\{\exists y\in\bigl[(j-1)n^{-1/3},jn^{-1/3}\bigr):\nonumber\\
&\qquad\qquad\qquad\left|\int_{u\in(U(a),U(a)+y]}\{\d-F_0(u)\}\,d\bigl(\hat{\P}_n-\P_n\bigr)(u,\d)\right| \ge\a(j-1)^2n^{-2/3}\Biggr\}\\
&=P_{M|Z}\Biggl\{\exists y\in\bigl[(j-1)n^{-1/3},jn^{-1/3}\bigr):\\
&\qquad\qquad\qquad \sqrt n\left|\int_{u\in(U(a),U(a)+y]}\{\d-F_0(u)\}\,d\bigl(\hat{\P}_n-\P_n\bigr)(u,\d)\right|\ge\a(j-1)^2n^{-1/6}\Biggr\}\\
&\le2\exp\left\{-m(j-1)^2n^{-1/6}/\left\|\left\|\sqrt{n}\left(\hat\P_n-\P_n\right)\right\|_{{\cal F}_{jn^{-1/3}}}^*\right\|_{\P_n,\psi_1}\right\}\\
&\le 2\exp\left\{-c_2m(j-1)^{3/2}\right\},
\end{align*}
for some $c_2>0$. This proves the statement.
\end{proof}

\begin{lemma}
	\label{lemma:bound_original_integral}
	For each $\e>0$ and $x\in[0,R-U(a)]$,
	\begin{align*}
	\left|\int_{u\in(U(a),U(a)+x]}\{\d-F_0(u)\}\,d\bigl(\P_n-P\bigr)(u,\d)\right|\le \e x^2+O_p\left(n^{-2/3}\right).
	\end{align*}
\end{lemma}
\begin{proof}
As in the proof of Lemma \ref{lemma:hat_U_n-deviation_bootstrap}, we consider the Vapnik-Cervonenkis collection of functions:
\begin{align*}
{\cal F}_t=\left\{(\d-F_0(v))1_{(U(a),U(a)+u]}(v):u\in[0,t],\,\d\in\{0,1\}\right\},
\end{align*}
with envelope
\begin{align*}
F_t(v,\d)=1_{(U(a),U(a)+t]}(v),\qquad v\in[0,t].
\end{align*}
We have, using Theorem 2.14.1 of \cite{vdvwe:96}:
\begin{align}
E_X\left\{\sup_{f\in{\cal F}_t}\left|\P_n-P\right|(f)\right\}^2\le Kn^{-1}\left\|F_t\right\|_{P,2}^2,
\end{align}
for some $K>0$. Since,
\begin{align*}
\left\|F_t\right\|_{P,2}^2=\int_{u\in U(a),U(a)+t]}\,dP(u,\d)=\int_{u\in U(a),U(a)+t]}\,dG(u)=O(t),
\end{align*}
for $t\in[U(a),R-U(a)]$, we get, by Markov's inequality,
	\begin{align*}
	&P\left\{n^{2/3}\left|\int_{u\in(U(a),U(a)+jn^{-1/3}]}\{\d-F_0(u)\}\,d\bigl(\P_n-P\bigr)(u,\d)\right|> A+ \e(j-1)^2\right\}\\
	&\le Kj/\left\{A+\e(j-1)^2\right\}^2.
	\end{align*}
	The result now easily follows, see, e.g., \cite{kimpol:90}. p.\ 201.
\end{proof}

As a consequence of Lemma \ref{lemma:hat_U_n-deviation_bootstrap} and Lemma \ref{lemma:bound_original_integral} we get the following result.

\begin{lemma}
\label{lemma:11.5}
Let $\hat V_n$ and $\hat{\bar V}_n$ be defined by
\begin{align}
\label{hatV_n}
\hat V_n(t)=\int_{u\in[0,t]}\d\,d\hat{\P}_n(u,\d),\qquad \hat{\bar V}_n(t)=\int_{u\in[0,t]} F_0(u)\,d\hat{\G}_n(u),\qquad\,t\in[0,R].
\end{align}
where the process $\hat \G_n$ is defined in (\ref{hatV_n_G_n}), and let $\hat D_n=\hat{V}_n-\hat{\bar V}_n$. Then there exist constants $K_1,K_2>0$ such that, for each $j\ge1,\,j\in\mathbb N$,
	\begin{align}
	\label{unconditional_upper_bound1a}
	&P_{M|Z}\left\{\exists y\in\bigl[(j-1)n^{-1/3},jn^{-1/3}\bigr):\hat{D}_n(U(a)+y)-\hat{D}_n(U(a))\right.\nonumber\\
	&\left.\qquad\qquad\qquad\qquad\qquad\qquad\qquad\le-\int_{U(a)}^{U(a)+y}\bigl\{F_0(u)-F_0(U(a))\bigr\}\,d\hat{\G}_n(u)\right\}\nonumber\\
	&\le K_1\exp\left\{-K_2(j-1)^{3/2}\right\},
	\end{align}
	in probability. Likewise, there exist constants $K_1,K_2>0$ such that, for each $j\ge1,\,j\in\mathbb N$,
	\begin{align}
	\label{unconditional_upper_bound2a}
	&P_{M|Z}\left\{\exists y\in\bigl[-jn^{-1/3},-(j-1) n^{-1/3}\bigr):\hat{D}_n(U(a)+y)-\hat{D}_n(U(a))\right.\nonumber\\
	&\left.\qquad\qquad\qquad\qquad\qquad\qquad\qquad\le-\int_{U(a)+y}^{U(a)}\bigl\{F_0(u)-F_0(U(a))\bigr\}\,d\hat{\G}_n(u)\right\}\nonumber\\
	&\le K_1\exp\left\{-K_2(j-1)^{3/2}\right\},
	\end{align}
	in probability.
\end{lemma}

\begin{proof}
We again only prove (\ref{unconditional_upper_bound1}), since the proof of (\ref{unconditional_upper_bound2}) is similar. First note:
\begin{align*}
&P_{M|Z}\Biggl\{\exists y\in\bigl[(j-1)n^{-1/3},jn^{-1/3}\bigr):\hat D_n(U(a)+y)-\hat D_n(U(a))
\\
&\qquad\qquad\qquad\qquad\qquad\qquad\qquad
\le-\int_{U(a)}^{U(a)+y}\bigl\{F_0(u)-F_0(U(a))\bigr\}\,d\hat \G_n(u)\Biggr\}\\
&\quad\le P_{M|Z}\Biggl\{\exists y\in\bigl[(j-1)n^{-1/3},jn^{-1/3}\bigr):\left|\hat D_n(U(a)+y)-\hat D_n(U(a))\right|
\\
&\qquad\qquad\qquad\qquad\qquad\qquad\qquad \ge\int_{U(a)}^{U(a)+y}\bigl\{F_0(u)-F_0(U(a))\bigr\}\,d\hat \G_n(u)\Biggr\}.\nonumber\\
\end{align*}	
Furthermore:
\begin{align}
	\int_{U(a)}^{U(a)+y}&\bigl\{F_0(u)-F_0(U(a))\bigr\}\,d\hat{\G}_n(u)\nonumber\\
	&=\int_{U(a)}^{U(a)+y}\bigl\{F_0(u)-F_0(U(a))\bigr\}\,d\G_n(u)\nonumber\\
	&\qquad+\int_{U(a)}^{U(a)+y}\bigl\{F_0(u)-F_0(U(a))\bigr\}\,d\bigl(\hat{\G}_n-\G_n\bigr)(u)\nonumber\\
	&=\int_{U(a)}^{U(a)+y}\bigl\{F_0(u)-F_0(U(a))\bigr\}\,dG(u)\nonumber\\
	&\qquad+\int_{U(a)}^{U(a)+y}\bigl\{F_0(u)-F_0(U(a))\bigr\}\,d\bigl(\G_n-G\bigr)(u) \nonumber \\
	&\qquad +\int_{U(a)}^{U(a)+y}\bigl\{F_0(u)-F_0(U(a))\bigr\}\,d\bigl(\hat{\G}_n-\G_n\bigr)(u),
\end{align}
and for the dominant term on the right-hand side we get:
\begin{align*}
&\int_{U(a)}^{U(a)+y}\bigl\{F_0(u)-F_0(U(a))\bigr\}\,dG(u)\ge m_0\int_{U(a)}^{U(a)+y}\{u-U(a)\}\,dG(u)\\
&\ge m_0m_1\int_{U(a)}^{U(a)+y}\{u-U(a)\}\,du=\tfrac12m_0m_1\{y-U(a)\}^2,
\end{align*}
where $m_0=\inf_{u\in[U(a),R]} f_0(u)$ and $m_1=\inf_{u\in[U(a),R]} g(u)$.
We therefore consider the probability:
\begin{align}
	\label{simplified_ineq2}
	&P_{M|Z}\Bigl\{\exists y\in\bigl[(j-1)n^{-1/3},jn^{-1/3}\bigr):\left|\hat{D}_n(U(a)+y)-\hat{D}_n(U(a))\right|\\ 
	&\qquad\qquad\qquad\qquad\qquad\qquad\qquad\qquad\qquad\qquad\qquad\qquad\ge m(j-1)^2n^{-2/3}\Bigr\}.\nonumber
\end{align}
where
\begin{align*}
m=\tfrac12\min\left\{\inf_{u\in[t_0,R]}f_0(u),\inf_{u\in[t_0,R]}g(u)\right\}.
\end{align*}
We also have:
\begin{align*}
&\hat{D}_n(U(a)+y)-\hat{D}_n(U(a))=\int_{u\in(U(a),U(a)+y]}\{\d-F_0(u)\}\,d\hat{\P}_n(u,\d)\\
&=\int_{u\in(U(a),U(a)+y]}\{\d-F_0(u)\}\,d\bigl(\hat{\P}_n-P\bigr)(u,\d)\\
&=\int_{u\in(U(a),U(a)+y]}\{\d-F_0(u)\}\,d\bigl(\hat{\P}_n-\P_n\bigr)(u,\d)\\
&\qquad+\int_{u\in(U(a),U(a)+y]}\{\d-F_0(u)\}\,d\bigl(\P_n-P\bigr)(u,\d).
\end{align*}
By Lemma \ref{lemma:bound_original_integral}, we may assume that for $x\in[0,R-U(a)]$,
\begin{align}
\label{original_small}
\left|\int_{u\in(U(a),U(a)+x]}\{\d-F_0(u)\}\,d\bigl(\P_n-P\bigr)(u,\d)\right|\le \e x^2+K n^{-2/3},
\end{align}
for some $K>0$ and $0<\e<m/2$. Considering sequences $X=(T_1,\dd_1),(T_2,\dd_2)\dots$, satisfying (\ref{original_small}), we get:
\begin{align*}
&P_{M|Z}\Bigl\{\exists y\in\bigl[(j-1)n^{-1/3},jn^{-1/3}\bigr):\left|\hat{D}_n(U(a)+y)-\hat{D}_n(U(a))\right|\\
&\qquad\qquad\qquad\qquad\qquad\qquad\qquad\qquad\qquad\qquad\qquad\qquad\ge m(j-1)^2n^{-2/3}\Bigr\}\\
&\le P_{M|Z}\Biggl\{\exists y\in\bigl[(j-1)n^{-1/3},jn^{-1/3}\bigr):\\
&\qquad\qquad\quad\left|\int_{u\in(U(a),U(a)+y]}\{\d-F_0(u)\}\,d\bigl(\hat{\P}_n-\P_n\bigr)(u,\d)\right|\ge \tfrac12m(j-1)^2n^{-2/3}\Biggr\}\\
&\le K_1\exp\left\{-K_2(j-1)^{3/2}\right\},
\end{align*}
with probability tending to one, using Lemma \ref{lemma:hat_U_n-deviation_bootstrap}. 
\end{proof}

We now prove Lemma \ref{lemma_Th11.3}.
\begin{proof}[Proof of Lemma \ref{lemma_Th11.3}]
Suppose that $n^{1/3}|\hat U_n(a)-U(a)|> x$ for some $x> 0$, then there exists a $y$ such that,  $n^{1/3}\left|y-U(a)\right|> x$ and $\hat V_n(y)-a\hat \G_n(y)\le \hat V_n(U(a))-a\hat \G_n(U(a))$. Hence,
\begin{align*}
&P_{M|Z}\left\{n^{1/3}\left|\hat U_n(a)-U(a)\right|\ge x\right\}\\
& \le 	P_{M|Z}\Biggl(\inf_{y-U(a)\ge n^{-1/3}x} \hat D_n(y)-\hat D_n(U(a))
\\
&\qquad\qquad\qquad\qquad\quad\quad\le-\int_{U(a)}^{y}\bigl\{F_0(u)-F_0(U(a))\bigr\}\,d\hat \G_n(u)  \Biggr)\\
& \le 	\sum_{j=i}^\infty P_{M|Z}\Biggl(\exists y\in\bigl[(j-1)n^{-1/3},jn^{-1/3}\bigr):\hat D_n(U(a)+y)-\hat D_n(U(a))\\
&\qquad\qquad\qquad\qquad\qquad \le-\int_{U(a)}^{U(a)+y}\bigl\{F_0(u)-F_0(U(a))\bigr\}\,d\hat \G_n(u) \Biggr),
\end{align*}
where $x \in [(i-1)n^{-1/3}, in^{-1/3}]$. By Lemma \ref{lemma:11.5}, this is bounded above by
\begin{align*}
&\sum_{j=i}^\infty K_1\exp\left\{K_2(j-1)^{3/2}\right\} \\
&= K_1\exp\left\{-K_2(i-1)^{3/2}\right\}\sum_{j=i}^\infty \exp\left\{-K_2[(j-1)^{3/2}-(i-1)^{3/2}] \right\} \\
&\le K_1'\exp\left\{K_2'(i-1)^{3/2}\right\},
\end{align*}	
for constants $K_1,K_1',K_2, K_2' >0$.	
\end{proof}

\subsection{Proof of Lemma \ref{lemma: toy-representation-bootstrapSMLE}}
We introduce notations $K_h$ and $\IK_h$ to denote the scaled versions of $K$ and $\IK$ respectively:
$$K_h(u) = h^{-1}K(u/h) \qquad \text{and} \qquad \IK_h(u) = \IK(u/h). $$
\begin{proof}

 Define the function
 \begin{align*}
 \psi_{t,h}(u) = \frac{K_h(t-u)}{g(u)}.
 \end{align*}
 Denote the points of jump of the MLE $\hat F_n$ by $\hat \t_1,\ldots,\hat \t_m$ and define the piecewise constant function $\bar\psi_{t,h}$  with only jumps at $\hat \t_1,\ldots,\hat \t_m$ by 
 \begin{align*}
 \bar\psi_{t,h}(u)=
 \left\{\begin{array}{lll}
 \psi_{t,h}(\hat\t_i),\,&\mbox{ if } F_{0}(u)>\hat F_n(\hat \t_i),\,u\in[\hat\t_i,\hat\t_{i+1}),\\
 \psi_{t,h}(s),\,&\mbox{ if } F_0(u)=\hat F_n(s),\mbox{ for some }s\in[\hat \t_i,\hat \t_{i+1}),\\
 \psi_{t,h}(\hat\t_{i+1}),\,&\mbox{ if }\tilde F_0(u)<\hat F_n(\t_i),\,u\in[\hat \t_i,\hat \t_{i+1}).
 \end{array}
 \right.
 \end{align*}
 By the convex minorant interpretation of $\hat F_n$, we have 
 \begin{align*}
 \int \bar \psi_{t,h}(u)(\d - \hat F_n(u))d\hat\P_n(u,\d) = 0,
 \end{align*}
 (see the discussion of the SMLE in \cite{piet_geurt:14}, p.\ 332).\\
 We can write
 \begin{align*}
 &\tilde F_{nh}^*(t) = \int \IK_h(t-u)\,d\hat F_n(u)\\&= \int \IK_h(t-u)\,d( \hat F_n-F_0)(u) + \int \IK_h(t-u)\,d F_0(u)\\
 &= \int \psi_{t,h}(u)\,\left\{  \hat F_n(u)- F_0(u)\right\}dG(u) +  \int \IK_h(t-u)\,d F_0(u)\\
 &=  \int \psi_{t,h}(u) \left\{ F_0(u)  - \hat F_n(u) \right\}d(\hat\G_n-G)(u) +\int \psi_{t,h}(u) \left\{\d -  F_0(u)\right\}d\hat\P_n(u,\d)\\
 &\qquad+\int \left\{\psi_{t,h}(u)-\bar\psi_{t,h}(u) \right \} \left\{\hat F_n(u)-\d \right\}d\hat\P_n(u,\d)  + \int \IK_h(t-u)\,d F_0(u)\\
 &=  \tilde F_{nh}^{(toy)*}(t) + \int \psi_{t,h}(u) \left\{ F_0(u)  - \hat F_n(u) \right\}d(\hat\G_n-G)(u,\d)\\
 &\qquad+\int \left\{\psi_{t,h}(u)-\bar\psi_{t,h}(u)\right \} \left\{\hat F_n(u)-\d \right\}d\hat\P_n(u,\d)\\
 & = \tilde F_{nh}^{(toy)*}(t)  + A_{I} + A_{II}.
 \end{align*}
 We first evaluate $A_I$ and show that this term is $o_{P_M}(n^{-2/5})$ in probability, we have:
 \begin{align*}
 A_I&=\int \psi_{t,h}(u) \left\{ F_0(u)  - \hat F_n(u) \right\}d(\hat\G_n-G)(u,\d)\\
 &=\int \psi_{t,h}(u) \left\{ F_0(u)  - \hat F_n(u) \right\}d(\hat\G_n-\G_n)(u,\d)\\
 &\qquad+\int \psi_{t,h}(u) \left\{ F_0(u)  - \hat F_n(u) \right\}d(\G_n-G)(u,\d)
 \end{align*}
 An argument similar to that of Lemma A.7 in \cite{piet_geurt_birgit:10} shows that 
 $$ \int \psi_{t,h}(u) \left\{ F_0(u)  - \hat F_n(u) \right\}d(\G_n-G)(u,\d) = o_p(n^{-2/5}),$$
 and hence,
  $$ \int \psi_{t,h}(u) \left\{ F_0(u)  - \hat F_n(u) \right\}d(\G_n-G)(u,\d) = o_{P_M}(n^{-2/5}),$$
in probability. Similarly to the proof of Lemma A.7 in \cite{piet_geurt_birgit:10}, we can also show that 
\begin{align}
\label{diff-hatGn-Gn}
\int \psi_{t,h}(u) \left\{ F_0(u)  - \hat F_n(u) \right\}d(\hat \G_n-\G_n)(u,\d) = o_{P_M}(n^{-2/5}),
\end{align}
  in probability, such that,
  $$ A_I = o_{P_M}(n^{-2/5})  \qquad \text{in probability}. $$
 We now study the term $A_{II}$.  Using the same inequality for $\psi_{t,h}-\bar\psi_{t,h}$ as used in the second display after (11.49) on p.\ 333
 of \cite{piet_geurt:14}, we get for some constant $C>0$ that:
 \begin{align}
 \label{diff-barpsi-psi}
 \left\vert \bar \psi_{t,h}(u)  - \psi_{t,h}(u) \right\vert \leq Ch^{-2}\left\vert \hat F_n(u)- F_0(u) \right\vert,
 \end{align} 
 for all $u$ such that $f_0$ is positive and continuous in a neighborhood around $u$. We decompose the term 
 $A_{II}$ as follows,
 \begin{align}
 \label{A_II}
 A_{II} &= \int \left\{\bar\psi_{t,h}(u) - \psi_{t,h}(u) \right \} \left\{\hat F_n(u)-F_0(u) \right\}d\hat\P_n(u,\d) \nonumber\\
 &\qquad + \int \left\{\bar\psi_{t,h}(u) - \psi_{t,h}(u) \right \} \left\{F_0(u)-\d \right\}d\hat\P_n(u,\d).
 \end{align}
 For the first term on the right-hand side of the above display we write,
 \begin{align}
 \label{A_IIa}
 & \int \left\{\bar\psi_{t,h}(u) - \psi_{t,h}(u) \right \} \left\{\hat F_n(u)-F_0(u) \right\}d\hat\P_n(u,\d) \nonumber
 \\
 &\qquad\qquad =  \int \left\{\bar\psi_{t,h}(u) - \psi_{t,h}(u) \right \} \left\{\hat F_n(u)-F_0(u)\right\}d(\hat\P_n- \P_n)(u,\d) \nonumber
 \\ 
 &\qquad\qquad\qquad +  \int \left\{\bar\psi_{t,h}(u) - \psi_{t,h}(u) \right \} \left\{\hat F_n(u)-F_0(u) \right\}d\P_n(u,\d) \nonumber\\
 &\qquad\qquad \le  \int \left\{\bar\psi_{t,h}(u) - \psi_{t,h}(u) \right \} \left\{\hat F_n(u)-F_0(u)\right\}d(\hat\P_n- \P_n)(u,\d) \nonumber
  \\ 
  &\qquad\qquad\qquad +  Ch^{-2}\int_{t-h}^{t+h} \left\{\hat F_n(u)-F_0(u) \right\}^2d\P_n(u,\d), 
 \end{align}
where we use (\ref{diff-barpsi-psi}) in the last inequality. The first term in the display above is $o_{P_M}(n^{-2/5})$ in probability by (\ref{diff-hatGn-Gn}) and (\ref{diff-barpsi-psi}). Since
\begin{align*}
E_{M|Z}\left\{ \hat F_n(t) - F_0(t)\right\}^2 < Kn^{-2/3} \qquad \forall t \in (0,R),
\end{align*}
in probability, we have by Markov's inequality and Fubini's theorem that,
\begin{align}
\label{L_2_truncated}
 \int_{t-h}^{t+h} \left\{\hat F_n(u)-F_0(u) \right\}^2d\P_n(u,\d) = O_{P_M}\left(hn^{-2/3}\right) \text{   in probability}.
\end{align}
Hence, for $h\asymp n^{-1/5}$, we get for the second term in (\ref{A_IIa}):
 \begin{align*}
&Ch^{-2}\int_{t-h}^{t+h} \left\{\hat F_n(u)-F_0(u) \right\}^2d\P_n(u,\d)
\\ &\qquad=  O_{P_M}\left(h^{-1}n^{-2/3}\right)=O_{P_M}\left(n^{-7/15}\right) = o_{P_M}\left(n^{-2/5}\right),
\end{align*}
in probability. For the second term of (\ref{A_II}) we have
 \begin{align*}
 &\int \left\{\bar\psi_{t,h}(u) - \psi_{t,h}(u) \right \} \left\{F_0(u)-\d \right\}d\hat \P_n(u,\d)
 \\
 &\qquad=\int \left\{\bar\psi_{t,h}(u) - \psi_{t,h}(u) \right \} \left\{F_0(u)-\d \right\}d(\hat\P_n-\P_n)(u,\d)\\
 &\qquad\qquad\qquad+ \int \left\{\bar\psi_{t,h}(u) - \psi_{t,h}(u) \right \}  \left\{F_0(u)-\d \right\}d\bigl(\P_n-P\bigr)(u,\d).
 \end{align*}
Similar to the arguments used in the treatment of term $A_I$ above, we get by using again arguments similar to that of Lemma A.7 in \cite{piet_geurt_birgit:10} that:
 $$ \int \left\{\bar\psi_{t,h}(u) - \psi_{t,h}(u) \right \}  \left\{F_0(u)-\d \right\}d\bigl(\hat\P_n-\P_n\bigr)(u,\d)= o_{P_M}(n^{-2/5}),$$
 and 
  $$ \int \left\{\bar\psi_{t,h}(u) - \psi_{t,h}(u) \right \}  \left\{F_0(u)-\d \right\}d\bigl(\P_n-P\bigr)(u,\d)= o_{P_M}(n^{-2/5}),$$
 in probability. 
 \end{proof}

 \subsection{The current status linear regression model: bootstrap validity }
 In this section we give a road map for the proof of the bootstrap validity in the current status linear regression model. We assume that the assumptions stated in Theorem 4.1 of \cite{kim_piet:17:AOS} hold. Since the proof is very similar to the proof of Theorem 4.1 in \cite{kim_piet:17:AOS}, we leave the details to the interested reader. 
 Consider the bootstrap score function
\begin{align}
\label{first_score_function}
\hat \psi_n^{(\ee)}(\b) =\int_{\hat F_{n,\b}(t-\b' x)\in[\ee,1-\ee]} x\{\d - \hat F_{n,\b}(t-\b'x) \}\,d\hat \P_n(t,x,\d),
\end{align}
for some fixed truncation parameter $\ee \in (0,1/2)$.
 
The main idea is to show that
 \begin{align}
 \label{score_bootstrap}
\hat \psi_n^{(\ee)}(\hat \b_n) &= V(\hat \b_n-\b_0)+ \int_{F_0(t-\b_0'x) \in[\ee,1-\ee]}\{x-E(X|T-\b_0' X= t-\b_0'x)\}\nonumber\\
&\qquad\qquad\qquad\qquad\qquad\qquad\qquad\cdot\{\d - F_0(t-\b_0' x)\}\,d(\hat \P_n-\P_n)(t,x,\d) \nonumber\\
& \quad+ \int_{F_0(t-\b_0'x) \in[\ee,1-\ee]}\{x-E(X|T-\b_0' X= t-\b_0'x)\}\nonumber\\
&\qquad\qquad\qquad\qquad\qquad\qquad\qquad\cdot\{\d - F_0(t-\b_0' x)\}\,d(\P_n-P)(t,x,\d) \nonumber\\
&\quad + o_{P_M}(n^{-1/2}+  (\hat\b_n - \b_0)),
\end{align}
in probability, where $E$ denotes the unconditional expectation. As in \cite{kim_piet:17:AOS} we can work with the definition  
\begin{align*}
  &\hat \psi_n^{(\ee)}(\hat \b_n) = 0,
\end{align*}
 for the score estimator $\hat \b_n$. Since by the proof of Theorem 4.1 in \cite{kim_piet:17:AOS},
\begin{align*}
 &-\sqrt n V(\b_n-\b_0) \\
 &\quad = \sqrt n\int_{F_0(t-\b_0'x) \in[\ee,1-\ee]}\{x-E(X|T-\b_0' X= t-\b_0'x)\}\nonumber\\
 &\qquad\qquad\qquad\qquad\qquad\qquad\qquad\cdot\{\d - F_0(t-\b_0' x)\}\,d(\P_n-P)(t,x,\d) \nonumber\\
 &\qquad \quad + o_{p}(1+ \sqrt n (\b_n - \b_0)),
   \end{align*}
 we get that,
 \begin{align*}
 &-\sqrt nV(\hat \b_n-\b_n) \\
 &\quad = \sqrt n\int_{F_0(t-\b_0'x) \in[\ee,1-\ee]}\{x-E(X|T-\b_0' X= t-\b_0'x)\}\nonumber\\
 &\qquad\qquad\qquad\qquad\qquad\qquad\qquad\cdot\{\d - F_0(t-\b_0' x)\}\,d(\hat\P_n-\P_n)(t,x,\d) \nonumber\\
 &\qquad \quad +  o_{P_M}(1+ \sqrt n (\hat\b_n - \b_0)).
 \end{align*}
 The validity of the bootstrap then follows by the arguments given in Section \ref{sec:applications_CSLR}. Very important in the proof of (\ref{score_bootstrap}) is the conditional bootstrapped $L_2$-result,
  \begin{align}
  \label{L2-bootsrap-CSLR}
  \sup_{\b}\int\left\{ \hat F_{n,\b}(t-\b'x) - F_{\b}(t-\b'x) \right\}^2d\P_n(t,x, \d )= O_{P_M}\left(n^{-2/3}\right),
  \end{align}
  in probability, where $F_{\b}$ is defined in (\ref{def_F_beta}).
 
Let $\bar\f_{\hat\b_n,\hat F_{n,\hat\b_n}}$ be a (random) piecewise constant version of $\f_{\hat\b_n}$, where
$$\f_{\b}\stackrel{def}{=}E\left\{X|T-\b'X=u\right\},$$
and where, for a piecewise constant distribution function $F$ with finitely many jumps at $\t_1<\t_2<\dots$, the function $\bar\f_{\b,F}$ is defined in the following way.
\begin{align}
\label{def_bar_phi_{beta,F}}
\bar\f_{\b,F}(u)=
\left\{\begin{array}{lll}
\f_{\b}(\t_i),\,&\mbox{ if }F_{\b}(u)>F(\t_i),\,u\in[\t_i,\t_{i+1}),\\
\f_{\b}(s),\,&\mbox{ if }F_{\b}(u)=F(s),\mbox{ for some }s\in[\t_i,\t_{i+1}),\\
\f_{\b}(\t_{i+1}),\,&\mbox{ if }F_{\b}(u)<F(\t_i),\,u\in[\t_i,\t_{i+1}).
\end{array}
\right.
\end{align}
Similar to the proof of Theorem 4.1 in \cite{kim_piet:17:AOS}, we get that,
\begin{align}
\label{difference_f-fbar}
\|\f_{\hat\b_n}(u)-\bar\f_{\hat\b_n,\hat F_{n,\hat\b_n}}(u)\|\leq K |\hat F_{n,\hat\b_n}(u)-F_{\hat\b_n}(u)|,
\end{align}
for some constant $K>0$ not depending on $\beta$. By the definition of the MLE $\hat F_{n,\hat\b_n}$ as the slope of the greatest convex minorant of the corresponding cusum diagram, we can write:
\begin{align*}
&\hat\psi_{n}^{(\ee)}(\hat\b_n)\\
&=\int_{\hat F_{n,\hat\b_n}(t-\hat\b_n'x)\in[\ee,1-\ee]}\bigl\{x-\f_{\hat\b_n}(t-\hat\b_n'x)\bigr\}\bigl\{\d - \hat F_{n,\hat\b_n}(t-\hat\b_n'x)\bigr\}\,d\hat\P_n(t,x,\d)\\
&\qquad+\int_{\hat F_{n,\hat\b_n}(t-\hat\b_n'x)\in[\ee,1-\ee]}\bigl\{\f_{\hat\b_n}(t-\hat\b_n'x)-\bar\f_{\hat\b_n,\hat F_{n,\hat\b_n}}(t-\hat\b_n'x)\bigr\}\\
&\qquad\qquad\qquad\qquad\qquad\qquad\qquad\qquad\qquad\cdot\bigl\{\d - \hat F_{n,\hat\b_n}(t-\hat\b_n'x)\bigr\}\,d\hat \P_n(t,x,\d)\\
&= I + II,
\end{align*}
For the second term, we have:
\begin{align*}
&II =\int_{\hat F_{n,\hat\b_n}(t-\hat\b_n'x)\in[\ee,1-\ee]}\Bigl\{\f_{\hat\b_n}(t-\hat\b_n'x)-\bar\f_{\hat\b_n,\hat F_{n,\hat\b_n}}(t-\hat\b_n'x)\Bigr\}\\
	&\qquad\qquad\qquad\qquad\qquad\qquad\qquad\cdot \Bigl\{\d -\hat F_{n,\hat\b_n}(t-\hat\b_n'x)\Bigr\}\,d(\hat \P_n-\P_n)(t,x,\d)\\
	&\qquad+\int_{\hat F_{n,\hat\b_n}(t-\hat\b_n'x)\in[\ee,1-\ee]}\Bigl\{\f_{\hat\b_n}(t-\hat\b_n'x)-\bar\f_{\hat\b_n,\hat F_{n,\hat\b_n}}(t-\hat\b_n'x)\Bigr\}\\
	&\qquad\qquad\qquad\qquad\qquad\qquad\qquad\cdot\Bigl\{\d - \hat F_{n,\hat\b_n}(t-\hat\b_n'x)\Bigr\}\,d \P_n(t,x,\d)
	\\
	&=  II_a + II_b
\end{align*}
It is shown in the proof of Theorem 4.1 in \cite{kim_piet:17:AOS} that
$$ II_b =  o_p(n^{-1/2} + (\hat\b_n - \b_0)),$$
and therefore
$$ II_b =  o_{P_M}(n^{-1/2}+ (\hat\b_n - \b_0)) \text{ in probability}.$$
 Using similar arguments as in in the proof of Theorem 4.1 in \cite{kim_piet:17:AOS} we can also show that
 $$ II_a =  o_{P_M}(n^{-1/2}) \text{ in probability.} $$
 Hence, we get:
 \begin{align*}
 &\hat\psi_{n}^{(\ee)}(\hat\b_n)\\
 &=\int_{\hat F_{n,\hat\b_n}(t-\hat\b_n'x)\in[\ee,1-\ee]}\bigl\{x-\f_{\hat\b_n}(t-\hat\b_n'x)\bigr\}\bigl\{\d - \hat F_{n,\hat\b_n}(t-\hat\b_n'x)\bigr\}\,d\hat\P_n(t,x,\d)\\
 & \qquad +o_{P_M}(n^{-1/2}+ (\hat\b_n - \b_0)),
 \end{align*}
 in probability. We now write,
  \begin{align*}
 &\int_{\hat F_{n,\hat\b_n}(t-\hat\b_n'x)\in[\ee,1-\ee]}\bigl\{x-\f_{\hat\b_n}(t-\hat\b_n'x)\bigr\}\bigl\{\d -\hat F_{n,\hat\b_n}(t-\hat\b_n'x)\bigr\}\,d\hat\P_n(t,x,\d)\\
  & =  \int_{\hat F_{n,\hat\b_n}(t-\hat\b_n'x)\in[\ee,1-\ee]}\bigl\{x-\f_{\hat\b_n}(t-\hat\b_n'x)\bigr\}\\
  &\qquad\qquad\qquad\qquad\qquad\qquad\qquad\cdot\bigl\{\d - \hat F_{n,\hat\b_n}(t-\hat\b_n'x)\bigr\}\,d(\hat\P_n-\P_n)(t,x,\d)\\
  &\quad +  \int_{\hat F_{n,\hat\b_n}(t-\hat\b_n'x)\in[\ee,1-\ee]}\bigl\{x-\f_{\hat\b_n}(t-\hat\b_n'x)\bigr\}\bigl\{\d - \hat F_{n,\hat\b_n}(t-\hat\b_n'x)\bigr\}\,d\P_n(t,x,\d)
  \end{align*}
 It follows from the proof of Theorem 4.1 in \cite{kim_piet:17:AOS} that there exists a random variable $R_n$ of order $o_p(n^{-1/2} + \hat\b_n-\b_0)$ (and hence of order $o_{P_M}(n^{-1/2} + \hat\b_n-\b_0)$ in probability) such that,
  \begin{align}
  \label{term0}
 & \int_{\hat F_{n,\hat\b_n}(t-\hat\b_n'x)\in[\ee,1-\ee]}\bigl\{x-\f_{\hat\b_n}(t-\hat\b_n'x)\bigr\}\bigl\{\d -\hat F_{n,\hat\b_n}(t-\hat\b_n'x)\bigr\}\,d\P_n(t,x,\d)\nonumber\\
 &=\int_{ F_{0}(t-\b_0'x)\in[\ee,1-\ee]}\Bigl\{x-\f_0(t-\b_0'x)\Bigr\}\Bigl\{\d - F_0(t-\b_0'x)\Bigr\}\,d(\P_n-P)(t,x,\d)\nonumber\\
 &\qquad+ \psi_{1,\ee}'(\b_0)(\hat\b_n - \b_0) +R_n.
\end{align}
where $\f_0 \equiv \f_{\b_0}$. Therefore, (\ref{score_bootstrap}) follows if we can show that,
  \begin{align}
  \label{term1}
  &\int_{\hat F_{n,\hat\b_n}(t-\hat\b_n'x)\in[\ee,1-\ee]}\bigl\{x-\f_{\hat\b_n}(t-\hat\b_n'x)\bigr\}\bigl\{\d - \hat F_{n,\hat\b_n}(t-\hat\b_n'x)\bigr\}\,d(\hat\P_n-\P_n)(t,x,\d)\nonumber\\
  &\quad=\int_{ F_{0}(t-\b_0'x)\in[\ee,1-\ee]}\Bigl\{x-\f_0(t-\b_0'x)\Bigr\}\Bigl\{\d -F_0(t-\b_0'x)\Bigr\}\,d(\hat\P_n-\P_n)(t,x,\d)\nonumber
  \\ &\qquad\quad+ o_{P_M}(n^{-1/2} + (\hat\b_n - \b_0)).
  \end{align}
Equality (\ref{term1}) follows by similar arguments used in the proof of (\ref{term0}) based on asymptotic equicontinuity using the closeness of $\hat F_{n,\b}$ to $F_{\b}$ and using entropy results for the functions $u\mapsto \hat F_{n,\b}(u)$ and the simpler parametric functions $u\mapsto F_{\b}(u)$ and $u\mapsto \f_{\b}(u)$, parametrized by the finite dimensional parameter $\b$.

 \section*{Acknowledgements}
 \label{section:acknowledgements}
 The research of the second author was supported by the Research Foundation Flanders (FWO) [grant number 11W7315N]. Support from the IAP Research Network P7/06 of the Belgian State (Belgian Science Policy) is gratefully acknowledged. For the simulations we used the infrastructure of the VSC - Flemish Supercomputer Center, funded by the Hercules Foundation and the Flemish Government - department EWI.
 
\bibliographystyle{imsart-nameyear}

\end{document}